\documentclass[11pt]{article}
\usepackage{sustyle}

\usepackage{fullpage}
\usepackage{amsmath}
\usepackage{graphicx}
\usepackage{enumerate}
\usepackage{diagbox}
\usepackage{wrapfig}
\usepackage{tikz}
\usetikzlibrary{positioning}
\usepackage[round]{natbib}
\usepackage{url}  
\usepackage[svgnames]{xcolor}
\usepackage{tcolorbox}
\usepackage{multirow}
\usepackage{microtype}
\usepackage{amsmath,amsfonts,bbm,bbding,bm}
\usepackage{amsthm}
\usepackage{etoolbox}
\usepackage{subfigure}
\usepackage{xcolor}
\usepackage{booktabs} 
\usepackage{nicefrac}   
\usepackage{color} 
\usepackage{enumitem}
\usepackage{dsfont}
\usepackage[linktocpage]{hyperref}
\usepackage{lmodern}
\usepackage[T1]{fontenc}
\usepackage{tikz}
\usepackage{authblk}
\author[1]{Kaizhao Liu} 
\author[2]{Qi Long}
\author[3]{Zhekun Shi}
\author[2]{Weijie J.~Su}
\author[2]{Jiancong Xiao}
\affil[1]{Massachusetts Institute of Technology}
\affil[2]{University of Pennsylvania}
\affil[3]{Princeton University}

\definecolor{blue}{rgb}{0, 0, 1}

  \hypersetup{
  colorlinks,
  citecolor=gray,
  linkcolor=blue,
  urlcolor=black}
\begin{document}

  \title{Statistical Impossibility and Possibility of Aligning LLMs with Human Preferences: From Condorcet Paradox to Nash Equilibrium}
    \date{}
  \maketitle

\begin{abstract}

Aligning large language models (LLMs) with diverse human preferences is critical for ensuring fairness and informed outcomes when deploying these models for decision-making. In this paper, we seek to uncover fundamental statistical limits concerning aligning LLMs with human preferences, with a focus on the probabilistic representation of human preferences and the preservation of diverse preferences in aligned LLMs. We first show that human preferences can be represented by a reward model if and only if the preference among LLM-generated responses is free of any Condorcet cycle. Moreover, we prove that Condorcet cycles exist with probability converging to one exponentially fast under a general probabilistic preference model called the Luce model, thereby demonstrating the \textit{impossibility} of fully aligning human preferences using reward-based approaches such as reinforcement learning from human feedback. Next, we explore the conditions under which LLMs would employ mixed strategies---meaning they do not collapse to a single response---when aligned in the limit using a non-reward-based approach, such as Nash learning from human feedback. We identify a necessary and sufficient condition for mixed strategies: the \textit{absence} of a response that is preferred over all others by a majority. As a blessing, we prove that this condition holds with high probability under the Luce model, thereby highlighting the statistical \textit{possibility} of preserving minority preferences without explicit regularization in aligning LLMs.

\end{abstract}

\begingroup
\renewcommand\thefootnote{}\footnotemark\footnotetext{This paper has been accepted for publication in the \textit{Annals of Statistics}.}
\endgroup
\begingroup
\renewcommand\thefootnote{}\footnotemark\footnotetext{Authors are listed in alphabetical order. Emails: \texttt{mrzt@mit.edu, qlong@upenn.edu, zs3417@princeton.edu, suw@wharton.upenn.edu, and jcxiao@upenn.edu.}}
\endgroup
\tableofcontents

\section{Introduction}\label{sec:intro}
Large language models (LLMs), such as ChatGPT-4o and Claude-3.7 Sonnet, have exhibited remarkable capabilities across diverse domains, including code generation, data analysis, elementary mathematics, and reasoning \citep{hurst2024gpt,anthropic2024claude,chowdhery2023palm,touvron2023llama,ji2025overview}. These models are increasingly integrated into decision-making processes previously considered unlikely to be automated in the near future \citep{bubeck2023sparks,eloundou2024gpts}.

A central pillar enabling their popularity and effectiveness lies in alignment: LLMs learn to interact with human users and accommodate diverse opinions and values by aligning outputs with human preferences through reinforcement learning from human feedback (RLHF) \citep{Ouyang2022,casper2023open,dong2024rlhf}. In essence, RLHF begins by constructing a reward model trained on preference data from human labelers using the Bradley-Terry (BT) \citep{bradley1952rankanalysis} model. A high scalar score assigned by the reward model to a response generated by an LLM suggests that this response is relatively preferred by human labelers, thereby presumably representing the broader human population. The LLM is then fine-tuned with respect to this reward model to generate responses that reflect high preferences in distribution.

From a statistical perspective, a critical question is whether reward models can sufficiently capture all possible human preferences structures. Without understanding the potential statistical limits of using reward models for preference representation, one risks failing to capture the diversity of human preferences, potentially introducing biases and unwanted consequences in high-stakes decision-making \citep{casper2023open,azar2024general}. It is plausible that statistical limitations may arise because reward models in RLHF implicitly assume homogeneous preferences across individuals. Previous literature has informally discussed the potential shortcomings of such an approach \citep{calandriello2024human,zhang2024general}. If fundamental limitations exist, it is crucial to consider alternatives to reward-based alignment techniques and explore their statistical properties.

\subsection{A Statistical Impossibility}
In this paper, we first identify a fundamental impossibility result regarding the representational capacity of reward models for human preferences. We observe that a human preferences structure (among multiple responses) cannot be represented by a reward model if and only if a \textit{Condorcet cycle} exists. A Condorcet cycle of length \(r \geqslant 3\) consists of \(r\) responses \(\{y_i\}_{i=1}^r\) such that each \(y_i\) is preferred over \(y_{i+1}\) by a majority of labelers, with the cyclic convention \(y_{r+1} = y_1\) (see Example \ref{ex:condorcet} in Section \ref{sec:benefit}). While it is evident that a reward-induced preference structure cannot admit Condorcet cycles---because a response with the highest score (assuming no ties) cannot be beaten by any other response---we show that the existence of a Condorcet cycle is also necessary for any human preferences structure that cannot be represented by a reward model.

We next demonstrate that Condorcet cycles emerge under mild conditions, establishing that reward models are fundamentally incapable of fully capturing human preferences for LLM alignment. We consider a common setting called \emph{linear preference ranking} that each individual's preference on $n$ responses is a strict ranking, and adopt a probabilistic model called \emph{Luce model}\footnote{To avoid potential confusion, we emphasize that the Luce model considered here is a probabilistic framework and should not be confused with the Plackett–Luce model \citep{luce2012individual}, a K-ary extension of the BT model.} for the individuals' preferences. The Luce model assigns each responses a positive latent utility parameter and generating rankings probabilistically, where the probability of selecting an responses at each stage is proportional to its utility. When all utility parameters are equal, the Luce model reduces to a uniform distribution over all rankings, corresponding to the Impartial Culture condition \citep{guilbaud1952theories,kelly1974voting}, one of the most widely used assumptions in social choice theory. Consequently, our assumption strictly generalizes this classical setting by allowing structured yet principled deviations from uniform randomness. Under this model, we prove that the probability of observing at least one Condorcet cycle is bounded below by $1 - e^{-\mathrm{poly}(n)}$ when the number of labelers $m \geqslant 3$, where $\mathrm{poly}(n)$ denotes a (positive) polynomial in the number $n$ of responses. This result complements the earlier work of \citet{blin1973intransitive} in social choice theory, which analyzed the probability that a Condorcet cycle exists when the number of voters (labelers) $m \to \infty$.

\subsection{A Statistical Possibility}
Given the impossibility of fully capturing human preferences using reward models, a natural alternative is to directly model the pairwise probability that one response is preferred over another. One promising alternative is Nash learning from human feedback (NLHF) \citep{munos2024nash}, recently introduced by Google DeepMind. NLHF employs a pairwise preference modeling strategy by using a neural network to estimate the probability of one response being preferred over another. The alignment process is then carried out through a two-player game, where two LLMs are trained to generate responses, each aiming to maximize the probability that its own response is preferred over the other's. As such, any optimal solution of NLHF must be a Nash equilibrium of this two-player game.

Although the pairwise preference modeling in NLHF can, by design, represent \emph{all} human preferences, a crucial statistical question arises: does an NLHF-aligned LLM \emph{diversify} its outputs, or, put differently, preserve minority preferences, rather than ``collapse'' to a single response for a prompt? This is important because generating only a single response (even if widely preferred) can suppress minority preferences, leading to biases and fairness concerns, particularly when prompts solicit opinions rather than objective facts. Empirically, RLHF-aligned LLMs often exhibit substantial bias toward certain preferred responses \citep{santurkar2023whose,chakraborty2024maxmin}. Indeed, \citet{xiao2024algorithmic} demonstrated that RLHF, unless regularized, can induce ``preference collapse'' due to its reward-based nature, where the LLM tends to generate the response with the highest reward with a probability higher than expected, even when the reward model captures human preferences.

To address this, we derive a necessary and sufficient condition under which any optimal solution of NLHF is a \emph{mixed strategy}---that is, the NLHF-aligned LLM generates at least two\footnote{As we will show in Section \ref{sec:align}, when the NLHF solution has a mixed strategy, it must have at least three distinct responses with positive probability under a no-tie assumption.} distinct responses with positive probability, thereby avoiding collapse to a single response. 
We prove that a mixed strategy emerges if and only if there is no ``winning response,'' which is defined as a response that a majority of labelers prefer over every other response, often referred to as the winner in the literature of social choice theory \citep{gehrlein2006condorcet}. The necessity of this condition is intuitive: if a Condorcet winning response exists, then a best response in NLHF trivially selects that single response. Establishing the sufficiency, however, requires novel proof techniques.

Under the same probabilistic ranking model we utilize above, we show that the probability of a Condorcet winning response is small, thereby illustrating the provable statistical possibility of NLHF in preserving minority preferences. Formally, we prove that this probability is of the order 
\[
n^{-\frac{m}{\lceil m/2\rceil} + 1},
\]
up to logarithmic factors, when the number $m$ of labelers is greater than or equal to 3, under the classical impartial culture assumption and for any Luce model with uniformly bounded weights. Note that, above, $\frac{m}{\lceil m/2\rceil} = 2$ when $m$ is even and $\frac{m}{\lceil m/2\rceil} = \frac{2m}{m+1}$ when $m$ is odd. As $n \to \infty$, this probability converges to zero. These results provide new insight into the nature of LLM alignment. Under the impartial culture assumption, it is natural to expect that the existence of a Condorcet winning response is unlikely, and our results provide a rigorous theorem formalizing this intuition. However, under the Luce model, our analysis yields new and surprising conclusions. Consider a setting in which one preference ranking is assigned a relatively large utility weight, so that individuals select this ranking with high probability, while the remaining rankings have smaller but strictly positive weights. One might reasonably expect the former to be the Condorcet winning response. Nevertheless, our results show that even in this scenario, the probability that a Condorcet winning response exists still converges to zero.

\subsubsection*{Summary of Statistical Limits} Our results of statistical limits discussed above can be summarized in the following tree diagram (Figure \ref{fig:preference_tree_intro}), which illustrates the statistical limits of aligning LLMs with human preferences. Whether human preferences can be captured by reward models depends on the presence of Condorcet cycles. In scenarios where the coefficient of the KL term tends to zero and Condorcet cycles exist (with high probability), an NLHF-aligned LLM would output a mixed strategy if no Condorcet winning response is present, and would otherwise output the winning response. In contrast, under the same regime, an LLM aligned using standard RLHF with a reward model, if fully optimized in fine-tuning, is prone to collapsing to a single response, irrespective of whether a Condorcet winning response exists or not. As we have demonstrated that the Condorcet winning response does not exist with high probability under the probabilistic ranking model, the practical implication is that RLHF-aligned LLMs tend to collapse, whereas NLHF is more likely to generate diverse responses.

\begin{figure}[!ht]
\centering
\scalebox{0.48}{
\begin{tikzpicture}[
    >=latex,
    node distance=2.4cm and 3.8cm,
    every node/.style={font=\large},
    base/.style={
        draw=black!65,
        rounded corners=3pt,
        thick,
        align=center,
        fill=white,
        inner xsep=10pt,
        inner ysep=8pt,
        minimum height=1.0cm
    },
    mainbox/.style={
        base,
        fill=gray!6
    },
    resultbox/.style={
        base,
        text width=8.8cm
    }
]
\node[mainbox] (root) {\LARGE Human preferences};
\node[mainbox, above right=3.6cm and 5.6cm of root] (noreward) {\LARGE No reward model};
\node[mainbox, below right=3.6cm and 5.8cm of root] (reward) {\LARGE Reward model};
\node[resultbox, right=8.0cm of noreward] (mixed) {\LARGE
    \textcolor{red}{Converge to a mixed strategy}\\[4pt]
    \textcolor{blue}{Collapse to one response}
};
\node[resultbox, right=8.0cm of noreward, yshift=-4.5cm] (winner1) {\LARGE
    \textcolor{red}{Collapse to the winning response}\\[4pt]
    \textcolor{blue}{Collapse to one response}
};
\node[resultbox, right=8.4cm of reward] (winner2) {\LARGE
    Collapse to the winning response
};

\draw[->, thick] (root) -- (noreward)
    node[pos=0.45, above, sloped, fill=white, inner sep=1pt, font=\normalsize]
    {\LARGE Condorcet cycles exist (w.h.p.)};
\draw[->, thick] (root) -- (reward)
    node[pos=0.45, below, sloped, fill=white, inner sep=1pt, font=\normalsize]
    {\LARGE No Condorcet cycles};
\draw[->, thick] (noreward) -- (mixed)
    node[pos=0.52, above, sloped, fill=white, inner sep=1pt, font=\normalsize]
    {\LARGE No winning response (w.h.p.)};
\draw[->, thick] (noreward) -- (winner1)
    node[pos=0.48, above, sloped, fill=white, inner sep=1pt, font=\normalsize]
    {\LARGE Winning response exists};
\draw[->, thick] (reward) -- (winner1)
    node[pos=0.52, above, sloped, fill=white, inner sep=1pt, font=\normalsize]
    {\LARGE Others};
\draw[->, thick] (reward) -- (winner2)
    node[pos=0.52, below, sloped, fill=white, inner sep=1pt, font=\normalsize]
    {\LARGE BT model};
\end{tikzpicture}}
\caption{Illustration of statistical impossibility and possibility of aligning LLMs with human preferences. Reward models are impossible to capture human preferences with Condorcet cycles, while non-reward-based NLHF can maintain minority preferences when there is no Condorcet winning response. \textcolor{red}{Red text} indicates properties of \textcolor{red}{NLHF}, while \textcolor{blue}{blue text} indicates properties of \textcolor{blue}{RLHF}. Under the BT model, both RLHF- and NLHF-aligned LLMs have the same property, thus we use black text. The abbreviation ``w.h.p.'' stands for ``with high probability.'' In general, when collapsing, a RLHF-aligned LLM might not necessarily collapse to generate the Condorcet winning response (see discussion in Example~\ref{exam:rlhf_solution}). ``Other'' reward models are discussed in Section~\ref{sec:benefit}.}
\label{fig:preference_tree_intro}
\end{figure}

\subsection{Technical Contributions}
Beyond the context of LLM alignment, our work also provides a strong probabilistic contribution to social choice theory from a theoretical perspective. To discuss the technical contributions, we begin from a reverse direction, starting with the probabilities of the Condorcet winning response, known in classical studies as the Condorcet winner, exists. In 1968, \citet{garman1968paradox} conjectured, based on numerical evidence, that the probability of a Condorcet winner exists under the impartial culture condition tends to zero as the number of alternatives increases. Later, \citet{may1971some} proved that this probability converges to zero at a rate of $O(n^{-\frac12})$ for any fixed $m \geqslant 3$. However, these results have two main limitations. First, it is unclear whether the bound is tight, as it lacks a matching lower bound. As a result, the following problem have remained open for fifty years: \emph{What is the tight bound on the probability that a Condorcet winner exists under the uniform distribution?}
This problem was resolved after 50 years by \citet{sauermann2022probability}. Under the impartial culture assumption, they provided the tight rate $\Theta\!\left(n^{-\frac{m}{\lceil m/2\rceil} + 1}\right)$ for odd \(m\). 

The second is that the analysis is restricted to impartial culture, that is, the uniform distribution, which is a strong and often unrealistic assumption in practice. However, the proof techniques in \citet{sauermann2022probability} rely heavily on the assumption that all rankings are equally likely, and it remains unclear how to obtain a bound when the underlying distribution deviates from uniformity. Therefore, the following problem has remained open until this work:
\begin{center}
\emph{How can one obtain a (tight) bound on the probability that a Condorcet winner exists under a non-uniform distribution?}
\end{center}
In the existing literature, methods for deriving such bounds are largely unexplored, let alone establishing tight ones. Our results (Theorem \ref{thm:no Condorcet winning} and \ref{thm:no Condorcet winning lower bound}) show that, up to logarithmic factors, the rate $n^{-\frac{m}{\lceil m/2\rceil} + 1}$ continues to hold for any Luce model with uniformly bounded weights. These results show that the bound in \cite{may1971some} is tight for \(m=3\), and they improve upon the bounds obtained in \cite{may1971some} for any \(m > 3\), since \(\frac{m}{\lceil m/2\rceil} > \frac{3}{2}\). Moreover, the exponent \(-\frac{m}{\lceil m/2\rceil}+1\) is sharp and cannot be further improved. 
A comparison of the bounds is provided in Figure~\ref{fig:condorcet_improvement}. Technically, our idea of proof is completely different from the existing literature, which enables us to tackle non-uniform settings.
\begin{figure}[ht]
    \centering
    \includegraphics[width=0.75\linewidth]{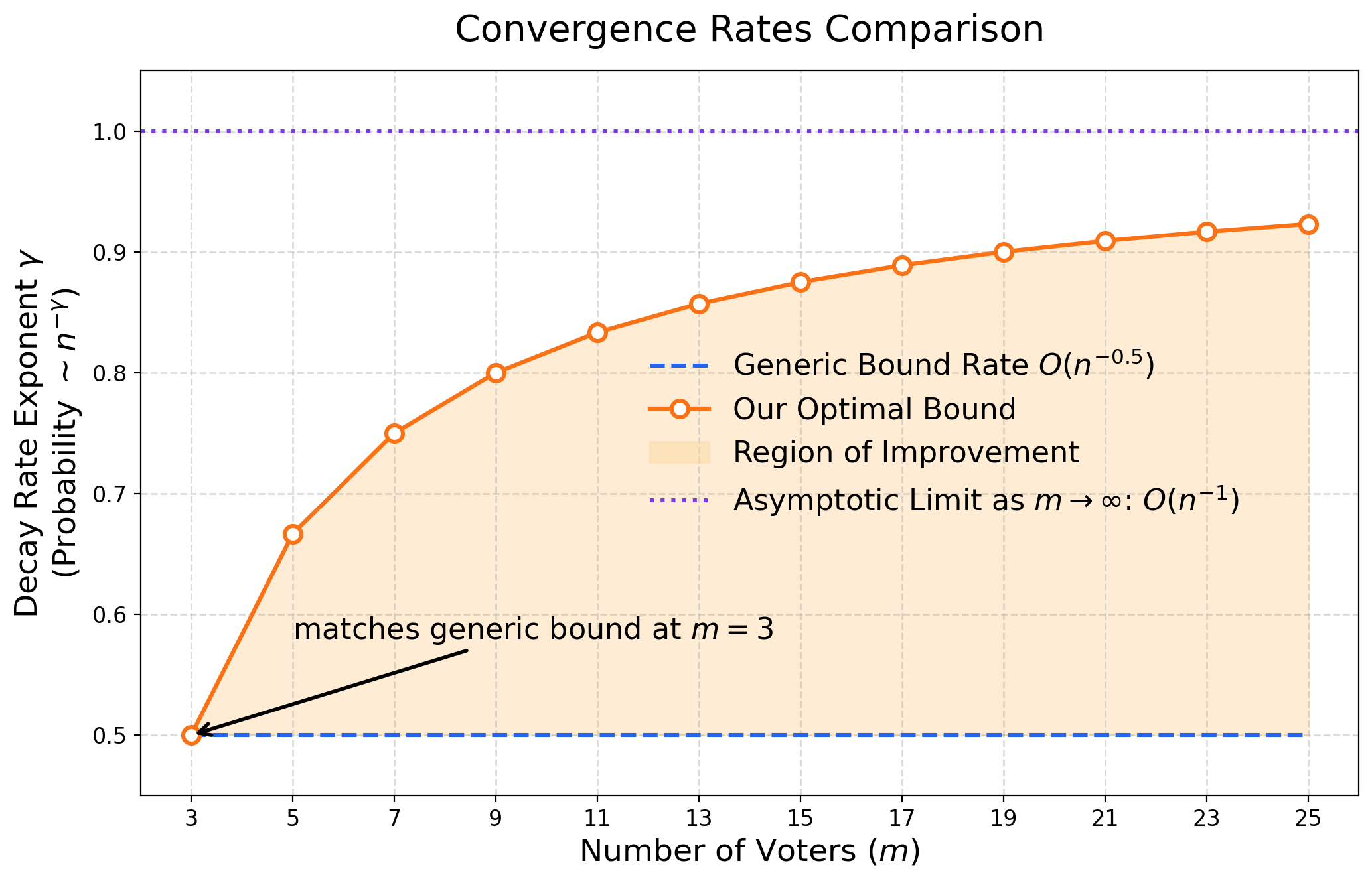}
\caption{Comparison of convergence rate exponents for the probability of a Condorcet winner, restricted to odd numbers of voters $m$. The $y$-axis represents the exponent $\gamma$ in the asymptotic decay rate $n^{-\gamma}$ (ignoring logarithmic factors) as the number of candidates $n \to \infty$. A higher value of $\gamma$ indicates a faster decay and a strictly tighter bound. The blue dashed line indicates the previously known generic bound of $O(n^{-0.5})$. The orange line with markers shows the exponent derived from our optimal bound, $\tilde{\Theta}(n^{1-m/\lceil m/2 \rceil})$, for both uniform and non-uniform settings. While matching the existing bound at $m=3$, our work demonstrates a monotonic improvement for all odd $m > 3$, asymptotically approaching the theoretical limit of $O(n^{-1})$.}
    \label{fig:condorcet_improvement}
\end{figure}

We then turn our attention to the probability that Condorcet cycles exist. To the best of our knowledge, this question has received comparatively less attention in the classical social choice literature. \citet{blin1973intransitive} studied this problem under a model where, for every pair of alternatives, the associated pairwise comparison probabilities are \(50\%\) to \(50\%\). This probability model is adopted primarily for mathematical tractability; however, it is unrealistic in practice, not only because it assumes equal preferences, but also because it imposes no linear-order restriction on individual preferences.
After that, questions similar to the above two problems for Condorcet cycles under linear preference rankings remain largely unexplored. Our work fills this gap by providing tight bounds under both uniform and non-uniform distribution on rankings, namely, under impartial culture and the Luce model.

\subsection{Related Work} 
Concerns have been raised about LLMs regarding fairness and potential biases, which have been extensively discussed in recent literature \citep{bender2021dangers}. These issues include gender bias \citep{kotek2023gender}, political bias \citep{feng2023pretraining}, and language bias \citep{luo2024perspectival}. One contributing factor to these biases arises from RLHF---\citet{xiao2024algorithmic} observed that conventional RLHF methods often converge to a single preference type, a phenomenon known as preference collapse. Additionally, a recent study \citep{li2024entropic} analyzed diversity in the supervised fine-tuning phase, further highlighting the challenges in maintaining preference diversity.

Several preference fine-tuning methods have been proposed for RLHF. The work of \citet{li2023remax} showed that proximal policy optimization (PPO)~\citep{schulman2017proximal} does not fully leverage the advantages of RLHF in aligning LLMs with human preferences. The work of \citet{tang2024understanding} examined the performance gap between online and offline algorithms for alignment tasks, while \citet{ye2024theoretical} introduced an online iterative RLHF algorithm that incorporates a general preference model. 

A growing body of work studies LLM alignment through the lens of social choice theory, seeking to understand when alignment objectives faithfully aggregate heterogeneous human preferences. Recent studies formalize axiomatic desiderata for alignment from human feedback \citep{siththaranjan2023distributional,ge2024axioms,golzdistortion}, analyze reward-based alignment methods under social-choice criteria \citep{xiao2025theoretical}, and connect alignment mechanisms to classical solution concepts such as maximal lotteries \citep{maura2025jackpot}. In subsequent work, we further investigate the Condorcet and Smith consistency of game-theoretic alignment approaches under noisy preference matrices \citep{shi2025fundamental}.

In social choice theory, the probabilities of the existence of Condorcet cycles and Condorcet winning responses have been studied for decades \citep{guilbaud1952theories, may1971some,kelly1974voting,bell1978happens, gehrlein2004probability, gehrlein2006condorcet} under various assumptions. These studies primarily focus on electoral settings, where the number of candidates (responses) $n$ is relatively small, while the number of voters (labelers) $m$ is large and is often assumed to tend to infinity in analysis.

The remainder of the paper is structured as follows. In Section \ref{sec:benefit}, we demonstrate the limitation of reward models for representing human preferences. Section \ref{sec:align} investigates under what conditions NLHF would preserve minority preferences. 
Section \ref{sec:upper bound proof} and Section \ref{sec:lower bound proof} provide proof sketches of our main technical contributions, Theorem \ref{thm:no Condorcet winning} and Theorem \ref{thm:no Condorcet winning lower bound} repsectively.
Finally, we conclude the paper with a discussion in Section \ref{sec:discuss}.

\section{Impossibility of Reward Models for Preference Representation}
\label{sec:benefit}  

In this section, we analyze the statistical limits of using reward models to capture human preferences. Specifically, we will demonstrate the following results:
\begin{enumerate}
    \item Human preferences can be modeled by a reward model if and only if there does not exist a Condorcet cycle in the preference structure. 
    \item Under mild assumptions, the probability of a Condorcet cycle emerging approaches one exponentially as the number of responses increases.
\end{enumerate}

\subsection{Preliminaries of RLHF} We briefly introduce the RLHF framework. Let $\pi(y\mid x)$ denote the LLM's distribution over responses $y\in\mathcal{Y}$ given a prompt $x\in\mathcal{X}$. Let $\pi_{\mathrm{ref}}$ denote a reference LLM, and let $\rho$ be a distribution over prompts. In practice, \( \rho \) can be either a fixed database of prompts \citep{rafailov2023direct} or adaptively updated based on prior responses \citep{xiong2023iterative}.

In RLHF, human preferences are typically modeled using the BTL model, which assumes the existence of a reward model $r(x,y)$ such that:
\begin{equation}
\label{eq:BT}
\cP(y \succ y'\mid x)=\frac{\exp(r(x,y))}{\exp(r(x,y))+\exp(r(x,y'))}.
\end{equation}
The standard RLHF objective \citep{Ouyang2022} is
\begin{equation}\label{eq:prlloss}
\max_\pi \mathbb{E}_{x\sim\rho}\left[\mathbb{E}_{y\sim\pi(\cdot\mid x)} r(x,y)-\tau \operatorname{KL}(\pi (\cdot\mid x)\| \pi_{\textnormal{ref}}(\cdot\mid x)) \right],
\end{equation}
where $\operatorname{KL}(\cdot\|\cdot)$ is the Kullback–Leibler divergence. For notational simplicity, we omit the dependence on the prompt \( x \) when it does not cause confusion in the following discussion. The formulation can also be interpreted as maximizing the reward subject to a constraint that the KL divergence between $\pi$ and $\pi_{\textnormal{ref}}$ remains small. In practice, $\tau$ is typically small; accordingly, in our statistical analysis we focus on the regime where $\tau \to 0$.

\subsection{Reward Model and Condorcet Cycle}
To facilitate our analysis, we begin with the following standard assumptions commonly used in the social science study of human preferences \citep{guilbaud1952theories, may1971some,bell1978happens, gehrlein2004probability,gehrlein2006condorcet}.

\begin{assumption}[Linear Preference Ranking]\label{ass:individual}
Given a prompt $x$ and its associated $n$ responses ${y_1,y_2,\ldots,y_n}$, an individual (labeler) expresses a preference in the form of a ranking of these $n$ responses. 
\end{assumption}

Assumption \ref{ass:individual} states that each individual has a rational preference, meaning they maintain a strict ranking over the \( n \) responses. Individuals cannot have cyclic preferences, such as \( y_1 \succ y_2 \succ y_3 \succ y_1 \), as this would make it impossible to align with a single individual, let alone a group. Given this assumption, individuals can express up to \( n! \) distinct preference rankings.  

The aggregate preference \( \mathcal{P}(y \succ y') \) is defined as the expected fraction of individuals who prefer \( y \) over \( y' \). Naturally, it follows that \( \mathcal{P}(y \succ y') + \mathcal{P}(y' \succ y) = 1 \). For any distinct responses \( y \) and \( y' \), we say that \( y \) is preferred over \( y' \) if \( \mathcal{P}(y \succ y') > \frac{1}{2} \). In parts of the analysis, we impose the following assumption:
\begin{assumption}[No-Tie]
\label{ass:strict_preference}
For any distinct responses \( y \) and \( y' \), we assume that \( \mathcal{P} (y \succ y') \neq 1/2 \).  
\end{assumption}  
This assumption is practical in real-world settings. First, if the number of labelers is odd, it automatically holds. Even in cases where a tie occurs, it can always be resolved through a more precise comparison without affecting the overall distribution or our theoretical analysis.

Having introduced our assumptions, we now turn our attention to the reward model. At its core, a reward model should reflect how labelers prefer certain responses, meaning that whenever the preference model satisfies $\cP(y \succ y') > \frac{1}{2}$, the  reward must satisfy $r(y)>r(y')$, leading to the following definition:

\begin{definition}[Reward-Consistent Preference]\label{def:what_is_reward}
A reward model $r: \mathcal{Y}\rightarrow \mathbb{R}$ captures a preference $\cP(y \succ y')$ if for any two distinct responses $y$ and $y'$, $\cP(y\succ y')>1/2$ implies $r(y)>r(y')$. 
\end{definition}

In this definition, to maintain generality, we do not assume a specific functional relationship between the preference model and the reward model. Rather than restricting the preference-reward relationship, we focus on the most general reward-based framework for modeling preferences, ensuring only that the underlying decision problem remains consistent. Notably, \( \mathcal{P}(y \succ y') \) is not required to be translation-invariant with respect to \( r(y) - r(y') \) as in the BT model. Our Definition~\ref{def:what_is_reward} imposes the weakest possible constraint on reward-based preferences, thereby encompassing the broadest class of such models. As a result, our impossibility result is particularly strong. The BT model \eqref{eq:BT} used in practice is a specific instance of a reward model within our proposed class.

From a decision-theoretic perspective, this definition establishes an equivalence between two decision-making paradigms: one induced by the utility function (reward model) \citep{tadelis2013game} and another defined by the preference relation \citep{ismail2023existence}. 

Given a reward model, for any three distinct responses $y_{i_1}$, $y_{i_2}$, and $y_{i_3}$, if $r(y_{i_1})>r(y_{i_2})$ and $r(y_{i_2})>r(y_{i_3})$, then $r(y_{i_1})>r(y_{i_3})$ by transitivity.
Therefore, any preference model induced by a reward model must be \emph{acyclic} (or \emph{transitive}): if $\cP(y_{i_1}\succ y_{i_2})>\frac{1}{2}$ and $\cP(y_{i_2}\succ y_{i_3})>\frac{1}{2}$, then $\cP(y_{i_1}\succ y_{i_3})>\frac{1}{2}$ must hold. However, many real-world preferences are cyclic, as demonstrated by Condorcet paradox \citep{gehrlein2006condorcet}.

\begin{example}[Condorcet paradox \citep{gehrlein2006condorcet}]\label{ex:condorcet}
   Suppose one-third of the population prefers $y_1\succ y_2\succ y_3$, one-third prefers $y_2\succ y_3\succ y_1$, and the remaining third prefers $y_3\succ y_1\succ y_2$. In this case we have $\cP(y_1\succ y_2)=\frac{2}{3}$, $\cP(y_2\succ y_3)=\frac{2}{3}$, and $\cP(y_3\succ y_1)=\frac{2}{3}$. Therefore this preference relationship is cyclic.

\end{example}
 In this case, cyclic human preferences cannot be captured by any reward model. Condorcet paradox demonstrates that cyclic preferences can arise from simple scenarios: just three labelers with different preferences can produce a cyclic group preference. Inspired by this observation, we introduce the following definition of Condorcet $r$-cycles:
\begin{definition}[Condorcet Cycle]
A sequence of responses $y_{i_1},\ldots,y_{i_r}$ forms a Condorcet $r$-cycle if $\cP(y_{i_p}\succ y_{i_{p+1}})>\frac{1}{2}$ for all $p=1,\ldots,r$, where $y_{i_{r+1}}:=y_{i_1}$ by convention.
\end{definition}

To characterize the capacity of reward models precisely, we show that transitivity is not only necessary but also sufficient for a preference model to be induced from a reward model, as detailed in Theorem \ref{thm:preferencetoreward}. 

\begin{restatable}[Necessary and Sufficient Conditions for Reward Modeling]{theorem}{preferencetoreward}\label{thm:preferencetoreward}
    Under Assumption~\ref{ass:strict_preference}, for any set of responses $\{y_1,\ldots,y_n\}$ with $n\geqslant 3$ and any preference $\cP(y\succ y')$ defined on this set, there exists a reward model $r(y)$ that captures the preference $\cP(y\succ y')$ if and only if there is no Condorcet cycle in the set of responses.
\end{restatable}

\begin{remark}

The non-existence of a reward model has previously been established through the lens of social welfare theory, where the reward model $r(y)$ can be interpreted as a social welfare function.
Arrow's impossibility theorem demonstrates that when there are at least three choices, no social welfare function can simultaneously satisfy three conditions: non-imposition, non-dictatorship, and independence of irrelevant alternatives \citep{wilson1972social}.
Here, in contrast, we require that the reward model represents the preference relationship, and we establish the equivalence between the non-existence of a reward model and cyclic preferences. 
\end{remark}

To conclude this subsection, we demonstrate that the non-existence of Condorcet cycles can be simplified to checking for Condorcet triangles, which are defined as Condorcet 3-cycles. The following proposition establishes that the existence of any Condorcet cycle implies the existence of a Condorcet triangle. 

\begin{restatable}{proposition}{cycleequivalenttotriangle}
\label{prop: cycle equivalent to triangle}
Let ${y_1,\ldots,y_n}$ be a set of responses with $n\geqslant 3$. If this set contains a Condorcet cycle, then it must also contain a Condorcet triangle.
\end{restatable}
Conversely, if a set of responses contains a Condorcet triangle, it automatically contains a Condorcet cycle, that is, the triangle itself. Therefore, calculating the probability that human preferences cannot be captured by a reward model reduces to estimating the probability of a Condorcet cycle existing, which further reduces to estimating the probability of a Condorcet triangle existing. We will discuss this probability analysis in the next subsection.

\subsection{Probability of Condorcet Cycle} 

Next, we analyze the probability of a Condorcet cycle occurring given \( m \) labelers and \( n \) responses, denoted as \( \mathbb{P}_{m,n}(\text{Condorcet cycle}) \). To get an initial understanding of this probability, we begin with the case where $n=3$.
\begin{proposition}\label{prop:n=3 transitivity}
Under Assumption~\ref{ass:individual},
let $\alpha_1$ represent the proportion of labelers who prefer $y_1 \succ y_2 \succ y_3$, 
$\alpha_2$ the proportion who prefer $y_2 \succ y_3 \succ y_1$, 
$\alpha_3$ the proportion who prefer $y_3 \succ y_1 \succ y_2$, 
$\alpha_4$ the proportion who prefer $y_2 \succ y_1 \succ y_3$, 
$\alpha_5$ the proportion who prefer $y_3 \succ y_2 \succ y_1$, and 
$\alpha_6$ the proportion who prefer $y_1 \succ y_3 \succ y_2$, where $\alpha_i \geqslant 0$ for $i = 1, \dots, 6$, 
and $\alpha_1 + \alpha_2 + \alpha_3 + \alpha_4 + \alpha_5 + \alpha_6 = 1$. Then, the overall preference relationship is cyclic if and only if the following conditions hold:
\[
\begin{cases}
    \alpha_1 + \alpha_2 + \alpha_4 > \frac{1}{2}, \\
    \alpha_2 + \alpha_3 + \alpha_5 > \frac{1}{2}, \\
    \alpha_3 + \alpha_1 + \alpha_6 > \frac{1}{2},
\end{cases}
\text{ or }
\begin{cases}
    \alpha_1 + \alpha_2 + \alpha_4 < \frac{1}{2}, \\
    \alpha_2 + \alpha_3 + \alpha_5 < \frac{1}{2}, \\
    \alpha_3 + \alpha_1 + \alpha_6 < \frac{1}{2}.
\end{cases}
    \]
\end{proposition}

Through a Monte Carlo simulation, we can see that approximately 6\% of the values of $\{\alpha_i\}_{i=1}^6$ in the simplex satisfy the condition in Proposition \ref{prop:n=3 transitivity}, resulting in cyclic preferences. This demonstrates that the probability of the existence of a Condorcet cycle is not small even in the case of \( n = 3 \). When $\alpha_1=\alpha_2=\alpha_3=\frac{1}{3}$, and $\alpha_4=\alpha_5=\alpha_6=0$, we recover the Condorcet paradox described earlier. 
Beyond these synthetic examples, cyclic preferences are also found in real world data \citep{kurrild2001empirical}. Next, we consider the case for arbitrary \( n \).

The study of this problem dates back to the work of \citet{becker1967value}, who observed through simulations that \( \mathbb{P}_{m,n}(\text{Condorcet cycle}) \) remains very large for large \( n \) and all \( m \). Later, as noted in the introduction, an early contribution by \citet{blin1973intransitive} examined this problem under a model in which, for every pair of alternatives, the corresponding pairwise comparison probabilities are fixed at \(50\%\)–\(50\%\). While this assumption facilitates mathematical analysis, it is not well aligned with practical preference modeling, as it both enforces equal pairwise probabilities and allows individual preferences that are not constrained to arise from linear orderings.

To address this limitation, we consider two probabilistic assumptions under the linear preference ranking framework specified in Assumption~\ref{ass:individual}. The first is the impartial culture condition in Assumption~\ref{ass:population}, under which each ranking is uniformly distributed. The second is the Luce model in Assumption~\ref{ass:luce}, which allows for non-uniform preference distributions.

\begin{assumption}[Impartial Culture Condition \citep{guilbaud1952theories}]\label{ass:population} Given Assumption~\ref{ass:individual}, each individual independently samples a linear preference ranking with equal probability $1/n!$.\end{assumption}

The term \emph{impartial culture} (IC) was first introduced in social choice theory and is one of the most widely used assumptions in the literature. Although the uniform distribution underlying IC is strong and can be unrealistic in practice, it provides a natural starting point for probabilistic analysis. From a mathematical perspective, the resulting analysis is already nontrivial, as demonstrated by the works discussed before in the social choice literature. Below, we present our first probabilistic bound on the existence of Condorcet cycles.

\begin{restatable}{theorem}{existcondorcet}\label{thm:exist condorcet}
    Suppose the population satisfies Assumption \ref{ass:population}. Let the number of responses be $n\geqslant 3$ and the number of labelers be $m\geqslant 3$. Then,
    \[
\begin{aligned}
    &\ \mathbb{P}_{m,n}(\textnormal{Condorcet cycle})\geqslant 1-c_1e^{-c_2n},
\end{aligned}
\]
where $c_1,c_2>0$ are universal constants.
\end{restatable}

However, the IC condition requires strict uniformity between responses, which is sometimes too restrictive. In this paper, we propose to study the following Luce model (see for example \citet{borga2025permuton}):
\begin{assumption}[Luce Model]\label{ass:luce}
    Given Assumption \ref{ass:individual}, each individual independently samples a linear preference ranking from $\operatorname{Luce}(\lambda_1,\dots,\lambda_n)$, where $\lambda_1,\dots,\lambda_n$ are some positive real weights. Mathematically, for any permutation $\{i_1,\ldots,i_n\}$ of $\{1,\ldots,n\}$, the probability of each linear preference is given by:
    \begin{align*}
        \mathbb{P}\left(y_{i_1}\succ\cdots\succ y_{i_n}\right)=\frac{\lambda_{i_1}}{\lambda_{i_1}+\cdots+\lambda_{i_n}}\times \frac{\lambda_{i_2}}{\lambda_{i_2}+\cdots+\lambda_{i_n}}\times\cdots\times\frac{\lambda_{i_{n-1}}}{\lambda_{i_n}+\lambda_{i_{n-1}}}\,.
    \end{align*}
\end{assumption}
This mechanism is indeed a sequential sampling scheme without replacement: at step $t$, choose item $i_t$ from the remaining set $R_t$ with probability:
\begin{align*}
\mathbb{P}\left(i_t=j\mid R_t\right)=\frac{\lambda_j}{\sum_{k\in R_t}\lambda_k}\,,
\end{align*}
remove $i_t$ from the set, and repeat until a complete ranking $(i_1,\ldots,i_n)$ is obtained. It is worthy to note that when $\lambda_1=\dots=\lambda_n$, the Luce model reduces to the impartial culture condition.

Under the Luce model, expoential decay still holds, whereas the constant might depend on the number of labelers:
\begin{restatable}{theorem}{luceexsitcondorcet}\label{thm:luce exist condorcet}
    Suppose the population satisfies Assumption \ref{ass:luce}. 
    Let the number of responses be $n\geqslant 3$ and the number of labelers be $m\geqslant 3$. 
    In addition, suppose the weights are bounded with $\lambda_i\in[A,B]$, where $B\geqslant A>0$. 
    Then,
    \[
\begin{aligned}
    &\ \mathbb{P}_{m,n}(\textnormal{Condorcet cycle})\geqslant 1-c_1e^{-c_2n},
\end{aligned}
\]
where $c_1,c_2>0$ are constants that depend on $m,A,B$.
\end{restatable}

Notably, the above two probability results also hold for $m$ even, where Assumption~\ref{ass:strict_preference} might fail.
As \( n \to \infty \), the probability of Condorcet cycles existing approaches one exponentially at a rate of \( 1 - e^{-\textnormal{poly}(n)} \). As a result, human preferences cannot be accurately captured by a reward model with high probability. Consequently, RLHF-trained models will fail to effectively represent human preferences.

This demonstrates the necessity of using preference models directly, even under the idealized assumption of rational labelers. Moreover, since people typically do not exhibit perfect rationality in practice, the case for direct preference modeling becomes even stronger. This analysis highlights why modeling preferences directly is more appropriate than using reward-based approaches.

\section{Possibility of Preserving Diverse Preferences}
\label{sec:align}

Having shown the impossibility of representing general human preferences using reward models, in this section, we turn to the analysis of NLHF, which uses pairwise preference and is thus non-reward-based. In particular, NLHF can represent any possible human preferences. We are interested in whether this enhanced preference representation would bring benefits in preserving diverse preferences. Specifically, we examine whether NLHF can avoid preference collapse---a phenomenon observed in RLHF \citep{xiao2024algorithmic}, where LLMs tend to generate a single response if fine-tuned sufficiently long in alignment. In NLHF, this question reduces to determining when the Nash equilibria consist of mixed strategies rather than pure strategies. We will show the following results:

\begin{enumerate}
\item A Nash equilibrium in pure strategies exists if and only if there exists a Condorcet winning response (see Definition \ref{def:Condorcet winning response}) in the human preferences structure.
\item Under mild assumptions, the probability of a Condorcet winning response existing approaches zero.
\end{enumerate}
Combining these two, the Nash equilibria of NLHF are mixed strategies with a probability approaches one under mild assumptions. As a result, NLHF achieves better alignment than RLHF, which fails to preserve diverse human preferences under such conditions. 

\subsection{Preliminaries of NLHF} We first briefly introduce preference models and NLHF. Given the preference between two responses, we define the preference between two policies $\pi$ and $\pi'$ conditioned on a prompt $x$ by 
$$
\cP\left(\pi\succ \pi'\mid x\right):=\mathbb{E}_{y\sim\pi(\cdot\mid x),y'\sim\pi'(\cdot\mid x)}\left[\cP\left(y\succ y'\mid x\right)\right]\,,
$$
and given a distribution $\rho$ over prompts, we define the preference between two policies $\pi$ and $\pi'$ by 
$$
\cP\left(\pi\succ \pi'\right):=\mathbb{E}_{x\sim\rho}\left[\cP\left(\pi\succ \pi'\mid x\right)\right]\,.
$$
NLHF is a framework for aligning LLMs with human preferences by formulating model training as a two-player zero-sum game. Unlike the reward-based methods, NLHF ensures that the learned policy is resistant to adversarial preference comparisons, leveraging the concept of Nash equilibrium from game theory.
For a given prompt $x$, the LLM's policy $\pi$ competes against an opposing policy $\pi'$ in a pairwise preference contest, where the objective is to find a policy that maximizes its worst-case preference score. Formally, NLHF solves the following min-max optimization problem:
\begin{equation} \label{eq:ne} \max_\pi \min_{\pi'} \mathcal{P}(\pi \succ \pi').\end{equation}
As in RLHF, the practical objective of NLHF also includes two KL divergence terms between $\pi$ (or $\pi'$) and $\pi_{\textnormal{ref}}$, with coefficients $\tau$ and $\tau'$, respectively. In our analysis, we focus on the regime where both coefficients tend to zero, and we retain the objective in Eq.~(\ref{eq:ne}).

\subsection{Nash Equilibria and Condorcet Winning Responses} We begin with the main question of this section: When are the Nash equilibria of NLHF mixed strategies rather than pure strategies? Based on our analysis in Section~\ref{sec:benefit}, it is natural to hypothesize that Nash equilibria are mixed strategies when human preferences contain Condorcet cycles. This hypothesis seems to be supported by the Condorcet paradox in the previous section. 
\begin{example}\label{eg:condorcet mixed}
    Let the human preferences over three responses be given by the matrix in Table \ref{table:condorcet mixed}, where $\alpha_{12},\alpha_{23},\alpha_{31}>0.5$. This preference structure contains a Condorcet cycle $y_1\succ y_2\succ y_3\succ y_1$. In this case, no pure strategy Nash equilibrium exists, and the Nash equilibrium must be mixed, with positive probability assigned to all three responses.
    \begin{table}[htbp]
    \caption{Payoff matrix with three responses $\{y_1,y_2,y_3\}$.}
    \label{table:condorcet mixed}
    \centering
    \scalebox{1.1}{
    \begin{tabular}{c | c  c  c }
          $\cP(y\succ y')$ &  
          $y'=y_1$ & 
          $y'=y_2$ &
          $y'=y_3$
         \\\hline
         $y=y_1$ &
         0.5 &
         $\alpha_{12}$ &
         $1-\alpha_{31}$
         \\
         $y=y_2$ &
         $1-\alpha_{12}$ &
         0.5 &
         $\alpha_{23}$
         \\
         $y=y_3$ &
         $\alpha_{31}$ &
         $1-\alpha_{23}$ &
         0.5
         \\ 
    \end{tabular}
    }
\end{table}

\end{example}
However, the following example demonstrates that a Nash equilibrium in pure strategies can exist even in the presence of a Condorcet cycle, disproving the initial hypothesis.
\begin{example}\label{eg:Condorcet winning response}
    Let the human preferences over four responses be given by the matrix in Table \ref{table:Condorcet winning response}. This preference structure contains a Condorcet cycle $y_1\succ y_2\succ y_3\succ y_1$. The pure strategy $(y_4,y_4)$ is a Nash equilibrium in pure strategies.
    \begin{table}[htbp]
    \caption{Payoff matrix with four responses $\{y_1,y_2,y_3,y_4\}$.}
    \label{table:Condorcet winning response}
    \centering
    \scalebox{1.1}{
    \begin{tabular}{c | c  c  c  c }
          $\cP(y\succ y')$ &  
          $y'=y_1$ & 
          $y'=y_2$ &
          $y'=y_3$ &
          $y'=y_4$
         \\\hline
         $y=y_1$ &
         0.5 &
         $0.51$ &
         $0.46$ &
         0.47
         \\
         $y=y_2$ &
         $0.49$ &
         0.5 &
         $0.51$ &
         0.48
         \\
         $y=y_3$ &
         $0.54$ &
         $0.49$ &
         $0.5$ &
         0.49
         \\
         $y=y_4$ &
         $0.53$ &
         $0.52$ &
         0.51 &
         0.5
         \\ 
    \end{tabular}
    }
\end{table}
\end{example}
In this case, \( (y_4, y_4) \) is a Nash equilibrium due to the unbeaten status of \( y_4 \), as \( \mathcal{P}(y_4 \succ y_i) > 1/2 \) for \( i = 1,2,3 \). This observation leads to our second hypothesis: Nash equilibria consist of mixed strategies when no such response exists. This is essentially the concept of a Condorcet winning response. We formalize this notion with the following definition:
\begin{definition}[Condorcet Winning Response]\label{def:Condorcet winning response}  
A response \( y^\star \) is called a Condorcet winning response if \( \mathcal{P}(y^\star \succ y) > 1/2 \) for all \( y \neq y^\star \).  
\end{definition}

\begin{remark}
The concept of a Condorcet winning response corresponds to the winner of an election in traditional social choice theory. In the literature, such a winner has been referred to by various names, including the outright winner \citep{may1971some} and the pairwise majority rule winner \citep{gehrlein2006condorcet}, but they all describe the same fundamental idea. In the context of LLMs, where the ranked alternatives are responses rather than candidates, we use the term Condorcet winning response to maintain clarity and avoid confusion.
\end{remark}

By the definition, there cannot exist two Condorcet winning responses for a
single human preferences structure, which leads to the following proposition.

\begin{restatable}{proposition}{zerooneondorcetwinningresponse}\label{prop: 0/1 Condorcet winning response}
    There is either 0 or 1 Condorcet winning response.
\end{restatable}
Building upon this, we characterize when NLHF solutions avoid pure strategies by establishing a necessary and sufficient condition based on Condorcet winning responses.

\begin{restatable}{theorem}{Condorcetwinningresponse}\label{thm:Condorcet winning response} Under Assumption~\ref{ass:strict_preference}, there exists a Nash equilibrium in pure strategies if and only if there exists a Condorcet winning response. 
    Moreover, when there exists a Condorcet winning response, the Nash equilibrium is unique and corresponds to this Condorcet winning response.
\end{restatable}

A similar result by \citet{duersch2012pure} showed that Nash equilibrium in pure strategies exists in a symmetric two-player zero-sum game if and only if it is not a generalized rock-paper-scissors matrix (gRPS). Their notion of gRPS is equivalent to games without any Condorcet winning response in our setting. 

In the context of NLHF, our Theorem~\ref{thm:Condorcet winning response} strengthens this result by further proving that when such an equilibrium exists, it is unique. This refinement is significant, as uniqueness ensures stability in equilibrium selection. In traditional social choice theory, a social choice function is said to be Condorcet consistent if it returns the Condorcet winning response whenever one exists. Thus, our Theorem~\ref{thm:Condorcet winning response} directly leads to the following proposition:
\begin{proposition}  
Under Assumption~\ref{ass:strict_preference}, the Nash equilibrium of NLHF, when considered as a function of human preferences, is Condorcet consistent.  
\end{proposition}

Therefore, calculating the probability of having only mixed-strategy Nash equilibria in NLHF reduces to estimating the probability of having no Condorcet winning response. We will discuss this probability analysis in the next subsection.

\subsection{Probability of No Condorcet Winning Response} 
Now, we analyze the probability of a Condorcet winning response occurring given \( m \) labelers and \( n \) responses. 
Previous studies of this probability are limited to the IC condition (Assumption \ref{ass:population}). In 1968, \citet{garman1968paradox} conjectured that the probability converges to 0 as \( n \to \infty \) for all \( m \geq 3 \). This conjecture was later resolved by \citet{may1971some}, who proved that the probability decreases at a rate of \( O(1/\sqrt{n}) \) for all \( m \geq 3 \) and at a rate of \( O(1/n) \) for \( m = \infty \). Recently, \citet{sauermann2022probability} obtained the tight rate $\Theta\left(n^{-\frac{m}{\lceil m/2\rceil} + 1}\right)$ for odd $m$ under the IC condition. 
However, the proof techniques in \citet{sauermann2022probability} highly relies on the equal probability between rankings, and it is unclear whether the established rate still holds true when the probability distribution deviates from strict equality.  

Surprisingly, our results reveal that the same rate  holds for any Luce model with bounded weights, up to logarithm factors. As a remark, the probability results in this section also hold for $m$ even, where Assumption~\ref{ass:strict_preference} might fail.

\begin{theorem}\label{thm:no Condorcet winning}
    Suppose the population satisfies Assumption \ref{ass:luce}. Let the number of responses be $n+1$, the number of labelers be $m\geqslant 3$, and define $l=\lceil m/2 \rceil$. 
    In addition, suppose the weights are bounded with $\lambda_i\in[A,B]$, where $B\geqslant A>0$.
    Then 
    \begin{align*}
        &\ \mathbb{P}_{m,n}(\textnormal{Condorcet winning response})\\
        \leqslant &\ (m+1)(n+1)n^{-\frac{m}{l}}\left(\log n\right)^{\frac{2m}{l}}\left(\log n+1\right)^{m} \left(\frac{2B}{A}\right)^{2m}+ \frac{m(m-1)(n+1)B^2}{A^2n^2} \\
        &\ \quad \quad + e^{\frac{B}{A}m}(n+1)^{(3+\frac{B}{A})m+1}\cdot\left(\exp\left(-\frac{3(\log n)^2}{16}\right)+\exp\left(-\left(\frac{1}{2}\right)^l(\log n)^2\right)\right)\,.
    \end{align*}
\end{theorem}

We discuss the challenges in bounding this probability and outline the main idea of the proof in Section~\ref{sec:upper bound proof}. In practical language modeling scenarios, the number of labelers, i.e., $m$, is finite, while the number $n$ of generated responses can be infinite. Therefore, we analyze the probability in the regime where \( m \) is fixed and \( n \to \infty \). Under this regime and when $A,B$ are fixed constants, the third term decays super-polynomially with \( n \), while both the first and second terms decay polynomially with \( n \). Moreover, the first term decays more slowly than the second term and satisfies
\[
\text{the first term}= \tilde{O}\left(n^{-\frac{m}{\lceil m/2\rceil}+1}\right)\,,
\]
which dominates the other two terms, matching the rate established under IC condition. When $A=\min_i\lambda_i$ and $B=\max_i\lambda_i$ is allowed to be dependent on $n$, we can still apply Theorem \ref{thm:no Condorcet winning} and obtain a different rate.

\begin{theorem}\label{thm:no Condorcet winning lower bound}
    Suppose the population satisfies Assumption \ref{ass:luce}. Let the number of responses be $n+1$, the number of labelers be $m\geqslant 3$, and define $l=\lceil m/2\rceil$. 
    In addition, suppose the weights are bounded with $\lambda_i\in[A,B]$, where $B\geqslant A>0$.
    Then for some universal constant $c$, 
    \begin{align*}
        \mathbb{P}_{m,n}(\textnormal{Condorcet winning response})
        \geqslant c\cdot\left(\frac{A}{B}\right)^m\cdot\frac{1}{2^{m+\frac{m}{l}}\binom{m}{l}^{\frac{m}{l}}}\cdot\frac{n+1}{n^{\frac{m}{l}}}\,.
    \end{align*}
\end{theorem}

In practical language modeling scenarios where \( m \) is fixed and \( n \to \infty \), our main results simplify to the following corollary:
\begin{corollary}\label{corollary:Condorcet winning}
    Suppose the population satisfies Assumption \ref{ass:luce}. Let the number of responses be $n+1$ and the number of labelers be $m\geqslant 3$, and define $l=\lceil m/2\rceil$. 
    In addition, suppose the weights are bounded with $\lambda_i\in[A,B]$, where $B\geqslant A>0$. Then
    \[\mathbb{P}_{m,n}(\textnormal{Condorcet winning response})= \Tilde{\Theta}\left(n^{1-\frac{m}{l}}\right)\to 0 \quad \text{as } n \to \infty.\]
    Furthermore, the probability that no Condorcet-winning response exists approaches one as $n\to\infty$. In this case, the Nash equilibria of NLHF are mixed strategies.
\end{corollary}

We conduct simulations to confirm this rate. Specifically, we consider a Luce model with weights \(\lambda_i=\exp\left(\lambda_p\,\mathds{1}\{i=1\}\right)\) for $i\in\{1,\ldots,n\}$. As shown in Figure~\ref{fig:simulation-probability-short}, the empirical probabilities match our theoretical rate. We provide more simulation details and results in Appendix.

\begin{figure}[htbp]
    \centering
    \includegraphics[width=0.96\linewidth]{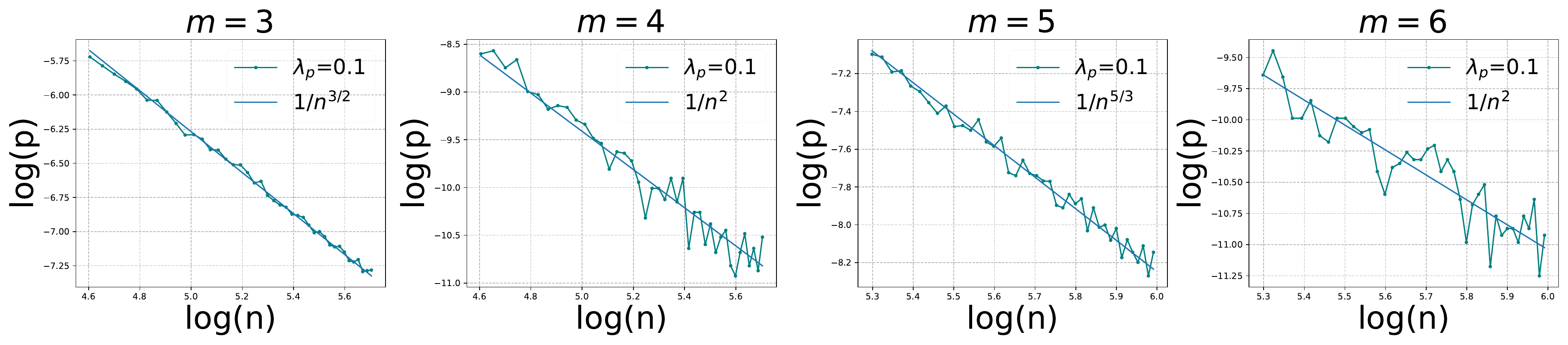}
    \includegraphics[width=0.96\linewidth]{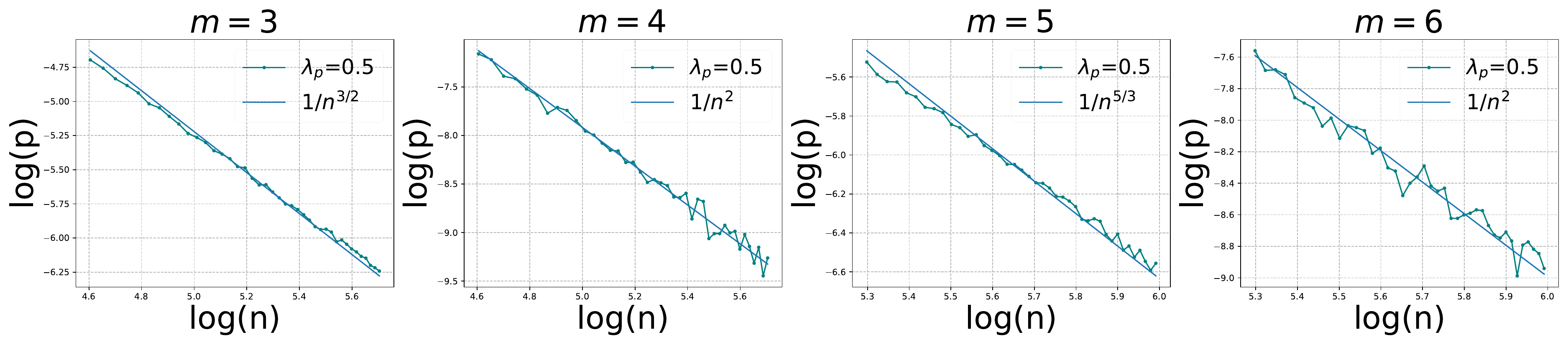}
    \caption{Simulation results for the probability that $y_1$ is the Condorcet winning response under the Luce model with varying $\lambda_p$, $m$, and $n$. The empirical estimates align closely with our theoretical rate $\tilde{\Theta}(n^{-\frac{m}{l}})$.}
    \label{fig:simulation-probability-short}
\end{figure}

Finally, we investigate the setting in which some individuals are allowed to have arbitrary preferences, including the possibility that their preferences are not chosen independently. We show that as long as the proportion of individuals that follow the Luce model is larger than a half, the probability that there exists a Condorcet winning response still tends to zero.
\begin{restatable}{corollary}{determinsticpreferencesluce}\label{coro:determinstic-preferences-luce}
     Let the number of responses be $n+1$ and the number of labelers be $m\geqslant 3$, and define $l=\lceil m/2\rceil$. For $0\leqslant k<m-l$, suppose that $m-k$ individuals follows Assumption \ref{ass:luce} with bounded weights $\lambda_i\in[A,B]$, where $B\geqslant A>0$, and the remaining $k$ individuals have arbitrary linear preferences (satisfying Assumption~\ref{ass:individual}), independent of the first $m-k$ individuals. Then, 
    \begin{align*}
        &\ \mathbb{P}_{m,n}\left(\text{Condorcet winning response}\right)\\
        &\ \leqslant (m-k+1)n^{-\frac{m-k}{l}}(n+1)\left(\log n\right)^{\frac{2m-2k}{l}}\left(\log n+1\right)^{m-k} \left(\frac{2B}{A}\right)^{2m-2k}\\
        &\ \ \ + \frac{(m-k)(m-k-1)(n+1)B^2}{A^2n^2} \\
        &\ \ \ + e^{\frac{B}{A}(m-k)}(n+1)^{(3+\frac{B}{A})(m-k)+1}\cdot\left(\exp\left(-\frac{3(\log n)^2}{16}\right)+\exp\left(-\left(\frac{1}{2}\right)^l(\log n)^2\right)\right)\,.
    \end{align*}
    Furthermore, since $(m-k)/l>1$, $\mathbb{P}_{m,n}\left(\text{Condorcet winning response}\right)\to 0$ as $n\to\infty$.
\end{restatable}

\subsection{An Overview of NLHF vs. RLHF} Now, we can make a comprehensive comparison between the solutions of NLHF and RLHF in different scenarios. To properly categorize the scenarios, we first discuss the relationship between Condorcet cycles and Condorcet winning responses. The general relationship between the existence of Condorcet cycles and Condorcet winning responses is the following:
\begin{restatable}{proposition}{condorcetandcondorcetwinningresponse}\label{prop: condorcet and Condorcet winning response}
    Under Assumption~\ref{ass:strict_preference}, suppose that there are $n$ ($n\geqslant 3$) responses $y_1,\ldots, y_n$. If there is no Condorcet winning response, there must exist a Condorcet cycle. In other words, if there is no Condorcet cycle, then there must exist a Condorcet winning response.
\end{restatable}
On the other hand, if there exists a Condorcet winning response, we cannot determine whether the preference structure contains Condorcet cycles. Building upon this, we can categorize human preferences into four scenarios. First, depending on whether human preferences contain Condorcet cycles, we can divide them into two cases based on whether the preferences can be modeled by a reward model. If they cannot be modeled by a reward model, we further divide into two cases based on whether a Condorcet winning response exists. If they can be modeled by a reward model, we further divide into two cases based on whether they can be modeled by the BTL model. Finally, this gives us four scenarios. An illustration is provided in Figure~\ref{fig:preference_tree_intro}.

\subsubsection*{Scenario 1: BTL Reward Model} We first consider the scenario where human preferences can be modeled by the BTL model. In this setting, there are no Condorcet cycles in the human preferences, thus there must exist a Condorcet winning response. Indeed, the Condorcet winning response in this case is the response with maximum BTL reward. As a result, when fully optimized to the optimal solution, both RLHF and NLHF\footnote{Again, under the regime where the coefficient $\tau$ of the KL term tends to zero. In addition, for RLHF, this regime can be broadened to cases where the output of $\pi_{\textnormal{ref}}$ collapses to a single response for any $\tau$ \citep{xiao2024algorithmic}.} converge to this response. Below, we formally state this result.

\begin{restatable}{theorem}{BTLNLHF}\label{thm:BTL NLHF}
Under Assumption~\ref{ass:strict_preference}, suppose there are \( n \) responses, and human preferences over these responses are modeled by the BTL model. In this case, a Condorcet winning response exists. Consequently, the Nash equilibrium of NLHF and the optimal solution of RLHF both correspond to this Condorcet winning response.
\end{restatable}
\subsubsection*{Scenario 2: Non-BTL Reward Model} For RLHF, since the human preferences cannot be modeled by the BTL model, Step 2 (reward learning) in RLHF only learns an approximated BTL reward, and the RLHF solution converges to a response with maximum approximated reward. Notice that such a response might not be the Condorcet winning response. Below we provide an example.

\begin{example}  
\label{exam:rlhf_solution}
Given a prompt, suppose five labelers rank three responses \( \{y_1, y_2, y_3\} \). Two labelers prefer \( y_1 \succ y_2 \succ y_3 \), while the remaining three prefer \( y_3 \succ y_1 \succ y_2 \).
\end{example}
Under the pairwise majority rule, \( y_3 \) emerges as the Condorcet winning response, as it is preferred over any other response in pairwise comparisons. However, \( y_1 \) has the highest average ranking, corresponding to the highest approximated reward. As a result, a fully optimized RLHF-trained model will collapse to always selecting \( y_1 \). Thus, Example~\ref{exam:rlhf_solution} directly leads to the following proposition:
\begin{proposition}
The optimal solution of RLHF, when considered as a function of human preferences, is not Condorcet consistent.
\end{proposition}
On the other hand, for NLHF, since there are no Condorcet cycles in the human preferences, there must exist a Condorcet winning response that NLHF will converge to. Therefore, although both RLHF and NLHF suffer from the preference collapse issue, the response that NLHF converges to has better properties.

\subsubsection*{Scenario 3: No Reward Model \& Exists a Condorcet winning Response} This scenario is essentially equivalent to Scenario 2 in terms of the solutions of RLHF and NLHF. For RLHF, regardless of whether the human preferences can be modeled by a non-BTL reward, RLHF cannot learn a BTL reward. For NLHF, the Nash solution is the Condorcet winning response.
\subsubsection*{Scenario 4: No Reward Model \& No Condorcet winning Response} In this case, the Nash equilibria of NLHF are mixed strategies, which preserve the diversity of opinions. 
Moreover, since the probability of this scenario approaches one under the IC condition, this suggest that NLHF is better than RLHF in diversifying outputs.

\subsection{Extension of Condorcet Winning Response}
Finally, we examine how mixed strategies are distributed across responses. To this end, we extend the concept of a Condorcet winning response to Condorcet winning sets of responses.
The following theorem provides a fundamental partition of the response space:
\begin{restatable}{theorem}{responsedecomposition}\label{thm:response decomposition}
Under Assumption~\ref{ass:strict_preference}, the set of responses can be partitioned into disjoint subsets $S_1,\ldots,S_k$ such that: 
\begin{enumerate}
\item Each $S_i$ either forms a Condorcet cycle or is a single response.
\item For any $j>i$, any response $y\in S_i$ and $y'\in S_j$, $\cP(y\succ y')>\frac{1}{2}$.
\end{enumerate}
Moreover, this decomposition is unique.
\end{restatable}
An algorithm for obtaining $S_1$ with complexity $O(n^2)$ and such decomposition with complexity $O(n^3)$ is provided in Appendix. The case where there exists a Condorcet winning response occurs when $|S_1|=1$. Based on Theorem~\ref{thm:response decomposition}, we refer to the subset $S_1$ as the Condorcet winning set. The following theorem generalizes our previous theorem to the Condorcet winning set $S_1$.

\begin{restatable}{theorem}{support}\label{thm:support}
    Under Assumption~\ref{ass:strict_preference}, the support of the strategies of any Nash equilibrium must be contained in the Condorcet winning subset $S_1$.
\end{restatable}
A natural following question is whether the support is exactly $S_1$. We observe that this is indeed true when $|S_1|=3$, where the Nash equilibrium is unique and fully supported on $S_1$.
However, this is generally false when $|S_1|>3$ as shown by the following example:
\begin{example}\label{eg:not fully supported mixed strategy}
    For the payoff matrix in Table \ref{table:4-cycle}, the four responses form a Condorcet 4-cycle. However, $(\frac{1}{3},\frac{1}{3},\frac{1}{3},0)$ is the unique Nash equilibrium of this game, which is not fully support on all four responses.
    \begin{table}[htbp]
    \caption{Payoff matrix with four responses $\{y_1,y_2,y_3,y_4\}$.}
    \label{table:4-cycle}
    \centering
    \renewcommand{\arraystretch}{1.5}
    \scalebox{1.1}{
    \begin{tabular}{c | c  c  c  c }
          $\cP(y\succ y')$ &  
          $y'=y_1$ & 
          $y'=y_2$ &
          $y'=y_3$ &
          $y'=y_4$
         \\\hline
         $y=y_1$ &
         $\frac{1}{2}$ &
         $\frac{2}{3}$ &
         $\frac{1}{3}$ &
         $\frac{1}{3}$
         \\
         $y=y_2$ &
         $\frac{1}{3}$ &
         $\frac{1}{2}$ &
         $\frac{2}{3}$ &
         $\frac{2}{3}$
         \\
         $y=y_3$ &
         $\frac{2}{3}$ &
         $\frac{1}{3}$ &
         $\frac{1}{2}$ &
         $\frac{2}{3}$
         \\
         $y=y_4$ &
         $\frac{2}{3}$ &
         $\frac{1}{3}$ &
         $\frac{1}{3}$ &
         $\frac{1}{2}$
         \\ 
    \end{tabular}
    }
\end{table}
\end{example}
To conclude this section, we note that under certain conditions, such as in Erd\H{o}s random graphs, the probability that the graph contains a Condorcet cycle including all responses approaches one as $n\to\infty$~\citep{erdds1959random}. Therefore, the probability of $|S_1|=n$ approaches one, and all responses could possibly be included in the Nash equilibrium strategies. This demonstrates NLHF's potential ability to preserve preferences for every single response.

\section{Proof Sketch of Theorem \ref{thm:no Condorcet winning}}\label{sec:upper bound proof}

In this subsection, we sketch the main idea to upper-bound the probability in Theorem~\ref{thm:no Condorcet winning}. 
First, we demonstrate in Lemma \ref{lem:luce-model-equivalent-formula} that selecting a permutation from the Luce model, as stated in Assumption \ref{ass:luce}, is equivalent to assigning exponentially-distributed scores to each response independently, and then ranking these scores accordingly.

\begin{lemma}\label{lem:luce-model-equivalent-formula}
    The Luce model $\operatorname{Luce}(\lambda_1,\dots,\lambda_n)$ is equivalent to the preference generation scheme: 
    \begin{enumerate}
        \item Sample $n$ scores from independent exponential random variables $U_k$ with parameters $\lambda_k$, denoted by $\{U_k\sim \operatorname{Exp}(\lambda_k)\}_{1\leqslant k\leqslant n}$;
        \item Order the scores from small to large as $U_{i_1}<\cdots<U_{i_n}$; the preference is then given by $y_{i_1}\succ \cdots\succ y_{i_n}$.
    \end{enumerate}
\end{lemma}

 This is a manifestation of the memorylessness of exponential random variable, recorded in the following Lemma.

\begin{restatable}{lemma}{integrationexponential}\label{lem:integration-exponential}
    Let $X_1,\ldots,X_n$ be independent exponential random variables with parameters $\lambda_1,\ldots,\lambda_n$ respectively. Then, we have
    \[\mathbb{P}\left(X_1>\cdots>X_n\right)=\frac{\lambda_n}{\lambda_1+\cdots+\lambda_n}\cdot\frac{\lambda_{n-1}}{\lambda_1+\cdots+\lambda_{n-1}}\cdot\cdots\cdot\frac{\lambda_2}{\lambda_1+\lambda_2}\,.\]
\end{restatable}

\begin{proof}[Proof of Lemma \ref{lem:luce-model-equivalent-formula}]
   By Lemma \ref{lem:integration-exponential}, for any permutation $\{i_1,\ldots,i_n\}=\{1,\ldots,n\}$, 
    \begin{align*}
        \mathbb{P}_{\rm score}\left(y_{i_1}\succ \cdots\succ y_{i_n}\right)= &\ \mathbb{P}\left(U_{i_n}>\cdots>U_{i_1}\right)\\
        \overset{\ref{lem:integration-exponential}}{=} &\ \frac{\lambda_{i_1}} {\lambda_{i_1}+\cdots+\lambda_{i_n}}\cdot \frac{\lambda_{i_2}}{\lambda_{i_2}+\cdots+\lambda_{i_n}}\cdot\cdots\cdot\frac{\lambda_{i_{n-1}}}{\lambda_{i_n}+\lambda_{i_{n-1}}} \\
        = &\ \mathbb{P}_{\rm Luce}\left(y_{i_1}\succ \cdots\succ y_{i_n}\right)\,.
    \end{align*}
    Hence, we complete the proof.
\end{proof}

Define the scores $\{U^{j}_i\}$ where $i=1,\ldots,n+1$ for each $j=1,\ldots,m$ according to Lemma \ref{lem:luce-model-equivalent-formula}. The event that there exists a Condorcet winning response can be expressed as a disjoint union, i.e. 
\[ \{\text{Condorcet winning response exists}\}=\cup_{i=1}^{n+1}\left\{y_i\text{ is the Condorcet winning response}\right\}.\]
Unraveling the definition of Condorcet winning response and the definition of preference yields
\[\mathbb{P}\left(y_1\text{ is the Condorcet winning response}\right)=\mathbb{P}\left(\bigcap_{i=2}^{n+1}\left\{\frac{1}{m}\sum_{j=1}^{m} \mathds{1}_{\{U^{j}_1<U^{j}_i\}}>\frac{1}{2}\right\}\right),\]
which is difficult to characterize due to the dependency between each event in the intersection.
To avoid this difficulty, we can leverage the fact that $y_1$ is the Condorcet winning response is equivalent to there does not exist any $i\neq 1$ such that $U_i^{j}<U_1^{j}$ for more than $l=\lceil\frac{m}{2}\rceil$ labelers.
We can define the following random subsets for $j=1,\ldots,m$,
\[ N_j := \left\{i\in \{2,\cdots,n+1\}\mid U^{j}_i<U^{j}_1\right\},\]
which represents the set of responses preferred over $y_1$ from the perspective of individual $j$. 
We also define the cardinality of the random sets $X_j=|N_j|$, which counts the number of responses preferred over $y_1$ by individual $j$. 
We can further define their order statistics $X_{(j)}$ with $X_{(1)}\geqslant X_{(2)} \geqslant \ldots\geqslant X_{(m)}$, and the corresponding sets are denoted by $N_{(j)}$. 
Then, we have the upper bound 
    \[ \mathbb{P}(y_1\text{ is the Condorcet winning response})\leqslant  \mathbb{P}\left(N_{(1)}\cap\cdots\cap N_{(l)}=\emptyset\right).\]
We then upper bound this probability by splitting it into two steps.

\paragraph*{Step 1.} 
In the first step, we examine the following conditional probability:
\begin{equation}
    \mathbb{P}\left(N_{(1)}\cap\cdots\cap N_{(l)}=\emptyset\mid X_{(1)}\cdots X_{(l)}\geqslant n^{l-1}(\log n)^2\right)\,.
\end{equation}
We can actually drop the order here and define the event $\mathcal{A}:=\left\{N_1\cap\cdots\cap N_l=\emptyset\right\}$ and $\boldsymbol{U}_1=(U_1^1,\ldots,U_1^m)^\top\in\mathbb{R}^m$. Now we define the critical set $\boldsymbol{B}\subseteq\mathbb{R}^m$ by
\[
\boldsymbol{B}:=\left\{\boldsymbol{u}_1=(u_1^1,\cdots,u_1^m)^\top\in\mathbb{R}^m: \prod_{j=1}^l\left(\sum_{i=2}^{n+1}\left(1-e^{-\lambda_i u_1^j}\right)\right)\geqslant \left(\frac{1}{2}\right)^l n^{l-1}(\log n)^2\right\}\,.
\]
Using Bayesian formula, we obtain the upper bound
\begin{equation}
\begin{aligned}
    &\ \mathbb{P}\left(\mathcal{A}\mid X_1=x_1,\cdots, X_m=x_m\right)\\
    &\ \leqslant\frac{\mathbb{P}\left(\boldsymbol{U}_1\in\boldsymbol{B}^c, X_1=x_1,\cdots, X_m=x_m\right)}{\mathbb{P}\left(X_1=x_1,\cdots, X_m=x_m\right)}+\frac{\mathbb{P}\left(\mathcal{A},\boldsymbol{U}_1\in\boldsymbol{B}\right)}{\mathbb{P}\left(X_1=x_1,\cdots, X_m=x_m\right)}\,.
\end{aligned}
\end{equation}

The denominator only contributes a polynomial factor in $n$ by the following lemma:
\begin{restatable}{lemma}{cardinalitylowerbound}\label{lem:cardinality-lower-bound}
Suppose $U_i\sim \operatorname{Exp}(\lambda_i)$ independently for $i=1,\dots,n+1$, where each $\lambda_i\in [A,B]$ and $B\geqslant A>0$.
Let $N=\{i\in \{2,\cdots,n+1\}, U_i<U_1\}$ and $X=|N|$. 
Then for any $k=0,\dots,n$,
    \[\mathbb{P}\left(X=k\right)\geqslant e^{-\frac{B}{A}}(n+1)^{-3-\frac{B}{A}}\,.\]
\end{restatable}

The numerator in the first term is expected to be exponentially small because condition on $\boldsymbol{u}_1\in \boldsymbol{B}^c$, each $X_j$ is a sum of Bernoulli random variables and the expectation of $X_{1}\dots X_{l}$ is less than $n^{l-1}(\log n)^2/2^{l}$. When $X_1\dots X_l\geqslant n^{l-1}(\log n)^2$, some $X_j$ must deviate from its expectation significantly, and the conclusion follows from Bernstein inequality.

To bound the numerator in the second term, we notice that the following conditional probability can be explicitly calculated by
\begin{equation*}
    \mathbb{P}\left(\mathcal{A}\mid \boldsymbol{U}_1=\boldsymbol{u}_1\right)
    =  \prod_{i=2}^{n+1}\left(1-\prod_{j=1}^l\left(1-e^{-\lambda_i u_1^j}\right)\right),
\end{equation*}
which can be shown to be exponentially small for any $\boldsymbol{u}_1\in \boldsymbol{B}$ using the following inequality:
\begin{restatable}{lemma}{revisionorderineqmultiitems}\label{lem:revision-order-ineq-multi-items}
    For any $n,l\in\mathbb{N}_+$, assume that $0\leqslant x_1^j\leqslant\ldots\leqslant x_n^j<1$ for any $1\leqslant j\leqslant l$, then
    \[
    \prod_{i=1}^n\left(1-\prod_{j=1}^l x_i^j\right)\leqslant \left(1-\prod_{j=1}^l\left(\frac{1}{n}\sum_{i=1}^n x_i^j\right)\right)^n
    \]
\end{restatable}

Putting all these together, we obtain
\begin{restatable}{lemma}{stepone}\label{lem:step-one}
Suppose $U_i^j\sim \operatorname{Exp}(\lambda_i)$ independently for $i=1,\dots,n+1$ and $j=1,\dots,m$, where each $\lambda_i\in [A,B]$ and $B\geqslant A>0$. Let $N_j=\{i\in \{2,\cdots,n+1\}, U_i^j<U_1^j\}$ and $X_j=|N_j|$. For any $1\leqslant l\leqslant m$,
    \[\begin{aligned}
    &\ \mathbb{P}\left(N_{1}\cap\cdots\cap N_{l}=\emptyset\mid X_{1}\cdots X_{l}\geqslant n^{l-1}(\log n)^2\right) \\
    &\ \leqslant e^{\frac{B}{A}m}(n+1)^{(3+\frac{B}{A})m}\cdot\left(\exp\left(-\frac{3(\log n)^2}{16}\right)+\exp\left(-\left(\frac{1}{2}\right)^l(\log n)^2\right)\right)\,.
\end{aligned}\]
\end{restatable}

\paragraph*{Step 2.}
In the second step, we control the probability of bad profiles 
\[\mathbb{P}(X_{(1)}\cdots X_{(l)}\leqslant n^{l-1}(\log n)^2)\,.\]
Note that $X_{(1)}\cdots X_{(l)}\leqslant n^{l-1}(\log n)^2$ implies $X_1\cdots X_m\leqslant \tilde{O}(n^{(l-1)\frac{m}{l}})$. 
When none of these $X_i$ is equal to zero, $\mathbb{P}(X_1\cdots X_m\leqslant \tilde{O}(n^{(l-1)\frac{m}{l}})\mid X_1\cdots X_m>0)$ can be upper bounded by $\tilde{O}(n^{-\frac{m}{l}})$ according to the following lemma:
\begin{restatable}{lemma}{lucecondorcetwinningbadprofileprob}\label{lem:luce Condorcet winning-bad profile prob}
Suppose $U_i^j\sim \operatorname{Exp}(\lambda_i)$ independently for $i=1,\dots,n+1$ and $j=1,\dots,m$, where each $\lambda_i\in [A,B]$ and $B\geqslant A>0$.
Let $N_j=\{i\in \{2,\cdots,n+1\}, U_i^j<U_1^j\}$ and $X_j=|N_j|$. Then for any $0<\alpha<1$ and $\beta>0$, we have
\[
\mathbb{P}\left(\prod_{j=1}^m X_j\leqslant n^{\alpha m}\left(\log n\right)^\beta\,\,\Big\vert\,\,X_1\cdots X_m>0\right)\leqslant \frac{n^{\alpha m}}{n^m} \left(\log n\right)^{\beta}\left(\log n+1\right)^{m}\left(\frac{2B}{A}\right)^{2m}.
\]
\end{restatable}

Thus, $\mathbb{P}(X_{(1)}\cdots X_{(l)}\leqslant n^{l-1}(\log n)^2)$ is also expected to be upper bounded by $\tilde{O}(n^{-\frac{m}{l}})$ as long as few $X_i$ are equal to zero. 
This is summarized in the following lemma:
\begin{restatable}{lemma}{steptwo}\label{lem:step-two}
    Suppose $U_i^j\sim \operatorname{Exp}(\lambda_i)$ independently for $i=1,\dots,n+1$ and $j=1,\dots,m$, where each $\lambda_i\in [A,B]$ and $B\geqslant A>0$. Let $N_j=\{i\in \{2,\cdots,n+1\}, U_i^j<U_1^j\}$ and $X_j=|N_j|$.
    Then, we define the order statistics as $X_{(1)}\geqslant \ldots\geqslant X_{(m)}$. For any $m\geqslant 2$ and $1\leqslant l\leqslant m-1$, we have
\begin{equation*}
\begin{aligned}
&\mathbb{P}\left(X_{(1)}\cdots X_{(l)}
\leqslant n^{l-1}\left(\log n\right)^2\right)\\
        \leqslant &(m+1)n^{-\frac{m}{l}}\left(\log n\right)^{\frac{2m}{l}}\left(\log n+1\right)^{m} \left(\frac{2B}{A}\right)^{2m}+ \frac{m(m-1)B^2}{A^2n^2}.
\end{aligned}
\end{equation*}
\end{restatable}

Finally, Theorem \ref{thm:no Condorcet winning} is proved by combining Lemma~\ref{lem:step-one} and~\ref{lem:step-two}. 
\begin{proof}[Proof of Theorem \ref{thm:no Condorcet winning}]
 Let $l=\lceil\frac{m}{2}\rceil$. As $m\geqslant 2$, we have $l\geqslant 1$.
    Recall that
    \begin{align*}
\mathbb{P}\left(\text{$y_1$ is the Condorcet winning response}\right)
        \leqslant \mathbb{P}\left(N_{(1)}\cap\cdots\cap N_{(l)}=\emptyset\right):=\mathbb{P}\left(\mathcal{A}\right) \,.
    \end{align*}
    Applying Lemma \ref{lem:step-one} and Lemma \ref{lem:step-two} with $l=\lceil\frac{m}{2}\rceil$, 
    \begin{align*}
        \mathbb{P}\left(\mathcal{A}\right)=&\  \mathbb{P}\left(\mathcal{A}\,,\, X_{(1)}\cdots X_{(l)}\leqslant n^{l-1}(\log n)^2\right)+ \mathbb{P}\left(\mathcal{A}\,,\, X_{(1)}\cdots X_{(l)}\geqslant n^{l-1}(\log n)^2\right)\\
        \leqslant &\ \mathbb{P}\left(X_{(1)}\cdots X_{(l)}\leqslant n^{l-1}(\log n)^2\right)+ \mathbb{P}\left(\mathcal{A}\mid X_{(1)}\cdots X_{(l)}\geqslant n^{l-1}(\log n)^2\right) \\
        \leqslant &\ (m+1)n^{-\frac{m}{l}}\left(\log n\right)^{\frac{2m}{l}}\left(\log n+1\right)^{m} \left(\frac{2B}{A}\right)^{2m}+ \frac{m(m-1)B^2}{A^2n^2} \\
        &\ \quad + e^{\frac{B}{A}m}(n+1)^{(3+\frac{B}{A})m}\cdot\left(\exp\left(-\frac{3(\log n)^2}{16}\right)+\exp\left(-\left(\frac{1}{2}\right)^l(\log n)^2\right)\right)\,.
    \end{align*}
    Finally, as the above bound does not depend on the specific choice of $y_1$, we have
    \begin{align*}
        &\ \mathbb{P}_{m,n}\left(\text{Condorcet winning response}\right)= \sum_{i=1}^{n+1}\mathbb{P}\left(\text{$y_i$ is the Condorcet winning response}\right)\\
        &\ \leqslant (m+1)(n+1)n^{-\frac{m}{l}}\left(\log n\right)^{\frac{2m}{l}}\left(\log n+1\right)^{m} \left(\frac{2B}{A}\right)^{2m}+ \frac{m(m-1)(n+1)B^2}{A^2n^2} \\
        &\ \quad + e^{\frac{B}{A}m}(n+1)^{(3+\frac{B}{A})m+1}\cdot\left(\exp\left(-\frac{3(\log n)^2}{16}\right)+\exp\left(-\left(\frac{1}{2}\right)^l(\log n)^2\right)\right)\,.
    \end{align*}
    Hence, we complete the proof.
\end{proof}

\section{Proof of Theorem \ref{thm:no Condorcet winning lower bound}}\label{sec:lower bound proof}

We first state the following technical lemma:

\begin{restatable}{lemma}{generatingfunction}\label{lem:generating function}
    There exists a universal constant $0<c<1$ such that 
    \begin{equation*}
        \sum_{k=0}^{\lfloor{\frac{n-1}{2}}\rfloor}\binom{n}{2k+1}\left(\frac{1}{2n}\right)^{2k+1} \leqslant 1-c\,,
    \end{equation*}
    for all integers $n\geqslant 1$.
\end{restatable}

\begin{proof}[Proof of Theorem \ref{thm:no Condorcet winning lower bound}]
Without loss of generality, we first lower bound the probability that $y_1$ is the Condorcet winning response. 
For $j=1,\ldots,m$, we define the random subsets
\[ N_j=\{i\in \{2,\cdots,n+1\}\mid U^{j}_i<U^{j}_1\},\]
which represent the set of responses that are preferred over $y_1$ from the perspective of individual $j$.
We also define their cardinality $X_j=|N_j|$, which counts the number of responses preferred over $y_1$ by individual $j$. 

Let $l=\lceil\frac{m}{2}\rceil$ and $\mathcal{T}$ be the set of all subsets of $\{1,\ldots,m\}$ with cardinality $l$. 
Then $|\mathcal{T}|=\binom{m}{l}$. For each element $t\in \mathcal{T}$, it is a set of indices $t=\{j_1,\ldots,j_l\}$, and we further define the set $T_t:=N_{j_1}\cap\ldots\cap N_{j_l}$. Now notice that 
\begin{equation}\label{eq:lower-winner-rewritten}
    \begin{aligned}
\mathbb{P}\left(\bigcap_{i=2}^{n+1}\left\{\mathcal{P}\left(y_1\succ y_i\right)>\frac{1}{2}\right\}\right)= &\ \mathbb{P}\left(\bigcap_{i=2}^{n+1}\left\{\frac{1}{m}\sum_{j=1}^m\mathds{1}_{\left\{U_1^j<U_i^j\right\}}>\frac{1}{2}\right\}\right)\\
= &\ \mathbb{P}\left(\bigcap_{i=2}^{n+1}\left\{\frac{1}{m}\sum_{j=1}^m\mathds{1}_{\left\{U_i^j<U_1^j\right\}}<\frac{1}{2}\right\}\right)= \mathbb{P}\left(\bigcap_{t\in\mathcal{T}}\left\{T_t=\emptyset\right\}\right)\,,
\end{aligned}
\end{equation}
where the last equality is because $y_1$ is the Condorcet winning response is equivalent to there does not exist any $i\neq 1$ such that $U_i^{j}<U_1^{j}$ for more than $l=\lceil\frac{m}{2}\rceil$ labelers.
Let us define $Z_i:=\mathds{1}\left\{i\in \cup_{t\in\mathcal{T}}T_t\right\}$ for any $i\in \{2,\ldots,n+1\}$. Then
\begin{equation}\label{eq:lower-winner-rewritten-Z}
    \begin{aligned}
        \mathbb{P}\left(\bigcap_{t\in\mathcal{T}}\left\{T_t=\emptyset\right\}\right)=\mathbb{P}\left(\bigcup_{t\in\mathcal{T}}T_t=\emptyset \right)=\mathbb{P}\left(\bigcap_{i=2}^{n+1}\left\{i\notin \bigcup_{t\in\mathcal{T}}T_t\right\} \right)=\mathbb{P}\left(Z_2=0,\cdots,Z_{n+1}=0\right)\,.
    \end{aligned}
\end{equation}
Denote $\boldsymbol{U}_1=(U_1^1,\ldots,U_1^m)^\top\in\mathbb{R}^m$ and define the critical set
\[
\boldsymbol{B}:=\left\{\boldsymbol{U}_1\in\mathbb{R}^m: 0\leqslant U_1^j\leqslant \frac{1}{B}\cdot\frac{1}{(2n)^{\frac{1}{l}}\binom{m}{l}^{\frac{1}{l}}}\,,\text{ for any } j=1,\dots,m\right\}\,.
\]
Using the numerical inequality $1-e^{-x}\geqslant \frac{x}{2}$ for $0\leqslant x\leqslant1$, we have
\begin{equation}\label{eq:lower-winner-critical-lower-bound-prob}
         \mathbb{P}\left(\boldsymbol{U}_1\in \boldsymbol{B}\right)=\prod_{j=1}^m \left(1-\exp{\left(-\frac{\lambda_1}{B\binom{m}{l}^{\frac{1}{l}}} \frac{1}{(2n)^{\frac{1}{l}}}\right)}\right)\geqslant\left(\frac{A}{B}\right)^m\cdot\frac{1}{2^{m+\frac{m}{l}}\binom{m}{l}^{\frac{m}{l}}}\cdot\frac{1}{n^{\frac{m}{l}}}\,.
\end{equation}

Next, for any $\boldsymbol{U}_1\in \boldsymbol{B}$, we are going to establish a lower bound for the conditional probability $\mathbb{P}\left(Z_2=0,\cdots,Z_{n+1}=0\mid\boldsymbol{U}_1\right)$. 
For any $k\in [n]$, let us first focus on $\mathbb{P}\left(Z_2=1,\ldots,Z_{k+1}=1\mid \boldsymbol{U}_1\right)$.
    Notice that conditional on $\boldsymbol{U}_1$, for any $t_{1},\ldots,t_k\in\mathcal{T}$, we have
    \begin{equation*}
        \mathbb{P}\left(2\in T_{t_1},\cdots, k+1\in T_{t_k}\mid \boldsymbol{U}_1\right)=\prod_{s=1}^k\prod_{j\in t_s}\left(1-e^{-\lambda_{s+1}U_1^j}\right)\,.
    \end{equation*}
    Using the numerical inequality $1-e^{-x}\leqslant x$ for $x\geqslant 0$, for any $U_1\in\bm{B}$, 
    \begin{equation*}
       \prod_{s=1}^k\prod_{j\in t_s}\left(1-e^{-\lambda_{s+1}U_1^j}\right)
        \leqslant \prod_{s=1}^k\prod_{j\in t_s}\lambda_{s+1}U_1^j
        \leqslant \prod_{s=1}^k\frac{1}{2n\binom{m}{l}}=\frac{1}{\binom{m}{l}^k} \left(\frac{1}{2n}\right)^k\,,
    \end{equation*}
    Therefore, using union bounds, 
    \begin{align*}
        \mathbb{P}\left(Z_2=1,\cdots,Z_{k+1}=1\mid \boldsymbol{U}_1\right)= &\ \mathbb{P}\left(2\in \bigcup_{t\in\mathcal{T}}T_t,\cdots, k+1\in\bigcup_{t\in\mathcal{T}}T_t\ \Bigg\vert \ \boldsymbol{U}_1\right)\\
        = &\ \mathbb{P}\left(\bigcup_{t_1,\cdots,t_k\in\mathcal{T}} \left\{2\in T_{t_1},\cdots, k+1\in T_{t_k}\right\}\ \Bigg\vert \ \boldsymbol{U}_1\right)\\
        \leqslant & \sum_{t_1,\cdots, t_k\in\mathcal{T}}\mathbb{P}\left(2\in T_{t_1},\cdots, k+1\in T_{t_k}\mid \boldsymbol{U}_1\right)\\
        \leqslant &\ \frac{\binom{m}{l}^k}{\binom{m}{l}^k}\left(\frac{1}{2n}\right)^k= \left(\frac{1}{2n}\right)^k\,.
    \end{align*}
    Notice that this result is independent of which $k$ responses we choose, thus it holds for any $\mathbb{P}\left(Z_{i_1}=1,\cdots, Z_{i_{k}}=1\mid \boldsymbol{U}_1\right)$ where $\{i_1,\ldots,i_k\}\subseteq\{2,\dots,n+1\}$.
    Then, by the inclusion-exclusion principle, we have 
    \begin{align*}
        &\mathbb{P}\left(Z_2=0,\cdots, Z_{n+1}=0\mid \boldsymbol{U}_1\right)\\= &\ 1-\sum_{k=1}^n(-1)^{k-1}\sum_{\{i_1,\cdots,i_k\}\subseteq\{2,\dots,n+1\}}  \mathbb{P}\left(Z_{i_1}=1,\cdots, Z_{i_{k}}=1\mid \boldsymbol{U}_1\right)\\
        \geqslant &\ 1-\sum_{k=0}^{\lfloor{\frac{n-1}{2}}\rfloor}\sum_{\{i_1,\cdots,i_{2k+1}\}\subseteq\{2,\dots,n+1\}}\mathbb{P}\left(Z_{i_1}=1,\cdots, Z_{i_{2k+1}}=1\mid \boldsymbol{U}_1\right)\\
        \geqslant &\ 1 - \sum_{k=0}^{\lfloor{\frac{n-1}{2}}\rfloor}\binom{n}{2k+1}\left(\frac{1}{2n}\right)^{2k+1}\geqslant c\,,
    \end{align*}
    where $c$ is the universal constant provided in Lemma \ref{lem:generating function}. Then, from \eqref{eq:lower-winner-rewritten}, \eqref{eq:lower-winner-rewritten-Z}, and \eqref{eq:lower-winner-critical-lower-bound-prob}, we get 
    \begin{align*}
        \mathbb{P}\left(\text{$y_1$ is the Condorcet winning response}\right)= &\  \mathbb{P}\left(Z_2=0,\cdots,Z_{n+1}=0\right)\\ 
        \geqslant &\ c\cdot\left(\frac{A}{B}\right)^m\cdot\frac{1}{2^{m+\frac{m}{l}}\binom{m}{l}^{\frac{m}{l}}}\cdot\frac{1}{n^{\frac{m}{l}}}\,.
    \end{align*}
   Finally, note that this bound holds true for any response $y_i$ where $i=1,\dots,n+1$, we have
    \begin{align*}
        \mathbb{P}_{m,n}\left(\text{Condorcet winning response}\right)= &\ \sum_{i=1}^{n+1}\mathbb{P}\left(\text{$y_i$ is the Condorcet winning response}\right)\\
        \geqslant &\  c\cdot\left(\frac{A}{B}\right)^m\cdot\frac{1}{2^{m+\frac{m}{l}}\binom{m}{l}^{\frac{m}{l}}}\cdot\frac{n+1}{n^{\frac{m}{l}}}\,.
    \end{align*}
    Hence, we complete the proof.
\end{proof}

\section{Discussion}\label{sec:discuss}
In this paper, we have analyzed the statistical limits of aligning LLMs with human preferences, establishing fundamental impossibility and possibility results. We have proven that reward models cannot fully capture human preferences in the presence of Condorcet cycles, which exist with probability converging to one under two probabilistic preference models, IC and the Luce model. This reveals an inherent limitation of reward-based alignment approaches such as RLHF. In response, we have investigated NLHF as an alternative, demonstrating that its solution avoids collapse to a single response if and only if no Condorcet winning response exists---a condition that holds with probability tending to one under the same probabilistic preference models. 

Our findings suggest several promising avenues for statistical research on LLM alignment. First, while we have established that reward models cannot fully capture human preferences in the presence of Condorcet cycles, it would be valuable to obtain quantitative bounds on the approximation gap, which may depend on the number and structure of cycles. This would provide deeper insights into the extent of distortion introduced by reward-based RLHF. Additionally, since the class of reward models in Definition \ref{def:what_is_reward} is the largest possible, an intriguing theoretical question is to explore more constrained, specific reward classes, such as those in which $\mathcal{P}(y \succ y')$ depends solely on the difference $r(y) - r(y')$. Another important direction involves extending NLHF so that, when human preferences follow the BT model, the aligned policy precisely matches the BT distribution---perhaps by regularizing the payoff in the NLHF two-player game. In this context, it is noteworthy that RLHF with appropriate regularization can match the BT model \citep{xiao2024algorithmic}. Furthermore, it would be crucial to preserve preferences for responses in Condorcet winning sets. The challenges shown in Table \ref{table:4-cycle}, where Nash equilibria may overlook certain preferences, highlight the need for the development of refined non-reward-based approaches. On a slightly different direction, it would be interesting to investigate approaches that directly combine different rankings to form preferences \citep{fan2024uncertainty, fan2024ranking, fan2024spectral}.

From a practical standpoint, we advocate for increased exploration of non-reward-based alignment approaches in commercial LLM development, given the crucial importance of preserving diverse human preferences. It would be valuable to investigate how much of the observed bias in current RLHF-aligned LLMs stems directly from the fundamental limitations of reward models. When interpreting our results for practitioners, it is important to recognize that we often consider scenarios with many responses and fully optimized fine-tuning, whereas in practice, labelers typically rank a few responses and fine-tuning terminates before convergence, thus leading to partial rather than complete collapse in RLHF-aligned LLMs \citep{casper2023open, xiao2024algorithmic, song2023reward}. Nevertheless, understanding these asymptotic statistical limits provides crucial insights into the fundamental capabilities and constraints of different alignment approaches, thereby informing more effective practices for developing LLMs that better reflect the full spectrum of human values and preferences.

Our work also opens up several theoretical questions. In our paper, we have characterized $\mathbb{P}(|S_1|=1)$. For $3\leqslant k\leqslant n-1$, we conjecture that $\mathbb{P}(|S_1|=k)=\Theta(n^{k(1-\frac{m}{l})})$ for $k=\Theta(1)$, $\Theta(n^{(n-k)(1-\frac{m}{l})})$ for $n-k=\Theta(1)$, and otherwise it decays superpolynomially. For $k=n$, $\mathbb{P}(|S_1|=n)$ is the probability that a Hamiltonian cycle exists, and we expect it to tend to $1$ as $n\to\infty$ with $\Theta(n^{1-\frac{m}{l}})$ convergence rate.
As far as we know, these probabilities have not been investigated even under the IC condition, let alone under the Luce model. 
As our results indicate that the Luce model has almost the same behaviors as the IC condition, an interesting question is whether these behaviors are somehow universal, i.e., they hold for other preference distributions; for example, preferences sampled from permutons. 
Finally, in our paper, we consider the regime where $m$ is fixed and $n$ tends to infinity. A further question is to determine these probabilities when both $m$ and $n$ tend to infinity.

{\small
\section*{Acknowledgments} 
We thank the associate editor and reviewers for constructive comments that helped us significantly improve the presentation of the paper. We are grateful to Hang Du for sharing insights that significantly simplified the proof of Theorem \ref{thm:no Condorcet winning}. This work was supported in part by NIH grants U01CA274576, R01EB036016, and R01EB037101, NSF grant DMS-2310679, a Meta Faculty Research Award, and Wharton AI for Business. The content is solely the responsibility of the authors and does not necessarily represent the official views of the NIH.

\bibliographystyle{abbrvnat}
\bibliography{Bibliography}
}

\clearpage
\appendix
\section{Additional Background}\label{app:Additional Background}
\subsection{Game Theory}

Following \citet{myerson2013game}, we denote any \emph{strategic-form game} \( \Gamma \) as
\[
\Gamma = (N, (C_i)_{i \in N}, (u_i)_{i \in N}),
\]
where:
\begin{itemize}
    \item \( N \) is the set of players,
    \item \( C_i \) is the set of actions for player \( i \) (in the context of LLMs, each response is an action),
    \item \( u_i: C \to \mathbb{R} \) is the payoff function for player \( i \),
\end{itemize}
and \( C \) denotes the set of all possible profiles (\emph{i.e.} combinations) of actions that may be chosen by the players. Specifically,
\[
C = \prod_{i \in N} C_i,
\]
where each player \( i \) chooses one of their actions from \( C_i \).

A strategy for any player \(i\) is a probability distribution over $C_i$. We let $\Delta(C_i)$ denote the set of all possible strategies for player $i$. 
A \emph{pure strategy} is a probability distribution that assigns probability 1 to a single action in $C_i$. For notational simplicity, we often define a pure strategy by an action in $C_i$. All other randomized strategies except pure strategies are called \emph{mixed strategies}.

A randomized-strategy profile $\sigma\in \prod_{i \in N} \Delta(C_i)$ is a \emph{Nash equilibrium} iff no player could increase his expected payoff by unilaterally deviating from the prediction of the randomized-strategy profile. The Nash equilibrium $\sigma$ is called an \emph{equilibrium in pure strategies} iff $\sigma$ is a pure-strategy profile.

For any player \( i \in N \) and any action \( d_i \in C_i \), the action \( d_i \) is said to be (strictly) \emph{dominated} by action \(e_i\in C_i \) if 
\[
u_i(c_{-i}, d_i) <  u_i(c_{-i}, e_i), \quad \forall c_{-i} \in C_{-i},
\]
where \( C_{-i} \) is the set of strategy profiles of all players other than \( i \).

\begin{definition}[Dominant Action]\label{def:dominant strategy}
    For a player \( i \in N \) and an action \( e_i \in C_i \),
    the action $e_i$ is said to be a (strictly) \emph{dominant action} if any other action \(d_i\in C_i\setminus\{e_i\}\) is (strictly) dominated by $e_i$.
\end{definition}
Next we introduce two-person zero-sum games, which is equivalent to the two-person constant-sum games which is the focus of this paper.

\begin{definition}[two-person zero-sum game]
    A \emph{two-person zero-sum game} in strategic form is any $F$ of the form $F = (\{1, 2\}, C_1, C_2, u_1, u_2)$ such that 

$$
u_2(c_1, c_2) = -u_1(c_1, c_2), \, \forall c_1 \in C_1, \, \forall c_2 \in C_2.
$$

\end{definition}

The following theorem summarizes important mathematical properties of such games: 
\begin{theorem}
    $(\sigma_1, \sigma_2)$ is a Nash equilibrium of a finite two-person zero-sum game $(\{1, 2\}, C_1, C_2, u_1, -u_1)$ if and only if

$$
\sigma_1 \in \arg\max_{\tau_1 \in \Delta(C_1)} \min_{\tau_2 \in \Delta(C_2)} u_1(\tau_1, \tau_2),\quad
\text{and}\quad
\sigma_2 \in \arg\min_{\tau_2 \in \Delta(C_2)} \max_{\tau_1 \in \Delta(C_1)} u_1(\tau_1, \tau_2).
$$

Furthermore, if $(\sigma_1, \sigma_2)$ is an equilibrium of this game, then

$$
u_1(\sigma_1, \sigma_2) = \max_{\tau_1 \in \Delta(C_1)} \min_{\tau_2 \in \Delta(C_2)} u_1(\tau_1, \tau_2) = \min_{\tau_2 \in \Delta(C_2)} \max_{\tau_1 \in \Delta(C_1)} u_1(\tau_1, \tau_2).
$$
\end{theorem}
The above property is a manifestation of the celebrated von Neumann's minimax theorem.
\begin{theorem}[von Neumann's Minimax Theorem]
    Let $X \subset \mathbb{R}^n$ and $Y \subset \mathbb{R}^m$ be compact convex sets. If $f: X \times Y \rightarrow \mathbb{R}$ is a continuous function that is concave-convex, i.e., $f(\cdot, y): X \to \mathbb{R}$ is concave for fixed $y$, and $f(x, \cdot): Y \to \mathbb{R}$ is convex for fixed $x$. Then we have that

$$
\max_{x \in X} \min_{y \in Y} f(x,y) = \min_{y \in Y} \max_{x \in X} f(x,y).
$$
\end{theorem}
As a remark, except from being constant-sum, the NLHF game \eqref{eq:ne} is symmetric in the sense that $C_1=C_2=C$ and $u_1(c,c')=u_2(c',c)$. This means that the game look the same to all players.

\subsection{Additional Related Work}

A popular alternative to reward-based RLHF is direct preference optimization (DPO) \citep{rafailov2023direct}. DPO fine-tunes LLMs directly on human preference data, eliminating the need to train a reward model, which makes it computationally more efficient. Several notable variants of DPO have been developed \citep{liu2023statistical,azar2024general,chang2024dataset,gorbatovski2024learn,rafailov2024r,yang2024asymptotics}. However, recent studies \citep{li2023policy,xu2024dpo,tajwar2024preference} suggest that DPO is less effective than reward-based RLHF for aligning LLMs. Both \cite{li2023policy} and \cite{xu2024dpo} argued that this inferiority stems from representation misspecification in DPO, which limits its ability to achieve robust alignment compared to reinforcement learning approaches like PPO. Moreover, the on-policy nature of reward-based fine-tuning enhances LLM performance by mitigating distribution shifts between the training dataset and online responses \citep{tajwar2024preference}.

In addition to NLHF, recent research has explored the formulation of two-player constant-sum games for aligning human preferences \citep{swamy2024minimaximalist, chen2024self}. These studies employ mirror descent \citep{beck2003mirror, beck2017first} to learn the Nash equilibrium. Another line of work focuses on directly learning the Nash equilibrium from a preference dataset \citep{wu2024self, rosset2024direct, calandriello2024human, zhang2024iterative}. Furthermore, the work of \citep{wang2025magnetic, liu2024comal} analyzed the convergence of NLHF to the Nash equilibrium in the last iterate.

\section{Auxiliary Lemmas}
\subsection{Properties of Tournament Graphs}
In this section, we present two fundamental results in graph theory that are useful for our analysis. A simple directed graph is a directed graph that contains no loops (i.e., edges connecting a vertex to itself) and no multiple edges between the same pair of vertices. A tournament graph is a simple directed graph in which every pair of vertices is connected by exactly one directed edge. In the following analysis of preferences, we often reformulate the preference problem in terms of a directed graph \( G \). Specifically, we can represent the \( n \) responses \( y_1, \dots, y_n \) as vertices \( 1, \dots, n \), and introduce a directed edge \( i \rightarrow j \) if and only if \( \cP(y_i \succ y_j) > 1/2 \). 

\begin{lemma}\label{lem:graph-Condorcet winning and cycle}
    Consider a tournament graph $G$ with $n(n\geqslant 3)$ points $y_1,\ldots,y_n$. If for any $y\in G$, there exists $y^\prime\neq y$ such that $y^\prime\rightarrow y$, then there exists a cycle in this graph. In other words, if there is no cycle in this graph, there exists $y^\star\in G$ such that $y^\star\rightarrow y^\prime$ for any $y^\prime\neq y^\star$.
\end{lemma}
\begin{proof}[Proof of Lemma \ref{lem:graph-Condorcet winning and cycle}]
Starting from an arbitrary vertex \(y_{i_1} \in G \), we can construct the following sequence in this graph: $y_{i_1} \leftarrow y_{i_2} \leftarrow \ldots$. However, since the graph contains only \( n \) vertices, there must exist indices \( j_1 < j_2 \) such that \( y_{j_1} = y_{j_2} \) while all vertices \( y_{j_1}, \dots, y_{j_2-1} \) are distinct. This implies the existence of a cycle $y_{j_2}\rightarrow \ldots \rightarrow y_{j_1}$. 
\end{proof}

\begin{lemma}\label{lem: graph theory}
    Consider a tournament graph $G$ with $n$ ($n\geqslant 3$) points $y_1,\ldots, y_n$. 
    Then there exists a Hamiltonian path in $G$, i.e., a sequence of distinct indices ${i_1, \dots, i_n} = {1, \dots, n}$ such that $y_{i_1}\rightarrow\ldots\rightarrow y_{i_n}$.
\end{lemma}
\begin{proof}[Proof of Lemma \ref{lem: graph theory}]
We prove the statement by induction on $n$.

Base case ($n=3$):
Without loss of generality, assume $y_1 \rightarrow y_2$. We analyze three cases based on the direction of edges involving $y_3$:
\begin{itemize} \item If $y_2 \rightarrow y_3$, then $y_1 \rightarrow y_2 \rightarrow y_3$ forms the required Hamiltonian path (see the left graph in Figure~\ref{fig:demonstration n=3}).
\item If $y_3 \rightarrow y_2$ and $y_1 \rightarrow y_3$, then $y_1 \rightarrow y_3 \rightarrow y_2$ forms the required Hamiltonian path (see the middle graph in Figure~\ref{fig:demonstration n=3}).
\item If $y_3 \rightarrow y_2$ and $y_3 \rightarrow y_1$, then $y_3 \rightarrow y_1 \rightarrow y_2$ forms the required Hamiltonian path (see the right graph in Figure~\ref{fig:demonstration n=3}).
\end{itemize}

Inductive step:
Suppose the claim holds for $n$. That is, for any tournament graph with $n$ vertices, there exists a Hamiltonian path. Consider a tournament graph with $n+1$ vertices, $y_1, \dots, y_{n+1}$. By the induction hypothesis, there exists an ordering
$y_{i_1}\rightarrow\ldots\rightarrow y_{i_n}$ such that $\{i_1, \dots, i_n\} = \{1, \dots, n\}$. We now insert $y_{n+1}$ into this sequence, considering three cases:
\begin{itemize} 
    \item If $y_{i_k} \rightarrow y_{n+1}$ for all $k \in [n]$, then appending $y_{n+1}$ at the end gives the Hamiltonian path  $y_{i_1}\rightarrow\ldots\rightarrow y_{i_n}\rightarrow y_{n+1}$. (See the left graph in Figure~\ref{fig: demonstration n+1 part I}.)
    \item If $y_{n+1} \rightarrow y_{i_1}$, then prepending $y_{n+1}$ at the beginning gives the Hamiltonian path $y_{n+1}\rightarrow y_{i_1}\rightarrow\ldots\rightarrow y_{i_n}$. (See the right graph in Figure~\ref{fig: demonstration n+1 part I}.)
    \item Otherwise, there exists a smallest index $j \in {2, \dots, n}$ such that $y_{n+1} \rightarrow y_{i_j}$. In this case, we observe that $y_{i_k} \rightarrow y_{n+1}$ for all $k = 1, \dots, j-1$, allowing us to insert $y_{n+1}$ at position $j$: $y_{i_1}\rightarrow\ldots\rightarrow y_{i_{j-1}}\rightarrow y_{n+1}\rightarrow y_{i_j}\rightarrow\ldots y_{i_n}$. (See Figure~\ref{fig:demonstration n+1 part II} for an illustration.)
\end{itemize}
By induction, the claim holds for all $n \geqslant 3$.
\end{proof}

\begin{figure}
\centering
\begin{tikzpicture}
\node[draw, circle] (1) at (-4.8, 0) {$y_1$};
\node[draw, circle] (2) at (-3.6, 1) {$y_2$};
\node[draw, circle] (3) at (-2.4, 0) {$y_3$};
\tikzset{>=stealth}
\draw[->, line width=0.3mm, color=red] (1) -- (2);
\draw[->, line width=0.3mm, color=red] (2) -- (3);
    
\node[draw, circle] (1) at (-1.2, 0) {$y_1$};
\node[draw, circle] (2) at (0, 1) {$y_2$};
\node[draw, circle] (3) at (1.2, 0) {$y_3$};
\tikzset{>=stealth}
\draw[->, line width=0.3mm] (1) -- (2);
\draw[->, line width=0.3mm, color=red] (3) -- (2);
\draw[->, line width=0.3mm, color=red] (1) -- (3);

\node[draw, circle] (1) at (2.4, 0) {$y_1$};
\node[draw, circle] (2) at (3.6, 1) {$y_2$};
\node[draw, circle] (3) at (4.8, 0) {$y_3$};
\tikzset{>=stealth}
\draw[->, line width=0.3mm, color=red] (1) -- (2);
\draw[->, line width=0.3mm] (3) -- (2);
\draw[->, line width=0.3mm, color=red] (3) -- (1);
\end{tikzpicture}
\caption{Demonstration for $n=3$: the red arrows show the Hamiltonian path in the directed graph.}
\label{fig:demonstration n=3}
\end{figure}

\begin{figure}
\centering
\begin{tikzpicture}
\node[draw, circle] (1) at (-6.6, 0) {$y_{i_1}$};
\node (rest) at (-4.8,0) {$\cdots\cdots$};
\node[draw, circle] (3) at (-3, 0) {$y_{i_n}$};
\node[draw, circle] (n+1) at (-4.8, 1.5) {$y_{n+1}$};
\tikzset{>=stealth}
\draw[->, line width=0.3mm, color=red] (1) -- (rest);
\draw[->, line width=0.3mm, color=red] (rest) -- (3);
\draw[->, line width=0.3mm, color=red] (3) -- (n+1);
        
\node[draw, circle] (1) at (-1.5, 0) {$y_{i_1}$};
\node (rest) at (0.3,0) {$\cdots\cdots$};
\node[draw, circle] (3) at (2.1, 0) {$y_{i_n}$};
\node[draw, circle] (n+1) at (0.3, 1.5) {$y_{n+1}$};
\tikzset{>=stealth}
\draw[->, line width=0.3mm, color=red] (1) -- (rest);
\draw[->, line width=0.3mm, color=red] (rest) -- (3);
\draw[->, line width=0.3mm, color=red] (n+1) -- (1);
\end{tikzpicture}
\caption{Demonstration for induction in case $n+1$: the red arrows show the Hamiltonian path in the graph.}
\label{fig: demonstration n+1 part I}
\end{figure}

\begin{figure}
\centering
\begin{tikzpicture}
\node[draw, circle, minimum size=1.2cm] (1) at (3., 0) {$y_{i_1}$};
\node (rest1) at (4.5,0) {$\cdots$};
\node[draw, circle, minimum size=1.cm] (j-1) at (6.,0) {$y_{i_{j-1}}$};
\node[draw, circle, minimum size=1.2cm] (j) at (7.75,0) {$y_{i_j}$};
\node (rest2) at (9.25,0) {$\cdots$};
\node[draw, circle, minimum size=1.2cm] (3) at (10.75, 0) {$y_{i_n}$};
\node[draw, circle, minimum size=1cm] (n+1) at (6.875, 2.5) {$y_{n+1}$};
\tikzset{>=stealth}
\draw[->, line width=0.3mm, color=red] (1) -- (rest1);
\draw[->, line width=0.3mm, color=red] (rest1) -- (j-1);
\draw[->, line width=0.3mm] (j-1) -- (j);
\draw[->, line width=0.3mm, color=red] (j) -- (rest2);
\draw[->, line width=0.3mm, color=red] (rest2) -- (3);
\draw[->, line width=0.3mm] (1) -- (n+1);
\draw[->, line width=0.3mm] (n+1) -- (3);
\draw[->, line width=0.3mm] (rest1) -- (n+1);
\draw[->, line width=0.3mm] (n+1) -- (rest2);
\draw[->, line width=0.3mm, color=red] (j-1) -- (n+1);
\draw[->, line width=0.3mm, color=red] (n+1) -- (j);
\end{tikzpicture}
\\[16pt]
\begin{tikzpicture}
\node[draw, circle, minimum size=1.1cm] (1) at (-4.4, 0) {$y_{i_1}$};
\node[draw, circle, minimum size=1.1cm] (rest1) at (-2.6,0) {$y_{i_2}$};
\node (rest2) at (-1.0,0) {$\cdots$};
\node[draw, circle, minimum size=1.1cm] (3) at (0.7, 0) {$y_{i_n}$};
\node[draw, circle, minimum size=1cm] (n+1) at (-3.5, 2.) {$y_{n+1}$};
\tikzset{>=stealth}
\draw[->, line width=0.3mm] (1) -- (rest1);
\draw[->, line width=0.3mm, color=red] (rest2) -- (3);
\draw[->, line width=0.3mm, color=red] (rest1) -- (rest2);
\draw[->, line width=0.3mm, color=red] (1) -- (n+1);
\draw[->, line width=0.3mm] (n+1) -- (3);
\draw[->, line width=0.3mm, color=red] (n+1) -- (rest1);
\node[draw, circle, minimum size=1.1cm] (1) at (2.9, 0) {$y_{i_1}$};
\node (rest1) at (4.65,0) {$\cdots$};
\node[draw, circle, minimum size=1cm] (rest2) at (6.4,0) {$y_{i_{n-1}}$};
\node[draw, circle, minimum size=1.1cm] (3) at (8.2, 0) {$y_{i_n}$};
\node[draw, circle, minimum size=1.1cm] (n+1) at (7.3, 2.) {$y_{n+1}$};
\tikzset{>=stealth}
\draw[->, line width=0.3mm, color=red] (1) -- (rest1);
\draw[->, line width=0.3mm] (1) -- (n+1);
\draw[->, line width=0.3mm, color=red] (rest1) -- (rest2);
\draw[->, line width=0.3mm] (rest1) -- (n+1);
\draw[->, line width=0.3mm, color=red] (rest2) -- (n+1);
\draw[->, line width=0.3mm, color=red] (n+1) -- (3);
\draw[->, line width=0.3mm] (rest2) -- (3);
\end{tikzpicture}
\caption{Demonstration for induction in case $n+1$: the red arrows show the Hamiltonian path in the graph.}
\label{fig:demonstration n+1 part II}
\end{figure}

\begin{lemma}\label{lem: graph theory of response decomposition}
Consider a tournament graph $G$ with $n$ ($n\geqslant 3$) vertices. Then the vertex set of $G$ can be partitioned into disjoint subsets $S_1,\ldots,S_k$ such that:
\begin{enumerate}
    \item For each $S_i$, either $|S_i|=1$ or the subgraph induced by $S_i$ contains a Hamiltonian cycle;
    \item For any $i<j$, and for all $y\in S_i$, $y'\in S_j$, we have $y\rightarrow y'$.
\end{enumerate}
Moreover, this decomposition is unique.
\end{lemma}
\begin{proof}[Proof of Lemma \ref{lem: graph theory of response decomposition}]
First, we prove the existence of such decomposition by induction on $n$. The base case $n=3$ is straightforward to verify. Then we prove this conclusion for any $n\geqslant 3$ by induction. In the case where $\vert G\vert=n+1$, there are two separate cases to deal with:

\subsubsection*{Case 1: there is no cycle in $G$} Using Lemma \ref{lem:graph-Condorcet winning and cycle}, there exists a vertex $y^\star\in G$, such that $y^\star\rightarrow y$ for any vertex $y\neq y^\star$. Then by the induction assumption, we could seek a decomposition $\{S_1,\ldots, S_k\}$ of the directed graph $G\backslash \{y^\star\}$ as $\vert G\backslash \{y^\star\}\vert=n$. Therefore, $\{y^\star, S_1,\ldots, S_k\}$ is a desired decomposition of $G$.

\subsubsection*{Case 2: there exists a cycle in $G$} We choose $S^\star$ to be a maximum cycle, i.e. we cannot add any point $y\in G\backslash S^\star$ into $S^\star$ such that $S^\star\cup\{y\}$ also forms a cycle\footnote{A maximum cycle can be found by iteratively adding vertices outside into the current cycle until none can be added.}. Then we introduce the following crucial conclusion.

\subsubsection*{Conclusion 1}
Given $S^\star$, for any $y\in G\backslash S^\star$, it must hold either $y\rightarrow y^\star$ for any $y^\star\in S^\star$ or $y^\star\rightarrow y$ for any $y^\star\in S^\star$.

\subsubsection*{Proof of the Conclusion 1}
We prove this conclusion by contradiction. If the conclusion does not hold, there exists $y_1^\star\in S^\star$ such that $y_1^\star \rightarrow y$. Now suppose that $S^\star$ forms a cycle $y_1^\star\rightarrow\ldots\rightarrow y^\star_{\vert S^\star\vert}\rightarrow y_1^\star$. Then we can find the smallest index $k\in \{2,\ldots, \vert S^\star\vert\}$ such that $y\rightarrow y_k^\star$ as we assume that this conclusion does not hold. Therefore, we find that $y_1^\star\rightarrow\ldots\rightarrow y_{k-1}^\star\rightarrow y\rightarrow y_k^\star\rightarrow\ldots\rightarrow y_1^\star$ also forms a cycle in $G$, which causes a contradiction to the fact that $S^\star$ is the maximum cycle.

\subsubsection*{Back to the proof of Case 2}
Using the above conclusion, we can divide $G$ into $S^+, S^\star, S^-$, where for any $y^+\in S^+, y^\star\in S^\star$, $y^+\rightarrow y^\star$, and for any $y^-\in S^-, y^\star\in S^\star$, $y^\star\rightarrow y^-$. By the induction assumption, we could find decomposition in both $S_+$ and $S_{-}$, and then we complete the proof.

Now we turn to the uniqueness part. Suppose there exists two distinct decomposition $S_1,\ldots,S_k$ and $S_1^\prime,\ldots,S_l^\prime$. If $|S_1|=1$, we write $S_1=\{s_1\}$, then $s_1$ is the Condorcet winner. It then follows that $S_1^\prime=\{s_1\}=S_1$. The same argument works for the case where $|S_1^\prime|=1$. We now assume $|S_1|\geqslant 2$ and $|S_1^\prime|\geqslant 2$, and define \( f(x) \) for each vertex \( x \in G \) by
\[
f(x) = \left\vert \{ y \in G | x \rightarrow y \}\right\vert.
\]
Notice that for any decomposition \( S_1, \dots, S_k \) where $|S_1|\geqslant 2$, we have $f(x)\geqslant |G|-|S_1|+1$ for $x\in S_1$ and $f(x)\leqslant |G|-|S_1|-1$ for $x\notin S_1$. If \( |S_1| \neq |S_1'| \), without loss of generality, we assume \( |S_1| > |S_1'| \). Then, for any \( x \in S_1' \) that is not in \( S_1 \),
\[
f(x) > |G| - |S_1'| > |G| - |S_1| > f(x),
\]
which is a contradiction. Thus, \( S_1' \subseteq S_1 \), implying that there exists some \( y \in S_1 \setminus S_1' \) such that \( y \rightarrow x \) for some \( x \in S_1' \), again yielding a contradiction.
If \( |S_1| = |S_1'| \), then for any \( x \in S_1 \) that is not in \( S_1' \),
\[
f(x) > |G| - |S_1| = |G| - |S_1'| > f(x),
\]
which is again a contradiction, implying \( S_1 = S_1' \). Applying the same argument iteratively to the subgraph \( G \setminus S_1 \), we conclude \( k = l \) and \( S_i = S_i' \) for all \( i \).  
Hence, the decomposition is unique, completing the proof.
\end{proof}

\subsection{Properties of Symmetric Constant-Sum Games}
In this section we summarize several important properties the games considered in Nash learning from human feedback. These properties arise from the symmeytic and constant-sum nature of the games.
\begin{lemma}
    \label{prop: equal strategy payoff = 1/2}
    For any strategy $\pi,\pi'$, we have the following identity, 
    \[
    \cP(\pi\succ \pi')+\cP(\pi'\succ \pi)=1.
    \]
    In particular, taking $\pi'=\pi$, we obtain that 
    \[
    \cP(\pi\succ \pi)=\frac{1}{2}.
    \]
\end{lemma}
\begin{proof}[Proof of Lemma \ref{prop: equal strategy payoff = 1/2}]
    Note that 
    \[
    \begin{aligned}
       &\ \cP(\pi\succ \pi')+\cP(\pi'\succ \pi)\\
       = &\ \E_{y\sim\pi}\E_{y'\sim\pi'}\cP(y\succ y')+\E_{y\sim\pi'}\E_{y'\sim\pi}\cP(y\succ y')\\
       = &\  \E_{y\sim\pi}\E_{y'\sim\pi'}\cP(y\succ y')+\E_{y'\sim\pi'}\E_{y\sim\pi}\cP(y'\succ y)\\
       = &\ \E_{y\sim\pi}\E_{y'\sim\pi'} \left(\cP(y\succ y')+\cP(y'\succ y)\right)\\
       =&\ 1.
    \end{aligned}
    \]
    Hence, we complete the proof.
\end{proof}
\begin{lemma}
    \label{prop: min-max equivalent}
    Given strategy $\pi$, the following identity holds:
    \[
\arg\max_{\sigma}\cP(\sigma\succ \pi)=
    \arg\min_\sigma \cP(\pi\succ \sigma).
    \]
\end{lemma}
\begin{proof}[Proof of Lemma \ref{prop: min-max equivalent}]
    This identity follows directly from Lemma \ref{prop: equal strategy payoff = 1/2}.
\end{proof}
\begin{lemma}\label{lem: Nash lemma zero-sum}
    If $(\pi,\pi')$ is a Nash equilibrium, then $(\pi,\pi)$ and $(\pi',\pi')$ are both Nash equilibria.
\end{lemma}
\begin{proof}[Proof of Lemma \ref{lem: Nash lemma zero-sum}]
    It is sufficient for us to prove that $(\pi,\pi)$ is a Nash equilibrium (the conclusion for $(\pi',\pi')$ can be obtained by symmetry). By the definition of Nash equilibrium, we obtain that 
    \[
    \cP(\pi\succ \pi')=\max_{\sigma}\cP(\sigma\succ \pi')\geqslant \cP(\pi'\succ \pi')=\frac{1}{2},
    \]
    where the last identity follows from Lemma \ref{prop: equal strategy payoff = 1/2}.
    Therefore, using Lemma \ref{prop: equal strategy payoff = 1/2}, we have
    \[
    \begin{aligned}
         \cP(\pi\succ \pi)=\frac{1}{2}\leqslant &\ \cP(\pi\succ \pi')\\
         =&\ \min_{\sigma}\cP(\pi\succ \sigma)\leqslant \cP(\pi\succ \pi).
    \end{aligned}
    \]
    Thus, all inequality is equality, and we obtain that 
    \begin{equation}
        \label{eq-proof-Nash lem: Nash equilibriium proof 1}
        \pi\in
    \arg\min_{\sigma}
    \cP(\pi\succ \sigma)=\arg\max_{\sigma}\cP(\sigma\succ \pi)
    \end{equation}
     by Lemma \ref{prop: min-max equivalent}. Thus, $(\pi,\pi)$ is a Nash equilibrium by definition.
\end{proof}

\section{Proofs of Main Technical Results in Section \ref{sec:benefit}}

\subsection{Proof of Proposition \ref{prop: cycle equivalent to triangle}}\label{pf: cycle equivalent to triangle}
\begin{proof}[Proof of Proposition \ref{prop: cycle equivalent to triangle}]
    We reformulate this problem in terms of a directed graph \( G \). Specifically, we represent the \( n \) responses \( y_1, \dots, y_n \) as vertices \( 1, \dots, n \), and we introduce a directed edge \( i \rightarrow j \) if and only if \( \cP(y_i \succ y_j) > 1/2 \). Under this reformulation, a Condorcet cycle in the preference structure corresponds to a cycle in \( G \).  

Now, consider the shortest cycle in \( G \), denoted as  
\[
i_1 \rightarrow i_2 \rightarrow \cdots \rightarrow i_r \rightarrow i_1.
\]
By definition, we must have \( r \geqslant 3 \). If \( r = 3 \), the conclusion follows immediately. Otherwise, if \( i_r \rightarrow i_2 \), then we obtain a strictly shorter cycle  
\[
i_2 \rightarrow \cdots \rightarrow i_r \rightarrow i_2,
\]
contradicting the assumption that \( i_1 \rightarrow \ldots \rightarrow i_r \rightarrow i_1 \) is the shortest cycle in \( G \). Therefore, it must be that \( i_2 \rightarrow i_r \), forming a cycle  
\[
i_1 \rightarrow i_2 \rightarrow i_r \rightarrow i_1,
\]
which constitutes a Condorcet triangle in the preference structure.  
\end{proof}

\subsection{Proof of Theorem \ref{thm:preferencetoreward}}
\label{pf:preferencetoreward}
\begin{proof}[Proof of Theorem \ref{thm:preferencetoreward}]
First, suppose there exist reward values \( r_1, \ldots, r_n \) that capture the preference. If there is a Condorcet cycle, by Proposition~\ref{prop: cycle equivalent to triangle}, there is a Condorcet triangle, denoted by $\{y_i,y_j,y_k\}$, which satisfies 
\[
\cP(y_i \succ y_j) > \frac{1}{2}\,\,,\,\, \cP(y_j \succ y_k) > \frac{1}{2}\,\,,\,\, \cP(y_k \succ y_i) > \frac{1}{2}\,.
\]
Since \( r_1, \ldots, r_n \) capture the preference, we have $r_i>r_j, r_j>r_k, r_k>r_i$, which causes a contradiction. Hence, there is no Condorcet cycle. 

Now, suppose there does not exist any triple \( i, j, k \) such that
\[
\cP(y_i \succ y_j) > \frac{1}{2}\,\,,\,\, \cP(y_j \succ y_k) > \frac{1}{2}\,\,,\,\, \text{ and } \, \cP(y_k \succ y_i) > \frac{1}{2}
\]
simultaneously. We aim to find reward values \( r_1, \dots, r_n \) that satisfy $r_i>r_j$ whenever $\cP(y_i\succ y_j)>\frac{1}{2}$. 
We reformulate the problem in terms of a directed graph \( G \). Specifically, we represent the \( n \) responses \( y_1, \dots, y_n \) as vertices \( 1, \dots, n \), and we introduce a directed edge \( i \rightarrow j \) if and only if \( \cP(y_i \succ y_j) > \frac{1}{2} \). Since we assume there is no such triple satisfying the above inequalities, it follows from Proposition \ref{prop: cycle equivalent to triangle} that there does not exist any cycle in this directed graph. By Lemma \ref{lem: graph theory}, we can find a Hamiltonian path \( y_{i_1} \rightarrow y_{i_2} \rightarrow \ldots \rightarrow y_{i_n} \) in \( G \). Given that there is no cycle in the graph, it must hold that \( y_{i_k} \rightarrow y_{i_l} \) for any \( k, l \in [n] \), with \( k < l \). Thus, we can define \( r_1, \dots, r_n \) such that \( r_{i_1} > r_{i_2} > \ldots > r_{i_n} \). Then, for any \( i, j \in [n] \) with \( r_i > r_j \), we obtain that 
\[ r_i > r_j\Longleftrightarrow y_i \rightarrow y_j \Longleftrightarrow \cP(y_i \succ y_j) > \frac{1}{2}\,.\] This completes our proof. 
\end{proof}

\subsection{Proof of Theorem \ref{thm:exist condorcet}}\label{pf:exist condorcet}
 First, in Lemma \ref{lem:lowerboundexistcondorcet}, we establish that for \( n = 3 \), the probability of the existence of a Condorcet cycle  is lowered bounded by a universal constant.

\begin{lemma}\label{lem:lowerboundexistcondorcet}
Under Assumption \ref{ass:population}, suppose that there are $3$ responses $y_1,y_2,y_3$.
Then there exists a universal constant $0<c<1/2$, such that the following lower bound holds for any $m\geqslant 5$,
\[\mathbb{P}
\left(
\cP(y_1\succ y_2)>\frac{1}{2},
\cP(y_2\succ y_3)>\frac{1}{2},
\cP(y_3\succ y_1)>\frac{1}{2}
\right)\geqslant c.\]
\end{lemma}
\begin{proof}[Proof of Lemma \ref{lem:lowerboundexistcondorcet}]
First, note the following lower bound for any $m\geqslant 5$, which comes from the fact that there exists at least one case that forms a Condorcet cycle in the total $6^m$ cases sharing the same probability,
\[
\mathbb{P}
\left(
\cP(y_1\succ y_2)>\frac{1}{2},
\cP(y_2\succ y_3)>\frac{1}{2},
\cP(y_3\succ y_1)>\frac{1}{2}
\right)\geqslant \frac{1}{6^m} > 0.
\]
Using Lemma \ref{lem:luce-model-equivalent-formula} we  consider i.i.d. random variables $\{U^j_{i}\}^{j\in [m]}_{i\in [3]}\sim \mathcal{U}[0,1]$. For simplicity, we denote by $X_j^{12}:=\mathds{1}_{\{U^j_1>U^j_2\}}$,  $X_j^{23}:=\mathds{1}_{\{U^j_2>U^j_3\}}$,  $X_j^{31}:=\mathds{1}_{\{U^j_3>U^j_1\}}$, for any $j\in [m]$. then we can rewrite the preference as follows, 
\[\cP(y_1\succ y_2)=\frac{1}{m}\sum_{j=1}^mX_j^{12}, \quad \cP(y_2\succ y_3)=\frac{1}{m}\sum_{j=1}^mX_j^{23}, \quad \cP(y_3\succ y_1)=\frac{1}{m}\sum_{j=1}^mX_j^{31}.\]
We can calculate the following expectations and covariances:
\[
\begin{aligned}
    &\ \mathbb{E}\left[X_j^{12}\right]=\mathbb{E}\left[X_j^{23}\right]=\mathbb{E}\left[X_j^{31}\right]=\frac{1}{2},\\
    &\ \operatorname{Var}\left(X_j^{12}\right)=\operatorname{Var}\left(X_j^{23}\right)=\operatorname{Var}\left(X_j^{31}\right)=\frac{1}{2}-\left(
    \frac{1}{2}\right)^2=\frac{1}{4},\\
    &\ \operatorname{Cov}\left(
    X_j^{12}, X_j^{23}
    \right) = \operatorname{Cov}\left(
    X_j^{23}, X_j^{31}
    \right) = 
    \operatorname{Cov}\left(
    X_j^{31}, X_j^{12}\right) = \frac{1}{6} - \frac{1}{2}\times\frac{1}{2}=-\frac{1}{12}.
\end{aligned}
\]
By the Central Limit Theorem, we obtain
{\small\[
\left(
\frac{1}{\sqrt{m}}\sum_{j=1}^m 
\left(X_j^{12}-\frac{1}{2}\right),
\frac{1}{\sqrt{m}}\sum_{j=1}^m 
\left(X_j^{23}-\frac{1}{2}\right),
\frac{1}{\sqrt{m}}\sum_{j=1}^m 
\left(X_j^{31}-\frac{1}{2}\right)
\right)\overset{d}{\to}\mathcal{N}\left(0, \Sigma\right),
\]}
where $\Sigma=\begin{pmatrix}
    \frac{1}{4} & -\frac{1}{12} & -\frac{1}{12} \\
    -\frac{1}{12} & \frac{1}{4} & -\frac{1}{12}\\
    -\frac{1}{12} & -\frac{1}{12} & \frac{1}{4} 
\end{pmatrix}$ is the covariance matrix. Consider $(Z_1, Z_2, Z_3)
\sim\mathcal{N}(0,\Sigma)$, then we obtain that 
{\footnotesize\[\begin{aligned}
&\lim_{m\to\infty}\mathbb{P}
\left(
\cP(y_1\succ y_2)>\frac{1}{2},
\cP(y_2\succ y_3)>\frac{1}{2},
\cP(y_3\succ y_1)>\frac{1}{2}
\right)\\
&\lim_{m\to\infty}\mathbb{P}\left(
\frac{1}{\sqrt{m}}\sum_{j=1}^m \left(X_j^{12}-\frac{1}{2}\right)>0, 
\frac{1}{\sqrt{m}}\sum_{j=1}^m \left(X_j^{23}-\frac{1}{2}\right)>0, 
\frac{1}{\sqrt{m}}\sum_{j=1}^m \left(X_j^{31}-\frac{1}{2}\right)>0
\right) \\ 
=&\mathbb{P}\left(Z_1>0, Z_2>0, Z_3>0\right) = c_1 >0,
    \end{aligned}
\]}
where $c_1$ is a universal constant. Therefore, there exists a universal constant $M_1\in\mathbb{N}_+$, such that for any $m>M_1$,
\[
    \mathbb{P}\left(
    \cP(y_1\succ y_2)>\frac{1}{2}, \cP(y_2\succ y_3)>\frac{1}{2}, \cP(y_3\succ y_1)>\frac{1}{2}
    \right)\geqslant \frac{c_1}{2} > 0.
\]
Thus, for any $m\geqslant 5$, we obtain the following lower bound 
\[
\mathbb{P}\left(
    \cP(y_1\succ y_2)>\frac{1}{2}, \cP(y_2\succ y_3)>\frac{1}{2}, \cP(y_3\succ y_1)>\frac{1}{2}
    \right)\geqslant\min\left\{\frac{1}{6^{M_1}}, \frac{c_1}{2}\right\}>0,
\]
and the conclusion follows.
\end{proof}

\begin{proof}[Proof of Theorem \ref{thm:exist condorcet}]
Let us first suppose $m=3$. Then
\[
\begin{aligned}
    &\ \mathbb{P}\left(
\text{there exists a Condorcet cycle}
\right)=1-\mathbb{P}\left(\text{there does not exist any Condorcet cycle}\right)\\
\geqslant &\ 1 - \mathbb{P}\left(\bigcap_{p=1}^{\lfloor n/3 \rfloor}\left\{\text{$3p-2, 3p-1, 3p$ does not form a Condorcet cycle}\right\}\right).
\end{aligned}
\]
Using Lemma \ref{lem:luce-model-equivalent-formula}, the preference choosing strategies can be characterized by i.i.d. random variables $\{U^j_i\}^{j\in [3]}_{i\in [n]}\sim \mathcal{U}[0,1]$. For any $p=1,\ldots, \lfloor n/3 \rfloor$, the random event that $3p-2, 3p-1, 3p$ does not form a Condorcet cycle can be characterized by i.i.d. random variables $\{U^j_i\}_{i\in \{3p-2,3p-1,3p\}}^{j\in [m]}$, so they are independent random events. We have 
\[
\begin{aligned}
    &\ \mathbb{P}\left(\bigcap_{p=1}^{\lfloor n/3 \rfloor}\left\{\text{$3p-2, 3p-1, 3p$ does not form a Condorcet cycle}\right\}\right)\\
    =&\ \prod_{p=1}^{\lfloor n/3 \rfloor}\mathbb{P}\left(\text{$3p-2, 3p-1, 3p$ does not form a Condorcet cycle}\right) \\
    =&\ \prod_{p=1}^{\lfloor n/3 \rfloor}\left(1-\mathbb{P}\left(\text{$3p-2, 3p-1, 3p$ form a Condorcet cycle}\right)\right)
    =\left(\frac{17}{18}\right)^{\lfloor n/3 \rfloor}.
\end{aligned}
\]
For $m=4$, similarly we have
\[
\begin{aligned}
    &\ \mathbb{P}\left(
\text{there exists a Condorcet cycle}
\right)=1-\mathbb{P}\left(\text{there does not exist any Condorcet cycle}\right)\\
\geqslant &\ 1 - \mathbb{P}\left(\bigcap_{p=1}^{\lfloor n/4 \rfloor}\left\{\text{$4p-3, 4p-2, 4p-1,4p$ does not form a Condorcet cycle}\right\}\right).
\end{aligned}
\]
Using Lemma \ref{lem:luce-model-equivalent-formula}, the preference choosing strategies can be characterized by i.i.d. random variables $\{U^j_i\}^{j\in [4]}_{i\in [n]}\sim \mathcal{U}[0,1]$. For any $p=1,\ldots, \lfloor n/4 \rfloor$, the random event that $4p-3, 4p-2, 4p-1,4p$ does not form a Condorcet cycle can be characterized by i.i.d. random variables $\{U^j_i\}_{i\in \{4p-3, 4p-2, 4p-1,4p\}}^{j\in [m]}$, so they are independent random events. We have 
\[
\begin{aligned}
    &\ \mathbb{P}\left(\bigcap_{p=1}^{\lfloor n/4 \rfloor}\left\{\text{$4p-3, 4p-2, 4p-1,4p$ does not form a Condorcet cycle}\right\}\right)\\
    =&\ \prod_{p=1}^{\lfloor n/4 \rfloor}\mathbb{P}\left(\text{$4p-3, 4p-2, 4p-1,4p$ does not form a Condorcet cycle}\right) \\
    =&\ \prod_{p=1}^{\lfloor n/4 \rfloor}\left(1-\mathbb{P}\left(\text{$4p-3, 4p-2, 4p-1,4p$ form a Condorcet cycle}\right)\right)
    \leqslant \left(\frac{2303}{2304}\right)^{\lfloor n/4 \rfloor}.
\end{aligned}
\]
Now for $m\geqslant 5$, we have similarly
\[
\begin{aligned}
    &\ \mathbb{P}\left(
\text{there exists a Condorcet cycle}
\right)=1-\mathbb{P}\left(\text{there does not exist any Condorcet cycle}\right)\\
\geqslant &\ 1 - \mathbb{P}\left(\bigcap_{p=1}^{\lfloor n/3 \rfloor}\left\{\text{$3p-2, 3p-1, 3p$ does not form a Condorcet cycle}\right\}\right).
\end{aligned}
\]
Similar to the proof for the case $m=3$, using Lemma \ref{lem:lowerboundexistcondorcet} and by symmetry, for all $m\geqslant 5$,
\[
\begin{aligned}
    &\ \mathbb{P}\left(\bigcap_{p=1}^{\lfloor n/3 \rfloor}\left\{\text{$3p-2, 3p-1, 3p$ does not form a Condorcet cycle}\right\}\right)\\
    =&\ \prod_{p=1}^{\lfloor n/3 \rfloor}\mathbb{P}\left(\text{$3p-2, 3p-1, 3p$ does not form a Condorcet cycle}\right) \\
    =&\ \prod_{p=1}^{\lfloor n/3 \rfloor}\left(1-\mathbb{P}\left(\text{$3p-2, 3p-1, 3p$ form a Condorcet cycle}\right)\right)
    \leqslant (1-2c)^{\lfloor n/3 \rfloor} .
\end{aligned}
\]
In summary, there exists constant $c_1,c_2>0$ such that for any $n\geqslant 3$ and $m\geqslant 3$,
 \[
\begin{aligned}
    &\ \mathbb{P}(\text{there exists a Condorcet cycle})\geqslant 1-c_1e^{-c_2n}.
\end{aligned}
\]
\end{proof}

\subsection{Proof of Theorem \ref{thm:luce exist condorcet}}

\begin{lemma}\label{lem:luce exist condorcet: linear solution-n=3}
    For any $m\geqslant 3$, there exists $(x_1,\ldots,x_6)\in\mathbb{N}^6$ satisfying the system
    \begin{align*}
        \begin{cases}
            x_1+x_2+x_3+x_4+x_5+x_6=m\\
            x_1+x_2+x_3 > m/2\\
            x_2+x_5+x_6 > m/2\\
            x_3+x_4+x_6 > m/2
        \end{cases}\,,
    \end{align*}
    if and only if $m\neq 4$. 
\end{lemma}
\begin{proof}[Proof of Lemma \ref{lem:luce exist condorcet: linear solution-n=3}]
    For $m\geqslant 6$, $x_3=x_6=\lceil m/2 \rceil-1, x_2=m+2-2\lceil m/2 \rceil, x_1=x_4=x_5=0$ satisfy $x_1+\ldots+x_6=m$, and 
    \begin{align*}
    \begin{cases}
        x_1+x_2+x_3=\lceil m/2 \rceil-1+m+2-2\lceil m/2 \rceil=m+1-\lceil m/2 \rceil\geqslant m/2+1>m/2\\
        x_2+x_5+x_6=m+2-2\lceil m/2 \rceil+\lceil m/2 \rceil-1=m+1-\lceil m/2 \rceil\geqslant m/2+1>m/2\\
        x_3+x_4+x_6=2\lceil m/2 \rceil-2 > m/2
    \end{cases}\,,
    \end{align*}
    where the last inequality follows from the fact that $m\geqslant 6$. When $m=3$, it is easy to check that $x_2=x_3=x_6=1, x_1=x_4=x_5=0$ is a solution to the system. When $m=5$, $x_2=x_3=2,x_6=1,x_1=x_4=x_5=0$ is a solution to the linear system. However, when $m=4$, if $x_1,\ldots,x_6$ is a solution to the system, i.e., $x_1+\ldots+x_6=4$, and
     {\footnotesize\begin{align*}
        \begin{cases}
            x_1+x_2+x_3\geqslant 3\\
            x_2+x_5+x_6\geqslant 3\\
            x_3+x_4+x_6\geqslant 3
        \end{cases}\Longrightarrow x_i\in\{0,1\}\,\, \text{for }i=1,\cdots,6\Longrightarrow x_1=x_2=x_3=x_4=x_5=x_6=1\,,
    \end{align*}}
    which causes a contradiction to the fact that $x_1+\ldots+x_6=4$. Hence, we complete the proof.
\end{proof}

\begin{lemma}\label{lem:luce exist condorcet: cycle prob-n=3}
    For $m\geqslant 5$ and $m=3$, Luce model with $n=3$ and $\lambda_i\in [A,B]$ for $i=1,2,3$ satisfies 
    \begin{align*}
        \mathbb{P}\left(\mathcal{P}(y_1\succ y_2)>\frac{1}{2},\mathcal{P}(y_2\succ y_3)>\frac{1}{2},\mathcal{P}(y_3\succ y_1)>\frac{1}{2}\right)\geqslant \left(\frac{A^2}{6B^2}\right)^m\,.
    \end{align*}
\end{lemma}
\begin{proof}[Proof of Lemma \ref{lem:luce exist condorcet: cycle prob-n=3}]
    Notice that for any permutation $\{\pi_1,\pi_2,\pi_3\}=\{1,2,3\}$, 
    \begin{align*}
        \mathbb{P}\left(y_{\pi_1}\succ y_{\pi_2}\succ y_{\pi_3}\right)=\frac{\lambda_{\pi_1}}{\lambda_{\pi_1}+\lambda_{\pi_2}+\lambda_{\pi_3}}\cdot \frac{\lambda_{\pi_2}}{\lambda_{\pi_2}+\lambda_{\pi_3}}\geqslant \frac{A}{3B}\cdot\frac{A}{2B}=\frac{A^2}{6B^2}\,.
    \end{align*}
    Let $x_1,\ldots,x_6$ satisfy the system in Lemma \ref{lem:luce exist condorcet: linear solution-n=3}, and $N_i=\sum_{j=1}^i x_j$ for $i=1,\ldots,6$ and $N_0=0$. We define 
    \begin{align*}
        \mathcal{A}_1:=\left\{j\in [m]:\text{$j$-th individual prefers }y_3\succ y_2\succ y_1\right\}\,,\\
        \mathcal{A}_2:=\left\{j\in [m]:\text{$j$-th individual prefers }y_2\succ y_3\succ y_1\right\}\,,\\
        \mathcal{A}_3:=\left\{j\in [m]:\text{$j$-th individual prefers }y_3\succ y_1\succ y_2\right\}\,,\\
        \mathcal{A}_4:=\left\{j\in [m]:\text{$j$-th individual prefers }y_1\succ y_3\succ y_2\right\}\,,\\
        \mathcal{A}_5:=\left\{j\in [m]:\text{$j$-th individual prefers }y_2\succ y_1\succ y_3\right\}\,,\\
        \mathcal{A}_6:=\left\{j\in [m]:\text{$j$-th individual prefers }y_1\succ y_2\succ y_3\right\}\,.
    \end{align*}
    Then  
    \begin{align*}
        &\ \mathbb{P}\left(\mathcal{P}(y_1\succ y_2)>\frac{1}{2},\mathcal{P}(y_2\succ y_3)>\frac{1}{2},\mathcal{P}(y_3\succ y_1)>\frac{1}{2}\right) \\ 
        = &\  \mathbb{P}\left(\vert\mathcal{A}_3\vert+\vert\mathcal{A}_4\vert+\vert \mathcal{A}_6\vert>\frac{m}{2}, \vert\mathcal{A}_2\vert+\vert\mathcal{A}_5\vert+\vert \mathcal{A}_6\vert>\frac{m}{2}, \vert\mathcal{A}_1\vert+\vert\mathcal{A}_2\vert+\vert \mathcal{A}_3\vert>\frac{m}{2}\right) \\
        \geqslant &\ \mathbb{P}\left(\bigcap_{i=1}^6\left\{\vert \mathcal{A}_i\vert=x_i\right\}\right)\geqslant\mathbb{P}\left(\bigcap_{i=1}^6\bigcap_{j=N_{i-1}+1}^{N_i}\left\{j\in \mathcal{A}_i\right\}\right)=\prod_{i=1}^6\prod_{j=N_{i-1}+1}^{N_i}\mathbb{P}\left(j\in \mathcal{A}_i\right)\\\geqslant & \left(\frac{A^2}{6B^2}\right)^m\,.
    \end{align*}
    Hence, we complete the proof.
\end{proof}

\begin{lemma}\label{lem:luce exist condorcet: cycle prob-n=4}
    For $m=4$, Luce model with $n=4$ and $\lambda_i\in [A,B]$ for $i=1,2,3,4$ satisfies 
    \begin{align*}
        \mathbb{P}\left(\mathcal{P}(y_1\succ y_2)>\frac{1}{2},\mathcal{P}(y_2\succ y_3)>\frac{1}{2},\mathcal{P}(y_3\succ y_4)>\frac{1}{2},\mathcal{P}(y_4\succ y_1)>\frac{1}{2}\right)\geqslant \left(\frac{A^3}{24B^3}\right)^4\,.
    \end{align*}
\end{lemma}
\begin{proof}[Proof of Lemma \ref{lem:luce exist condorcet: cycle prob-n=4}]
    Notice that for any permutation $\{\pi_1,\pi_2,\pi_3,\pi_4\}=\{1,2,3,4\}$, 
    \begin{align*}
        \mathbb{P}\left(y_{\pi_1}\succ y_{\pi_2}\succ y_{\pi_3}\succ y_{\pi_4}\right)=&\ \frac{\lambda_{\pi_1}}{\lambda_{\pi_1}+\lambda_{\pi_2}+\lambda_{\pi_3}+\lambda_{\pi_4}}\cdot \frac{\lambda_{\pi_2}}{\lambda_{\pi_2}+\lambda_{\pi_3}+\lambda_{\pi_4}}\cdot\frac{\lambda_{\pi_3}}{\lambda_{\pi_3}+\lambda_{\pi_4}}\\
        \geqslant &\ \frac{A}{4B}\cdot\frac{A}{3B}\cdot\frac{A}{2B}=\frac{A^3}{24B^3}\,.
    \end{align*}
    We define 
    \begin{align*}
        \mathcal{A}_1:=\left\{j\in [m]:\text{$j$-th individual prefers }y_1\succ y_2\succ y_3\succ y_4\right\}\,,\\
        \mathcal{A}_2:=\left\{j\in [m]:\text{$j$-th individual prefers }y_4\succ y_1\succ y_2\succ y_3\right\}\,,\\
        \mathcal{A}_3:=\left\{j\in [m]:\text{$j$-th individual prefers }y_3\succ y_4\succ y_1\succ y_2\right\}\,,\\
        \mathcal{A}_4:=\left\{j\in [m]:\text{$j$-th individual prefers }y_2\succ y_3\succ y_4\succ y_1\right\}\,.
    \end{align*}
    Then  
    \begin{align*}
        &\ \mathbb{P}\left(\mathcal{P}(y_1\succ y_2)>\frac{1}{2},\mathcal{P}(y_2\succ y_3)>\frac{1}{2},\mathcal{P}(y_3\succ y_4)>\frac{1}{2},\mathcal{P}(y_4\succ y_1)>\frac{1}{2}\right) \\ 
        \geqslant &\ \mathbb{P}\left(\vert\mathcal{A}_1\vert=1,\vert\mathcal{A}_2\vert=1,\vert\mathcal{A}_3\vert=1,\vert\mathcal{A}_4\vert=1\right)\geqslant\mathbb{P}\left(\bigcap_{j=1}^4\left\{j\in\mathcal{A}_j\right\}\right)=\prod_{j=1}^4\mathbb{P}\left(j\in\mathcal{A}_j\right)\geqslant\left(\frac{A^3}{24B^3}\right)^4\,.
    \end{align*}
    Hence, we complete the proof.
\end{proof}

\begin{proof}[Proof of Theorem \ref{thm:luce exist condorcet}]
    For $m\geqslant 5$ and $m=3$, we define for $k=1,\ldots,[n/3]$, 
    \begin{align*}
        \mathcal{S}^{(3)}_k=\left\{\text{No Condorcet cycle in }y_{3k-2}, y_{3k-1}, y_{3k}\right\}\,.
    \end{align*}
    By the equivalence shown in Lemma \ref{lem:luce-model-equivalent-formula}, random events $S^{(3)}_1,\ldots,S^{(3)}_{[n/3]}$ are independent. Then, using Lemma \ref{lem:luce exist condorcet: cycle prob-n=3},
    \begin{align*}
        \mathbb{P}\left(\mathcal{S}_k^{(3)}\right)&=1-\mathbb{P}\left(\left\{\text{Condorcet cycle in }y_{3k-2}, y_{3k-1}, y_{3k}\right\}\right)\\
        &\leqslant 1-\mathbb{P}\left(\mathcal{P}\left(y_{3k-2}\succ y_{3k-1}\right)>\frac{1}{2},\mathcal{P}\left(y_{3k-1}\succ y_{3k}\right)>\frac{1}{2},\mathcal{P}\left(y_{3k}\succ y_{3k-2}\right)>\frac{1}{2}\right)\\
        &\leqslant 1-\left(\frac{A^2}{6B^2}\right)^m\,.
    \end{align*}
    Hence, we get
    \begin{align*}
        \mathbb{P}_{m,n}\left(\text{Condorcet cycle}\right)=&\ 1-\mathbb{P}\left(\text{No Condorcet cycle}\right)\\
        \geqslant &\ 1-\mathbb{P}\left(\bigcap_{k=1}^{[n/3]}\mathcal{S}_k^{(3)}\right)=1-\prod_{k=1}^{[n/3]}\mathbb{P}\left(\mathcal{S}_k^{(3)}\right)\\
        \geqslant &\ 1-\left(1-\left(\frac{A^2}{6B^2}\right)^m\right)^{[n/3]}\geqslant 1-\exp\left(-\left(\frac{A^2}{6B^2}\right)^m\cdot\frac{n-2}{3}\right)\,.
    \end{align*}
    For $m=4$, we define for $k=1,\ldots,[n/4]$,
    \begin{align*}
        \mathcal{S}^{(4)}_k=\left\{\text{No Condorcet cycle in }y_{4k-3}, y_{4k-2}, y_{4k-1}, y_{4k}\right\}\,.
    \end{align*}
    Similarly, random events $S_1^{(4)},\ldots,S_{[n/4]}^{(4)}$ are independent by the equivalence shown in Lemma \ref{lem:luce-model-equivalent-formula}. Using Lemma \ref{lem:luce exist condorcet: cycle prob-n=4}, we get 
    \begin{align*}
        &\mathbb{P}\left(\left\{\text{Condorcet cycle in }y_{4k-3}, y_{4k-2}, y_{4k-1}, y_{4k}\right\}\right)\\
        &\geqslant \mathbb{P}\left(\mathcal{P}\left(y_{4k-3}\succ y_{4k-2}\right)>\frac{1}{2},\mathcal{P}\left(y_{4k-2}\succ y_{4k-1}\right)>\frac{1}{2},\mathcal{P}\left(y_{4k-1}\succ y_{4k}\right)>\frac{1}{2},\mathcal{P}\left(y_{4k}\succ y_{4k-3}\right)>\frac{1}{2}\right)\\
        &\geqslant\left(\frac{A^3}{24B^3}\right)^4\,.
    \end{align*}
    Thus, we have 
    \begin{align*}
        \mathbb{P}\left(\mathcal{S}_k^{(4)}\right)=1-\mathbb{P}\left(\left\{\text{Condorcet cycle in }y_{4k-3}, y_{4k-2}, y_{4k-1}, y_{4k}\right\}\right)\leqslant 1-\left(\frac{A^3}{24B^3}\right)^4\,.
    \end{align*}
    Hence, we get 
    \begin{align*}
        \mathbb{P}_{m,n}\left(\text{Condorcet cycle}\right)=&\ 1-\mathbb{P}\left(\text{No Condorcet cycle}\right)\\
        \geqslant &\ 1-\mathbb{P}\left(\bigcap_{k=1}^{[n/4]}\mathcal{S}_k^{(4)}\right)=1-\prod_{k=1}^{[n/4]}\mathbb{P}\left(\mathcal{S}_k^{(4)}\right)\\
        \geqslant &\ 1-\left(1-\left(\frac{A^3}{24B^3}\right)^4\right)^{[n/4]}\geqslant 1-\exp\left(-\left(\frac{A^3}{24B^3}\right)^4\cdot\frac{n-3}{4}\right)\,.
    \end{align*}
    We complete the proof.
\end{proof}

\section{Proofs of Main Technical Results in Section \ref{sec:align}}

\subsection{Proof of Proposition \ref{prop: 0/1 Condorcet winning response}}\label{pf: 0/1 Condorcet winning response}
\begin{proof}[Proof of Proposition \ref{prop: 0/1 Condorcet winning response}]
    If there exist a Condorcet winning response $y^\star$, then for any $y\neq y^\star$, $\cP(y^\star\succ y)>1/2$, so $\cP(y\succ y^\star)<1/2$, therefore $y$ is not a Condorcet winning response. Thus there is at most one Condorcet winning response.
\end{proof}

\subsection{Proof of Corollary \ref{coro:determinstic-preferences-luce}}

\begin{lemma}\label{lem:determinstic-preferences-luce}
    Let the number of responses be $n+1$ and the number of labelers $m\geqslant 3$, and define $l=\lceil \frac{m}{2}\rceil$. 
     For $0\leqslant k<m-l$, suppose that $m-k$ individuals follows Assumption \ref{ass:luce} with bounded weights $\lambda_i\in[A,B]$, where $B\geqslant A>0$, and the remaining $k$ individuals have arbitrary deterministic preference. Then, 
    \begin{align*}
        &\ \mathbb{P}_{m,n}\left(\text{Condorcet winning response}\right)\\
        &\ \qquad\leqslant (m-k+1)n^{-\frac{m-k}{l}}(n+1)\left(\log n\right)^{\frac{2m-2k}{l}}\left(\log n+1\right)^{m-k} \left(\frac{2B}{A}\right)^{2m-2k}\\
        &\ \qquad\qquad + \frac{(m-k)(m-k-1)(n+1)B^2}{A^2n^2} \\
        &\ \qquad\qquad + e^{(m-k)\frac{B}{A}}(n+1)^{(m-k)(3+\frac{B}{A})+1}\cdot\left(\exp\left(-\frac{3(\log n)^2}{16}\right)+\exp\left(-\left(\frac{1}{2}\right)^l(\log n)^2\right)\right)\,.
    \end{align*}
\end{lemma}
\begin{proof}[Proof of Lemma \ref{lem:determinstic-preferences-luce}]
    We use $\{A_{ij}\}_{1\leqslant i,j\leqslant n,i\neq j}$ to represent the number of individuals that have deterministic preferences and prefer $y_i$ over $y_j$. By the definition, we have $A_{ij}\leqslant k$, Then, for any $w\in\{1,\ldots,n+1\}$, 
    \begin{align*}
        &\ \mathbb{P}\left(y_w\text{ is a Condorcet winner}\right)\\
        = &\ \mathbb{P}\left(\bigcap_{\substack{i=1\\ i\neq w}}^{n+1}\left\{\mathcal{P}\left(y_w\succ y_i\right)>\frac{1}{2}\right\}\right)\\
            = &\ \mathbb{P}\left(\bigcap_{\substack{i=1\\ i\neq w}}^{n+1}\left\{\frac{1}{m}\left(A_{wi}+\sum_{j=1}^{m-k}\mathbbm{1}_{\left\{U^j_i>U_w^j\right\}}\right)>\frac{1}{2}\right\}\right)\\
             \leqslant &\ \mathbb{P}\left(\bigcap_{\substack{i=1\\ i\neq w}}^{n+1}\left\{\frac{1}{m}\left(k+\sum_{j=1}^{m-k}\mathbbm{1}_{\left\{U^j_i>U_w^j\right\}}\right)>\frac{1}{2}\right\}\right)\\
             = &\ \mathbb{P}\left(\bigcap_{\substack{i=1\\ i\neq w}}^{n+1}\left\{\sum_{j=1}^{m-k}\mathbbm{1}_{\left\{U^j_i<U_w^j\right\}}<\frac{m}{2}\right\}\right)\\
             \leqslant &\ (m-k+1)n^{-\frac{m-k}{l}}\left(\log n\right)^{\frac{2m-2k}{l}}\left(\log n+1\right)^{m-k} \left(\frac{2B}{A}\right)^{2m-2k}\\
        &\ + \frac{(m-k)(m-k-1)B^2}{A^2n^2} \\
        &\ + e^{(m-k)\frac{B}{A}}(n+1)^{(m-k)(3+\frac{B}{A})}\cdot\left(\exp\left(-\frac{3(\log n)^2}{16}\right)+\exp\left(-\left(\frac{1}{2}\right)^l(\log n)^2\right)\right)\,.
    \end{align*}
    where the last step uses Lemma \ref{lem:step-one} and Lemma \ref{lem:step-two} with $l=\lceil m/2\rceil \leqslant (m+1)/2<m-k$, i.e., $1\leqslant l\leqslant m-k-1$. Hence, 
    \begin{align*}
        &\ \mathbb{P}_{m,n}\left(\text{Condorcet winning response}\right)\\
        =&\ \sum_{w=1}^{n+1}\mathbb{P}\left(\text{$y_w$ is a Condorcet winning response}\right)\\
        \leqslant &\ (m-k+1)n^{-\frac{m-k}{l}}(n+1)\left(\log n\right)^{\frac{2m-2k}{l}}\left(\log n+1\right)^{m-k} \left(\frac{2B}{A}\right)^{2m-2k}\\
        &\ + \frac{(m-k)(m-k-1)(n+1)B^2}{A^2n^2} \\
        &\ + e^{(m-k)\frac{B}{A}}(n+1)^{(m-k)(3+\frac{B}{A})+1}\cdot\left(\exp\left(-\frac{3(\log n)^2}{16}\right)+\exp\left(-\left(\frac{1}{2}\right)^l(\log n)^2\right)\right)\,.
    \end{align*}
    We complete the proof.
\end{proof}

\begin{proof}[Proof of Corollary \ref{coro:determinstic-preferences-luce}]
    It directly follows from Lemma \ref{lem:determinstic-preferences-luce}.
\end{proof}

\subsection{Proof of Proposition \ref{prop: condorcet and Condorcet winning response}}\label{pf: condorcet and Condorcet winning response}
\begin{proof}[Proof of Proposition \ref{prop: condorcet and Condorcet winning response}]
    We reformulate this problem in terms of a directed graph \( G \). Specifically, we represent the \( n \) responses \( y_1, \dots, y_n \) as vertices, and we introduce a directed edge \( y_i \rightarrow y_j \) if and only if \( \cP(y_i \succ y_j) > 1/2 \). Under this reformulation, we obtain a tournament graph, and a Condorcet cycle in the preference corresponds to a cycle in \( G \). Then this result follows from Lemma \ref{lem:graph-Condorcet winning and cycle}. 
\end{proof}

\subsection{Proof of Theorem \ref{thm:Condorcet winning response}}\label{pf:Condorcet winning}
\begin{proof}[Proof of Theorem \ref{thm:Condorcet winning response}]
    First, suppose that there exists a Nash equilibrium in pure strategies that chooses $y^\star$. Then, using Lemma \ref{lem: Nash lemma zero-sum}, we obtain that $(y^\star, y^\star)$ is a Nash equilibrium. For any $y$, we have 
    \[
    \cP(y^\star\succ y) \geqslant \min_\sigma\cP(y^\star\succ \sigma)=\cP(y^\star\succ y^\star)=\frac{1}{2}.
    \]
    Therefore, for any $y\neq y^\star$, we have $\cP(y^\star\succ y)>\frac{1}{2}$, as we have assumed that $\cP(y\succ y')\neq \frac{1}{2}$ for any $y\neq y'$. Hence, $y^\star$ is a Condorcet winning response. Conversely, suppose that there exists a Condorcet winning response $y^\star$. It is easy to see that $(y^\star, y^\star)$ is a Nash equilibrium, since 
    \[
    \begin{aligned}
        &\ \max_{\sigma}\cP(\sigma\succ y^\star)\leqslant \frac{1}{2} = \cP(y^\star\succ y^\star)\Longrightarrow y^\star=\arg\max_{\sigma}\cP(\sigma\succ y^\star).
    \end{aligned}
    \]
    Then, we prove the uniqueness of Nash equilibrium by contradiction. Suppose that there exists another Nash equilibrium $(\pi,\pi')$ where $\pi$ is not the strategy of choosing $y^*$. Using Lemma \ref{lem: Nash lemma zero-sum}, we obtain a Nash equilibrium $(\pi, \pi)$. However, note that 
    \begin{equation}\label{eq: contradiction in uniqueness of Nash with Condorcet winning response}
        \frac{1}{2}=\cP(\pi\succ \pi)=\max_{\sigma}\cP(\sigma\succ \pi)\geqslant \cP(y^\star\succ \pi)>\frac{1}{2},
    \end{equation}
    where the first identity holds by Lemma \ref{prop: equal strategy payoff = 1/2}, and the last inequality holds as $\pi$ is not the Nash equilibrium choosing the Condorcet winning response $y^\star$. Hence, \eqref{eq: contradiction in uniqueness of Nash with Condorcet winning response} causes a contradiction and we complete the proof.
\end{proof}

\subsection{Proof of Theorem \ref{thm:BTL NLHF}}\label{pf:BTL NLHF}
\begin{proof}[Proof of Theorem \ref{thm:BTL NLHF}]
    Suppose that the strategy with the highest reward is $y^\star$. Then RLHF chooses $y^\star$ through maximizing the reward. In the following, we show that $y^\star$ is a Condorcet winning response. In fact, $y^\star$ is a dominant response, i.e. $\cP(y^\star\succ y')>\cP(y\succ y')$ for any $y\neq y^\star$ and any $y'$, see Definition \ref{def:dominant strategy} for general definitions in game theory.
    From this definition we can see that a dominant response must be a Condorcet winning response, while the converse is generally false. Now,
    {\small\[
    \begin{aligned}
        \cP(y^\star\succ y')=\frac{e^{r(y^\star)}}{e^{r(y^\star)}+e^{r(y')}}
            =\frac{1}{1+e^{r(y')-r(y^\star)}}
            >\frac{1}{1+e^{r(y')-r(y)}}
            =\frac{e^{r(y)}}{e^{r(y)}+e^{r(y')}}
            =\cP(y\succ y').
    \end{aligned}
    \]}
    By Theorem \ref{thm:Condorcet winning response}, the Nash equilibrium is unique and collapses to the dominant response $y^\star$. Hence, we complete the proof.
\end{proof}

\subsection{Proof of Theorem \ref{thm:response decomposition}}\label{appendix: proof of response decomposition}
\label{append:response decomposition}
\begin{proof}[Proof of Theorem \ref{thm:response decomposition}]
 We reformulate this problem in terms of a directed graph \( G \). Specifically, we represent the \( n \) responses \( y_1, \dots, y_n \) as vertices \( 1, \dots, n \), and we introduce a directed edge \( i \rightarrow j \) if and only if \( \cP(y_i \succ y_j) > 1/2 \).
Under this reformulation, a Condorcet cycle in the preference structure corresponds to a cycle in \( G \). Therefore, this theorem is equivalent to Lemma \ref{lem: graph theory of response decomposition}.
\end{proof}
Actually, the following algorithm achieves $O(n^2)$ complexity for seeking the dominated sub-decomposition $S_1$. Repeatedly using this algorithm provides a $O(n^3)$ complexity algorithm that yields the decomposition.
\begin{algorithm}[H]
    \caption{: Decomposition}
    \label{algorithm:decomposition}
    \begin{algorithmic}[1]
        \REQUIRE responses $y_1,\ldots,y_n$ ranked by out-degree
        \STATE $S=\{y_1\}$
        \FOR{$i=1,\ldots,n$}
        \IF{$y_i\in S$}
        \FOR{$j=n,\ldots,i+1$}
        \IF{$y_j$ beats $y_i$}
        \STATE{$S=\{y_1,\ldots,y_j\}$}
        \STATE Break
        \ENDIF
        \ENDFOR
        \ELSE
        \RETURN $S$
        \ENDIF
        \ENDFOR.
    \end{algorithmic}
\end{algorithm}

\subsection{Proof of Theorem \ref{thm:support}}
\label{appendix: proof of support}
\begin{proof}[Proof of Theorem \ref{thm:support}]
    We prove by contradiction. Suppose that there exists a Nash equilibrium $(\pi,\pi')$ where the support $\pi$ is not contained in $S_1$. By Lemma \ref{lem: Nash lemma zero-sum}, $(\pi,\pi)$ is a Nash equilibrium with payoff $1/2$. First we consider the case where $\operatorname{supp}(\pi)\cap S_1= \emptyset$, then for any $y\in S_1$, $\cP(y\succ \pi)>\frac{1}{2}$, contradicting to the assumption that $(\pi,\pi)$ is a Nash equilibrium.
    
    Now if $\operatorname{supp}(\pi)\cap S_1\neq \emptyset$, let $\tilde{\pi}$ be the strategy obtained by redistributing the mass outside of $S_1$ onto $S_1$ proportional to the original mass. That is, 
    \[\tilde{\pi}_i=\frac{\pi_i}{\sum_{y_j\in S_1}\pi_j}\text{ for } y_i\in S_1\text{ and }\tilde{\pi}_i=0\text{ for }y_i\notin S_1.\]
    Next, we show that the payoff at $(\tilde{\pi},\pi)$ is higher than $1/2$, contradicting to the assumption that $(\pi,\pi)$ is a Nash equilibrium.
\begin{align*}
    &\ \cP(\tilde{\pi}\succ \pi)\\
    =&\ \sum_{ij}\tilde{\pi}_i\pi_j\cP(y_i\succ y_j)\\
    =&\ \sum_{y_i,y_j\in S_1}\tilde{\pi}_i\pi_j\cP(y_i\succ y_j)+\sum_{y_i\in S_1,y_j\in S_1^c}\tilde{\pi}_i\pi_j\cP(y_i\succ y_j) \\
    =&\ \frac{1}{2}+\sum_{y_i,y_j\in S_1}\tilde{\pi}_i\pi_j\left(\cP(y_i\succ y_j)-\frac{1}{2}\right)+\sum_{y_i\in S_1,y_j\in S_1^c}\tilde{\pi}_i\pi_j\left(\cP(y_i\succ y_j)-\frac{1}{2}\right).
\end{align*}
Noting that $\sum_{y_i,y_j\in S_1}\pi_i\pi_j\left(\cP(y_i\succ y_j)-\frac{1}{2}\right)=0$ and by anti-symmetry, we have
\begin{align*}              
  &\ \cP(\tilde{\pi}\succ \pi)\\
    =&\ \frac{1}{2}+\sum_{y_i,y_j\in S_1}(\tilde{\pi}_i-\pi_i)\pi_j\left(\cP(y_i\succ y_j)-\frac{1}{2}\right)+\sum_{y_i\in S_1,y_j\in S_1^c}\tilde{\pi}_i\pi_j\left(\cP(y_i\succ y_j)-\frac{1}{2}\right)\\
    >&\ \frac{1}{2}+\sum_{y_i,y_j\in S_1}(\tilde{\pi}_i-\pi_i)\pi_j\left(\cP(y_i\succ y_j)-\frac{1}{2}\right)
\end{align*}
as $\cP(y_i\succ y_j)>\frac{1}{2}$ for $y_i\in S_1,y_j\in S_1^c$.
For $y_i\in S_1$, $\tilde{\pi}_i=k\pi_i$ for some $k$ by the definition of $\tilde{\pi}$, we have
\[\sum_{y_i,y_j\in S_1}(\tilde{\pi}_i-\pi_i)\pi_j\left(\cP(y_i\succ y_j)-\frac{1}{2}\right)=(k-1)\sum_{y_i,y_j\in S_1}\pi_i\pi_j\left(\cP(y_i\succ y_j)-\frac{1}{2}\right)=0\]
again by anti-symmetry. Hence, we have $\cP(\tilde{\pi}\succ\pi)>1/2$, concluding our proof.  
\end{proof}

\section{Proofs of Lemmas in Section \ref{sec:upper bound proof} and Section \ref{sec:lower bound proof}}

\subsection{Proof of Lemma \ref{lem:integration-exponential}}
\begin{proof}[Proof of Lemma \ref{lem:integration-exponential}]
    When $n=2$, through a direct calculation using the density of exponential random variables, we yield
    \[
    \mathbb{P}\left(X_1>X_2\right)=\int_0^\infty \lambda_2 e^{-\lambda_2 t} e^{-\lambda_1 t}  \mathrm{~d}t = \frac{\lambda_2}{\lambda_1+\lambda_2}\int_0^\infty (\lambda_1+\lambda_2) e^{-(\lambda_1+\lambda_2) t}\mathrm{~d}t=\frac{\lambda_2}{\lambda_1+\lambda_2}\,.
    \]
    Notice that by the memoryless property of exponential random variables, for any $t\geqslant 0$, $X_i\mid X_i > t \overset{d}{=} Y_i + t$, where $Y_i\overset{d}{=}X_i\sim \operatorname{Exp}(\lambda_i)$. Then, using the induction hypothesis, 
    \[
    \begin{aligned}
        &\ \mathbb{P}\left(X_1 > \cdots > X_{n+1}\right)
        = \int_0^\infty \lambda_{n+1}e^{-\lambda_{n+1}t}\cdot \mathbb{P}\left(X_1>\cdots > X_n> t\right)\mathrm{~d}t\\
        = &\ \int_0^\infty \lambda_{n+1}e^{-\lambda_{n+1}t}\cdot \mathbb{P}\left(X_1>t, \cdots, X_n> t\right)\cdot \mathbb{P}\left(X_1>\cdots>X_n\mid X_1>t,\cdots, X_n>t\right)\mathrm{~d}t\\
        = &\ \int_0^\infty \lambda_{n+1}e^{-\left(\lambda_1+\cdots+\lambda_{n+1}\right)t}\cdot \mathbb{P}\left(Y_1>\cdots>Y_n\right)\mathrm{~d}t\\
        = &\ \frac{\lambda_{n+1}}{\lambda_1+\cdots+\lambda_{n+1}} \cdot\underbrace{\frac{\lambda_n}{\lambda_1+\cdots+\lambda_n}\cdot\frac{\lambda_{n-1}}{\lambda_1+\cdots+\lambda_{n-1}}\cdot\cdots\cdot\frac{\lambda_2}{\lambda_1+\lambda_2}}_{\text{induction hypothesis}}\,.
    \end{aligned}
    \]
    Hence, the conclusion follows from mathematical induction. 
\end{proof}

\subsection{Proof of Lemma \ref{lem:cardinality-lower-bound}}
\label{pf:no Condorcet winning}

We leverage the following lemma which characterize the most probable number of successes in $n$ independent Bernoulli trials with possibly different parameters.
\begin{lemma}[Theorem 4 in \citet{Darroch1964}]\label{lem:revision-independent-trials}
    Consider $n$ independent Bernoulli random variables $Z_1,\ldots,Z_n$ that satisfies $\mathbb{P}(Z_i=1)=p_i$ for $i=1,\ldots,n$. Let $X=Z_1+\ldots+Z_n$, $\mu:=\sum_{i=1}^n p_i$ and $\delta(\mu)=\mu-[\mu]$, then 
    \[
    \operatorname*{argmax}_{k\in\{0,\cdots,n\}}\mathbb{P}(X=k)=\begin{cases}
        \mu\,, & \delta(\mu)=0\\
        [\mu]\,, & 0\leqslant \delta(\mu)< 1/([\mu]+2)\\
        \text{one or both of }[\mu],[\mu]+1\,, & 1/([\mu]+2)\leqslant \delta(\mu)\leqslant (n-[\mu])/(n-[\mu]+1)\\
        [\mu]+1\,, & (n-[\mu])/(n-[\mu]+1)<\delta(\mu)<1
        \end{cases}\,.
    \]
\end{lemma}

\begin{proof}[Proof of Lemma \ref{lem:cardinality-lower-bound}]
    First, for $k=0$, 
    \[ \mathbb{P}\left(X=k\right)=\frac{\lambda_1}{\sum_{i=1}^{n+1}\lambda_i}\geqslant \frac{A}{B}\frac{1}{n+1}.\]
    Second, for any $k\in\{1,\ldots,n-1\}$, we have 
    \[
    \mathbb{P}\left(X=k\right)=\int_0^\infty\lambda_1 e^{-\lambda_1 u_1}\cdot \mathbb{P}\left(Z_2+\cdots+Z_{n+1}=k\mid U_1=u_1\right)\mathrm{~d}u_1\,,
    \]
    where $Z_i:=\mathds{1}\{U_i<U_1\}$ for $i=2,\ldots,n+1$. Therefore, conditional on $U_1=u_1$, $Z_2,\ldots,Z_{n+1}$ are independent Bernoulli random variables with mean $p_i(u_1):=\mathbb{P}(Z_i=1\mid U_1=u_1)=1-e^{-\lambda_i u_1}$ for any $i\in\{2,\ldots,n+1\}$. 
    We define $\mu(u_1):=\sum_{i=2}^{n+1}p_i(u_1)=n-\sum_{i=2}^{n+1}e^{-\lambda_i u_1}$, which satisfies $\mu^\prime(u_1)=\sum_{i=2}^{n+1}\lambda_i e^{-\lambda_i u_1}\in (0, Bn]$. We consider $u_1^k$ such that $\mu(u_1^k)=k$, which satisfies 
    \[
    1\leqslant n-k= \sum_{i=2}^{n+1}e^{-\lambda_i u_1^k}\leqslant ne^{-Au_1^k}\Longrightarrow e^{Au_1^k}\leqslant n\Longrightarrow u_1^k\leqslant \frac{\log n}{A}\,.
    \]
    For any $u_1\in [u_1^k, u_1^k+\frac{1}{Bn(n+2)}]$, 
    \[
    k=\mu(u_1^k)\leqslant \mu(u_1)\leqslant \mu(u_1^k)+Bn\cdot\frac{1}{Bn(n+2)}=k+\frac{1}{n+2}<k+\frac{1}{k+2}\,.
    \]
    Through Lemma \ref{lem:revision-independent-trials}, for any $u_1\in [u_1^k, u_1^k+\frac{1}{Bn(n+2)}]$,
    {\small\[
    \mathbb{P}\left(Z_2+\cdots+Z_{n+1}=k\mid U_1=u_1\right)=\max_{k\in \{0,\cdots,n\}}\mathbb{P}\left(Z_2+\cdots+Z_{n+1}=k\mid U_1=u_1\right)\geqslant \frac{1}{n+1}\,.
    \]}
    Then using the numerical inequality $1-e^{-x}\geqslant \frac{x}{2}$ for $0\leqslant x\leqslant 1$, we have the following lower bound,
    \[\begin{aligned}
        \mathbb{P}(X=k)\geqslant &\ \int_{u_1^k}^{u_1^k+\frac{1}{Bn(n+2)}}\lambda_1 e^{-\lambda_1 u_1}\cdot\frac{1}{n+1}\mathrm{~d}u_1\\
        = &\ \frac{1}{n+1}\cdot\left(e^{-\lambda_1 u_1^k}-e^{-\lambda_1\left(u_1^k+\frac{1}{Bn(n+2)}\right)}\right)\\
        = &\ \frac{1}{n+1}\cdot e^{-\lambda_1 u_1^k}(1-e^{-\lambda_1 \frac{1}{Bn(n+2)}})\\
        \geqslant &\ \frac{1}{n+1}\cdot n^{-\frac{B}{A}}\cdot \frac{A}{2Bn(n+2)}\\
        \geqslant &\  \frac{A}{2B}\cdot (n+1)^{-3-\frac{B}{A}}\,.
    \end{aligned}\]
    Finally, we consider the case where $k=n$. We define $u_1^n$ such that $\mu(u_1^n)=n-1/2$, which satisfies 
    \[
    \frac{1}{2}=\sum_{i=2}^{n+1}e^{-\lambda_i u_1^n}\leqslant ne^{-Au_1^n}\Longrightarrow e^{Au_1^n}\leqslant 2n\Longrightarrow u_1^n\leqslant \frac{\log 2+\log n}{A}\,.
    \]
    For any $u_1\in (u_1^n,+\infty)$, 
    \[
    n-\frac{1}{2}=\mu(u_1^n)<\mu(u_1)< n\,.
    \]
    Through Lemma \ref{lem:revision-independent-trials}, for any $u_1\in(u_1^n,+\infty)$, 
    \[
    \mathbb{P}\left(Z_2+\cdot+Z_{n+1}=k\mid U_1=u_1\right)=\max_{k\in\{0,\cdots,n\}}\mathbb{P}\left(Z_2+\cdots+Z_{n+1}=k\mid U_1=u_1\right)\geqslant \frac{1}{n+1}\,.
    \]
    Thus, we have the following lower bound,
    \[\begin{aligned}
        \mathbb{P}\left(X=k\right)\geqslant &\ \int_{u_1^n}^\infty \lambda_1 e^{-\lambda_1 u_1}\cdot \frac{1}{n+1}\mathrm{~d}u_1\\
        = &\ \frac{1}{n+1}\cdot e^{-\lambda_1 u_1^n}\\
        \geqslant &\ \frac{1}{n+1}\cdot e^{-Bu_1^n}\geqslant \frac{1}{n+1}\cdot (2n)^{-\frac{B}{A}}\geqslant e^{-\frac{B}{A}}\left(n+1\right)^{-1-\frac{B}{A}}\,.
    \end{aligned}\]
    In summary, we obtain for any $k\in\{0,\ldots,n\}$,
    \[
    \mathbb{P}\left(X=k\right)\geqslant e^{-\frac{B}{A}}(n+1)^{-3-\frac{B}{A}}\,.
    \]
    Hence, we complete the proof.
\end{proof}

\subsection{Proof of Lemma \ref{lem:revision-order-ineq-multi-items}}

This is proved by the following two lemmas.
\begin{lemma}\label{lem:revision-order-ineq-basic-stronger}
    For any $n\in\mathbb{N}_+$, assume that $0\leqslant x\leqslant x_{n+1}<1, 0\leqslant y_1\leqslant \ldots\leqslant y_{n+1}<1$, then 
    \[
    \left(\prod_{i=1}^n\left(1-xy_i\right)\right)\cdot (1-x_{n+1}y_{n+1})\leqslant \prod_{i=1}^{n+1}\left(1-\left(\frac{nx+x_{n+1}}{n+1}\right)y_i\right)\,.
    \]
\end{lemma}
\begin{proof}[Proof of Lemma \ref{lem:revision-order-ineq-basic-stronger}]
    We define $nx+x_{n+1}=(n+1)C$ where $x\leqslant C\leqslant x_{n+1}$. Then, we consider the following function 
    \[
    f(t):=\left(\prod_{i=1}^n \left(1-ty_i\right)\right)\cdot \left(1-\left((n+1)C-nt\right)y_{n+1}\right)\,.
    \]
    Notice that 
    {\footnotesize\[\begin{aligned}
        f^\prime(t)=&\ \prod_{i=1}^n (1-ty_i)\cdot ny_{n+1} - \sum_{i=1}^n\frac{\prod_{j=1}^n (1-ty_i)}{1-ty_i} y_i\cdot\left(1-\left((n+1)C-nt\right)y_{n+1}\right)\\
        = &\ \prod_{i=1}^n (1-ty_i)\cdot\left[ny_{n+1}-\sum_{i=1}^n\frac{y_i}{1-ty_i}+\sum_{i=1}^n\frac{(n+1)Cy_iy_{n+1}}{1-ty_i}-\sum_{i=1}^n\frac{ny_iy_{n+1}t}{1-ty_i}\right]\\
        = &\ \prod_{i=1}^n (1-ty_i)\cdot\left[\sum_{i=1}^n\frac{y_{n+1}(1-ty_i)}{1-ty_i}-\sum_{i=1}^n\frac{y_i}{1-ty_i}+\sum_{i=1}^n\frac{(n+1)Cy_iy_{n+1}}{1-ty_i}-\sum_{i=1}^n\frac{ny_iy_{n+1}t}{1-ty_i}\right]\\
        = &\ \prod_{i=1}^n (1-ty_i)\cdot\left[\sum_{i=1}^n\frac{y_{n+1}-y_i}{1-ty_i}-\sum_{i=1}^n\frac{y_{n+1}y_it}{1-ty_i}+\sum_{i=1}^n\frac{(n+1)Cy_iy_{n+1}}{1-ty_i}-\sum_{i=1}^n\frac{ny_iy_{n+1}t}{1-ty_i}\right]\\
        = &\ \prod_{i=1}^n (1-ty_i)\cdot\left[\sum_{i=1}^n\frac{y_{n+1}-y_i}{1-ty_i}+\sum_{i=1}^n\frac{(n+1)y_iy_{n+1}(C-t)}{1-ty_i}\right]\,.
    \end{aligned}\]}
    Therefore, $f^\prime(t)\geqslant 0$ for any $x\leqslant t\leqslant C$, we have $f(x)\leqslant f(C)$, i.e., 
    \[
    \left(\prod_{i=1}^n\left(1-xy_i\right)\right)\cdot (1-x_{n+1}y_{n+1})\leqslant \prod_{i=1}^{n+1}\left(1-\left(\frac{nx+x_{n+1}}{n+1}\right)y_i\right)\,.
    \]
    Hence, we conlude our proof.
\end{proof}

\begin{lemma}\label{lem:revision-order-ineq-two-items}
    For any $n\in\mathbb{N}_+$, assume that $0\leqslant x_1\leqslant\ldots\leqslant x_n<1$ and $0\leqslant y_1\leqslant\ldots\leqslant y_n<1$, then 
    \[\prod_{i=1}^n(1-x_iy_i)\leqslant \prod_{i=1}^n\left(1-\left(\frac{1}{n}\sum_{i=1}^n x_i\right)y_i\right)\leqslant \left(1-\left(\frac{1}{n}\sum_{i=1}^n x_i\right)\left(\frac{1}{n}\sum_{i=1}^n y_i\right)\right)^n\,.\]
\end{lemma}
\begin{proof}[Proof of Lemma \ref{lem:revision-order-ineq-two-items}]
First, we prove the first inequality using mathematical induction. When $n=1$, the inequality is trivial. Suppose that the inequality is true for $n$, that is,
    \[
    \prod_{i=1}^n(1-x_iy_i)\leqslant \prod_{i=1}^n\left(1-\left(\frac{1}{n}\sum_{i=1}^n x_i\right)y_i\right)\leqslant \left(1-\left(\frac{1}{n}\sum_{i=1}^n x_i\right)\left(\frac{1}{n}\sum_{i=1}^n y_i\right)\right)^n\,,
    \]
    then for $n+1$ we have 
    \[\begin{aligned}
        \prod_{i=1}^{n+1}(1-x_iy_i)\leqslant &\ \prod_{i=1}^n\left(1-\left(\frac{1}{n}\sum_{i=1}^n x_i\right)y_i\right)\left(1-x_{n+1}y_{n+1}\right)\\
        \overset{\ref{lem:revision-order-ineq-basic-stronger}}{\leqslant} &\ \prod_{i=1}^{n+1}\left(1-\left(\frac{1}{n+1}\sum_{i=1}^{n+1}x_i\right)y_i\right).
    \end{aligned}\]
    Using similar arguments as above on $\{y_i\}_{i=1}^{n+1}$, we yield
    \[\prod_{i=1}^{n+1}\left(1-\left(\frac{1}{n+1}\sum_{i=1}^{n+1}x_i\right)y_i\right)\leqslant  \left(1-\left(\frac{1}{n+1}\sum_{i=1}^{n+1} x_i\right)\left(\frac{1}{n+1}\sum_{i=1}^{n+1} y_i\right)\right)^{n+1}\,.\]
    Hence, we complete the proof by induction.
\end{proof}

\begin{proof}[Proof of Lemma \ref{lem:revision-order-ineq-multi-items}]
    We prove this inequality by repeatedly using Lemma \ref{lem:revision-order-ineq-two-items}. For $k<l$, denote $\prod_{j=k}^lx_i^j$ by $T_i^{k:l}$. Then we have
    \[\begin{aligned}
        &\prod_{i=1}^n\left(1-\prod_{j=1}^l x_i^j\right)\\ = &\ \prod_{i=1}^n\left(1-x_i^1 T_i^{2:l}\right)\overset{\ref{lem:revision-order-ineq-two-items}}{\leqslant} \prod_{i=1}^n \left(1-\left(\frac{1}{n}\sum_{i=1}^n x_i^1\right)T_i^{2:l}\right)\\
        = &\ \prod_{i=1}^n \left(1-x_i^2\left(\frac{1}{n}\sum_{i=1}^n x_i^1\right)T_i^{3:l}\right)\overset{\ref{lem:revision-order-ineq-two-items}}{\leqslant}\prod_{i=1}^n\left(1-\left(\frac{1}{n}\sum_{i=1}^n x_i^1\right)\left(\frac{1}{n}\sum_{i=1}^n x_i^2\right)T_i^{3:l}\right)\\
        \leqslant &\ \cdots\cdots \leqslant \prod_{i=1}^n \left(1-\prod_{j=1}^l\left(\frac{1}{n}\sum_{i=1}^n x_i^j\right)\right)=\left(1-\prod_{j=1}^l\left(\frac{1}{n}\sum_{i=1}^n x_i^j\right)\right)^n\,.
    \end{aligned}\]
\end{proof}

\subsection{Proof of Lemma \ref{lem:step-one}}

To prove Lemma \ref{lem:step-one}, we will use the Bernstein inequality:
\begin{lemma}[Bernstein Inequality]\label{lem:bernstein}
    Let $X_1,\ldots,X_n$ be independent zero mean random variables. Suppose that $\vert X_i\vert\leqslant M$ for all $i$, then for any positive $t$,
    \[
    \mathbb{P}\left(\sum_{i=1}^n X_i\geqslant t\right)\leqslant \exp\left(-\frac{t^2/2}{\sum_{i=1}^n\mathbb{E}\left[X_i^2\right]+Mt/3}\right)\,.
    \]
\end{lemma}
Specifically, if $Z_1,\ldots,Z_n$ are independent Bernoulli random variables with $\mathbb{P}(Z_i=1)=p_i$ for all $i\in [n]$. Then, from Lemma \ref{lem:bernstein} we get 
    \[\begin{aligned}
        \mathbb{P}\left(\sum_{i=1}^n Z_i-\sum_{i=1}^n p_i\geqslant t\right)\leqslant \exp\left(-\frac{t^2/2}{t/3+\sum_{i=1}^n\mathrm{Var}(Z_i)}\right)\leqslant\exp\left(-\frac{t^2/2}{t/3+\sum_{i=1}^n p_i}\right)\,.
    \end{aligned}\]

\begin{proof}[Proof of Lemma \ref{lem:step-one}]
	
    Define the event $\mathcal{A}:=\left\{N_1\cap\cdots\cap N_l=\emptyset\right\}$ and $$\boldsymbol{U}_1=(U_1^1,\ldots,U_1^m)^\top\in\mathbb{R}^m.$$ Now we define the critical set $\boldsymbol{B}\subseteq\mathbb{R}^m$ by
\[
\boldsymbol{B}:=\left\{\boldsymbol{u}_1=(u_1^1,\cdots,u_1^m)^\top\in\mathbb{R}^m: \prod_{j=1}^l\left(\sum_{i=2}^{n+1}\left(1-e^{-\lambda_i u_1^j}\right)\right)\geqslant \left(\frac{1}{2}\right)^l n^{l-1}(\log n)^2\right\}\,.
\]
For simplicity, we denote by $f(u)=\sum_{i=2}^{n+1}\left(1-e^{-\lambda_i u}\right)$ in the following.

Notice that 
\begin{equation}
\begin{aligned}
    &\ \mathbb{P}\left(\mathcal{A}\mid X_1=x_1,\cdots, X_m=x_m\right)\\
    = &\ \mathbb{P}\left(\mathcal{A},\boldsymbol{U}_1\in\boldsymbol{B}^c\mid X_1=x_1,\cdots, X_m=x_m\right)+\mathbb{P}\left(\mathcal{A},\boldsymbol{U}_1\in\boldsymbol{B}\mid X_1=x_1,\cdots, X_m=x_m\right)\\
    \leqslant &\ \mathbb{P}\left(\boldsymbol{U}_1\in\boldsymbol{B}^c\mid X_1=x_1,\cdots, X_m=x_m\right)+\frac{\mathbb{P}\left(\mathcal{A},\boldsymbol{U}_1\in\boldsymbol{B},X_1=x_1,\cdots, X_m=x_m\right)}{\mathbb{P}\left(X_1=x_1,\cdots, X_m=x_m\right)}\\
    \leqslant &\ \frac{\mathbb{P}\left(\boldsymbol{U}_1\in\boldsymbol{B}^c, X_1=x_1,\cdots, X_m=x_m\right)}{\mathbb{P}\left(X_1=x_1,\cdots, X_m=x_m\right)}+\frac{\mathbb{P}\left(\mathcal{A},\boldsymbol{U}_1\in\boldsymbol{B}\right)}{\mathbb{P}\left(X_1=x_1,\cdots, X_m=x_m\right)}\,.
\end{aligned}
\end{equation}

By Lemma \ref{lem:cardinality-lower-bound}, we can lower bound the denominator by
\begin{equation}\label{eq:luce-model-goal-step1-result1}
    \mathbb{P}\left(X_1=x_1,\cdots, X_m=x_m\right)=\prod_{j=1}^m\mathbb{P}\left(X_j=x_j\right)\geqslant e^{-m\frac{B}{A}}(n+1)^{-m(3+\frac{B}{A})}\,.
\end{equation}

Now we upper bound the numerator in the first term when $x_1\dots x_l\geqslant n^{l-1}(\log n)^2$. For any $\boldsymbol{u}_1\in\boldsymbol{B}^c$, 
\[
\prod_{j=1}^l f\left(u_1^j\right)<\left(\frac{1}{2}\right)^ln^{l-1}(\log n)^2\Longrightarrow \prod_{j=1}^l \frac{x_j}{2f\left(u_1^j\right)}>\frac{n^{l-1}(\log n)^2}{n^{l-1}(\log n)^2}=1\,.
\]
Therefore, there exists $j^\star\in\{1,\ldots,l\}$ such that $x_{j^\star}>2f\left(u_1^{j^\star}\right)$. Hence,

\begin{equation*}
    \begin{aligned}
    \mathbb{P}\left(X_{j^\star}=x_{j^\star}\mid U_1^{j^\star}=u_1^{j^\star}\right)\overset{(i)}{\leqslant}&\  \mathbb{P}\left(X_{j^\star}-f\left(u_1^{j^\star}\right)\geqslant\frac{x_{j^\star}}{2}\mid U_1^{j^\star}=u_1^{j^\star}\right)\\
    \overset{(ii)}{\leqslant} &\ \exp\left(-\frac{x_{j^\star}^2/8}{x_{j^\star}/6+f\left(u_1^{j^\star}\right)}\right)\\
    \overset{(iii)}{\leqslant} &\ \exp\left(-\frac{x_{j^\star}^2/8}{x_{j^\star}/6+x_{j^\star}/2}\right)\\
    = &\ \exp\left(-\frac{3x_{j^\star}}{16}\right)\overset{(iv)}{\leqslant}\exp\left(-\frac{3(\log n)^2}{16}\right)\,.
\end{aligned}
\end{equation*}
where $(i)$ uses $x_{j^\star}>x_{j^\star}/2+f(u_1^{j^\star})$, $(ii)$ applies Bernstein inequality (Lemma \ref{lem:bernstein}) noticing $X_{j^\star}$ is the sum of $n$ independent Bernoulli r.v., with mean $f\left(u_1^{j^\star}\right)$ given $U_1^{j^\star}=u_1^{j^\star}$, $(iii)$ uses $f(u_1^{j^\star})<x_{j^\star}/2$, and $(iv)$ uses $x_{j^\star}\geqslant \min_{j\in [l]} x_j\geqslant (\log n)^2$, which follows from 
\[
n^{l-1}\min_{j\in [l]} x_j\geqslant x_1\cdots x_l\geqslant n^{l-1}(\log n)^2\Longrightarrow \min_{j\in [l]} x_j\geqslant (\log n)^2\,.
\]
Hence, 
 {\footnotesize\begin{equation*}
    \begin{aligned}
    \mathbb{P}\left(X_1=x_1,\cdots,X_m=x_m\mid \boldsymbol{U}_1=\boldsymbol{u}_1\right)\leqslant &\ \mathbb{P}\left(X_{j^\star}=x_{j^\star}\mid \boldsymbol{U}_1=\boldsymbol{u}_1\right)\\
    =&\ \mathbb{P}\left(X_{j^\star}=x_{j^\star}\mid U_1^{j^\star}=u_1^{j^\star}\right)\leqslant \exp\left(-\frac{3(\log n)^2}{16}\right)\,.
\end{aligned}
\end{equation*}}
Then,
\begin{equation}\label{eq:luce-model-goal-step1-result2-2-improved}
    \begin{aligned}
&\ \mathbb{P}\left(\boldsymbol{U}_1\in\boldsymbol{B}^c,X_1=x_1,\cdots,X_m=x_m\right)\\
= &\ \int_{\boldsymbol{u}_1\in\boldsymbol{B}^c}\mathbb{P}\left(X_1=x_1,\cdots,X_m=x_m\mid \boldsymbol{U}_1=\boldsymbol{u}_1\right)\mathrm{~d}\mathbb{P}_{\boldsymbol{U}_1}(\boldsymbol{u}_1)\\
\leqslant &\ \exp\left(-\frac{3(\log n)^2}{16}\right)\cdot \mathbb{P}\left(\boldsymbol{U}_1\in\boldsymbol{B}^c\right)\leqslant \exp\left(-\frac{3(\log n)^2}{16}\right)\,. 
\end{aligned}    
\end{equation}

Now we upper bound the numerator in the second term when $x_1\dots x_l\geqslant n^{l-1}(\log n)^2$.
We examine the probability $\mathbb{P}\left(\mathcal{A}\mid \boldsymbol{U}_1=\boldsymbol{u}_1\right)$ for any realization $\boldsymbol{u}_1$ of $\boldsymbol{U}_1\in\boldsymbol{B}$. Notice that 
\begin{align*}
    \mathbb{P}\left(\mathcal{A}\mid \boldsymbol{U}_1=\boldsymbol{u}_1\right)= &\ \mathbb{P}\left(\bigcap_{i=2}^{n+1}\left\{i\notin N_1\cap\cdots\cap N_l\right\}\ \bigg\vert\ \boldsymbol{U}_1=\boldsymbol{u}_1\right)\notag\\
    = &\ \prod_{i=2}^{n+1}\mathbb{P}\left(i\notin N_1\cap\cdots\cap N_l\mid\boldsymbol{U}_1=\boldsymbol{u}_1\right)\notag\\
    =&\ \prod_{i=2}^{n+1}\left(1-\mathbb{P}\left(i\in N_1\cap\cdots\cap N_l\mid \boldsymbol{U}_1=\boldsymbol{u}_1\right)\right)\label{eq:luce-model-goal-step1-result3-1}\\
    = &\  \prod_{i=2}^{n+1}\left(1-\prod_{j=1}^l\mathbb{P}\left(i\in N_j\mid \boldsymbol{U}_1=\boldsymbol{u}_1\right)\right)\notag\\
    = &\ \prod_{i=2}^{n+1}\left(1-\prod_{j=1}^l\left(1-e^{-\lambda_i u_1^j}\right)\right)\overset{\ref{lem:revision-order-ineq-multi-items}}{\leqslant} \left(1-\prod_{j=1}^l\left(\frac{1}{n}\sum_{i=2}^{n+1}\left(1-e^{-\lambda_i u_1^j}\right)\right)\right)^n\notag\,.
\end{align*}
Recall that the realization $\boldsymbol{u}_1\in\boldsymbol{B}$, 
\begin{equation*}
    \begin{aligned}
        &\left(1-\prod_{j=1}^l\left(\frac{1}{n}\sum_{i=2}^{n+1}\left(1-e^{-\lambda_i u_1^j}\right)\right)\right)^n\\= &\ \left(1-\frac{1}{n^l}\cdot \prod_{j=1}^l\left(\sum_{i=2}^{n+1}\left(1-e^{-\lambda_i u_1^j}\right)\right)\right)^n\\
        \leqslant &\ \left(1-\frac{1}{n^l}\cdot \left(\frac{1}{2}\right)^ln^{l-1}(\log n)^2\right)^n\\
        = &\ \left(1-\left(\frac{1}{2}\right)^l\cdot\frac{(\log n)^2}{n}\right)^n\leqslant \exp\left(-\left(\frac{1}{2}\right)^l(\log n)^2\right)\,.
    \end{aligned}
\end{equation*}
Then, 
\begin{equation}\label{eq:luce-model-goal-step1-result3}
    \begin{aligned}
\mathbb{P}\left(\mathcal{A},\boldsymbol{U}_1\in\boldsymbol{B}\right)=&\ \int_{\boldsymbol{u}_1\in\boldsymbol{B}}\mathbb{P}\left(\mathcal{A}\mid\boldsymbol{U}_1=\boldsymbol{u}_1\right)\mathrm{~d}\mathbb{P}_{\boldsymbol{U}_1}(\boldsymbol{u}_1)\\
\leqslant &\ \exp\left(-\left(\frac{1}{2}\right)^l(\log n)^2\right)\cdot\mathbb{P}\left(\boldsymbol{U}_1\in\boldsymbol{B}\right)\leqslant \exp\left(-\left(\frac{1}{2}\right)^l(\log n)^2\right)\,. 
\end{aligned}
\end{equation}

Combining \eqref{eq:luce-model-goal-step1-result1}, \eqref{eq:luce-model-goal-step1-result2-2-improved}, and \eqref{eq:luce-model-goal-step1-result3} together, for any $x_1\dots x_l\geqslant n^{l-1}(\log n)^2$, we have
\[\begin{aligned}
    &\ \mathbb{P}\left(N_{1}\cap\cdots\cap N_{l}=\emptyset\mid X_1=x_1,\cdots,X_m=x_m\right)\\
    &\ \quad\quad \leqslant e^{m\frac{B}{A}}(n+1)^{m(3+\frac{B}{A})}\cdot\left(\exp\left(-\frac{3(\log n)^2}{16}\right)+\exp\left(-\left(\frac{1}{2}\right)^l(\log n)^2\right)\right)\,.
\end{aligned}\]
Hence, 
\[\begin{aligned}
    &\ \mathbb{P}\left(N_{1}\cap\cdots\cap N_{l}=\emptyset\mid X_{1}\cdots X_{l}\geqslant n^{l-1}(\log n)^2\right) \\
    &\ \quad\quad \leqslant e^{m\frac{B}{A}}(n+1)^{m(3+\frac{B}{A})}\cdot\left(\exp\left(-\frac{3(\log n)^2}{16}\right)+\exp\left(-\left(\frac{1}{2}\right)^l(\log n)^2\right)\right)\,.
\end{aligned}\]

\end{proof}

\subsection{Proof of Lemma \ref{lem:luce Condorcet winning-bad profile prob}}

We first state a technical lemma that controls the moment of each $1/X_j$ when $X_j>0$. This shows that each $X_j$ is typically not too large.
\begin{lemma}\label{lem:luce moment}
Suppose $U_i\sim \operatorname{Exp}(\lambda_i)$ independently for $i=1,\dots,n+1$, where each $\lambda_i\in [A,B]$ and $B\geqslant A>0$.
Let $N=\{i\in \{2,\cdots,n+1\}, U_i<U_1\}$ and $X=|N|$. Then
     \[\mathbb{E}\left[\frac{1}{X}\bigg\vert X>0\right]\leqslant \left(\frac{2B}{A}\right)^2\frac{1+\log n}{n}\,.\]
\end{lemma}
\begin{proof}[Proof of Lemma \ref{lem:luce moment}]
First, we can write down the probability distribution of $X$.
For any $k=0,\ldots,n$, we have
 {\footnotesize\begin{align*}
&\mathbb{P}\left(X=k\right)\\= &\sum_{\substack{i_1,\cdots,i_k\in \{2,\cdots, n+1\}}}\int_0^\infty \lambda_1  e^{-\lambda_1 t}\cdot \prod_{j\in \{2,\cdots, n+1\}\backslash \{i_1,\cdots, i_k\}}e^{-\lambda_{j}t}\cdot \prod_{i\in \{i_1,\cdots, i_k\}}\left(1-e^{-\lambda_i t}\right)\mathrm{~d}t.
\end{align*}}
In particular, 
\[\mathbb{P}\left(X=0\right)=\int_0^\infty \lambda_1  e^{-\lambda_1 t}\cdot \prod_{i\in \{2,\cdots, n+1\}}e^{-\lambda_i t}\mathrm{~d}t=\frac{\lambda_1}{\sum_{i=1}^{n+1}\lambda_i}\leqslant \frac{B}{B+nA}.\]
Consider the generating function
\[
G_{X}(t):=\mathbb{E}t^{X}=\lambda_1\int_0^\infty e^{-\lambda_1 u}\prod_{i=2}^{n+1}(t+(1-t)e^{-\lambda_i u})\mathrm{~d}u.
\]

    Notice that 
    \[\begin{aligned}
        \mathbb{E}\left[\frac{1}{X}\bigg\vert X>0\right]=&\sum_{k=1}^n\frac{1}{k}\cdot\frac{\mathbb{P}(X=k)}{\mathbb{P}(X\neq 0)}\\=& \ \frac{1}{1-\mathbb{P}(X=0)}\int_0^1\sum_{k=1}^n\mathbb{P}(X=k)\cdot t^{k-1}\mathrm{~d}t\\
        = &\ \frac{1}{1-\mathbb{P}(X=0)}\int_0^1\frac{G_X(t)-\mathbb{P}(X=0)}{t}\mathrm{~d}t\\
        = &\ \frac{1}{1-\mathbb{P}(X=0)}\left(\underbrace{\int_0^{\frac{1}{n}}}_{:=\Delta_1}+\underbrace{\int_{\frac{1}{n}}^1}_{:=\Delta_2}\right)\frac{G_X(t)-\mathbb{P}(X=0)}{t}\mathrm{~d}t\,,
    \end{aligned}\]
    
    Note that for any $0\leqslant t\leqslant 1$,
    \[G_X(t)=\sum_{k=0}^n \mathbb{P}(X=k)t^k\leqslant \mathbb{P}(X=0)+t\,,\]
    Thus, $\Delta_1\leqslant\frac{1}{n}$. 
    
    For $\Delta_2$, notice that for any $0\leqslant t\leqslant 1$,
    \[\begin{aligned}
        G_X(t)=&\ \int_0^\infty\lambda_1 e^{-\lambda_1 u}\prod_{i=2}^{n+1}\left( t+(1-t) e^{-\lambda_i u}\right)\mathrm{~d}u\\
        \leqslant &\ \int_0^\infty B e^{-A u}\prod_{i=2}^{n+1}\left(t+(1-t) e^{-A u}\right)\mathrm{~d}u\\
        = &\ \frac{B}{A}\cdot \int_0^\infty A e^{-A u}\prod_{i=2}^{n+1}\left(t+(1-t) e^{-A u}\right)\mathrm{~d}u\\
        = &\ - \frac{B}{A}\cdot \int_0^\infty \prod_{i=2}^{n+1}\left(t+(1-t) e^{-A u}\right)\mathrm{~d}\left(e^{-A u}\right)\\
        = &\ \frac{B}{A}\cdot \int_0^1 \left(t+(1-t)z\right)^n\mathrm{~d}z\,.
    \end{aligned}\]
    Notice that 
    \[\begin{aligned}
        \int_0^1 \left(t+(1-t)z\right)^n\mathrm{~d}z=&\ \frac{1}{1-t}\int_0^1\left(t+(1-t)z\right)^n\mathrm{~d}\left(\left(t+(1-t)z\right)\right)\\
        = &\ \frac{1}{1-t}\int_t^1 x^n\mathrm{~d}x=\frac{1}{n+1}\cdot \frac{1-t^{n+1}}{1-t}\\
        = &\ \frac{1}{n+1}\left(1+t+\cdots+t^n\right)\,.
    \end{aligned}\]
    Therefore, we have 
    \[\begin{aligned}
        G_X(t)\leqslant \frac{1}{n+1}\cdot\frac{B}{A}\cdot\left(1+t+\cdots+t^n\right)\,.
    \end{aligned}\]
    Thus, 
    \[\begin{aligned}
        \Delta_2\leqslant &\ \int_{\frac{1}{n}}^1\frac{1}{n+1}\cdot\frac{B}{A}\cdot \left(\frac{1}{t}+1+\cdots+t^{n-1}\right)\mathrm{~d}t\\
        \leqslant &\ \frac{B}{A}\cdot\frac{1}{n+1}\cdot\left(\log n+1+\frac{1}{2}+\cdots+\frac{1}{n}\right)\\
        \leqslant &\ \frac{B}{A}\cdot\frac{1}{n+1}\cdot\left(2\log n+1\right)
       \,.
    \end{aligned}\]
    Hence, 
     \[\begin{aligned}
        \mathbb{E}\left[\frac{1}{X}\bigg\vert X>0\right]
        = &\ \frac{\Delta_1+\Delta_2}{1-\mathbb{P}(X=0)} \\
        \leqslant  &\ \left(1+\frac{B}{nA}\right)\left(\frac{1}{n}+\frac{B}{A}\cdot\frac{1}{n+1}\cdot\left(2\log n+1\right)\right)\\
        \leqslant  &\ \left(\frac{2B}{A}\right)^2\frac{1+\log n}{n}\,,
    \end{aligned}\]
    we complete the proof.
\end{proof}

\begin{proof}[Proof of Lemma \ref{lem:luce Condorcet winning-bad profile prob}]
    We have by Markov's inequality
    \begin{align*}
        &\mathbb{P}\left(\prod_{j=1}^m X_j\leqslant n^{\alpha m}\left(\log n\right)^\beta\Big\vert X_1\cdots X_m>0\right)\\
        =&  \mathbb{P}\left(
        \sum_{j=1}^m -\log X_j\geqslant -\alpha m\log n - \beta\log\log n \Big\vert X_1\cdots X_m>0
        \right)\\
        \leqslant &\ e^{\alpha m\log n + \beta\log \log n}\cdot\mathbb{E}\left[e^{-\sum_{j=1}^m\log X_j}\Big\vert X_1\cdots X_m>0\right]\\
        =&\ n^{\alpha m}\left(\log n\right)^\beta\cdot\prod_{j=1}^m\mathbb{E}\left[\frac{1}{X_j}\Big\vert X_j>0\right],
    \end{align*}
    Using Lemma \ref{lem:luce moment}, we have 
    \[
    \begin{aligned}
    \mathbb{P}\left(\prod_{j=1}^m X_j\leqslant n^{\alpha m}\left(\log n\right)^\beta\Big\vert X_1\cdots X_m>0\right)
       \leqslant &\ n^{\alpha m}\left(\log n\right)^\beta\cdot\prod_{j=1}^m\left(\frac{2B}{A}\right)^2\frac{1+\log n}{n}\\=&\ \frac{n^{\alpha m}}{n^m} \left(\log n\right)^{\beta}\left(\log n+1\right)^{m}\left(\frac{2B}{A}\right)^{2m}.
       \end{aligned}
    \]
    Hence, we complete the proof.
\end{proof}

\subsection{Proof of Lemma \ref{lem:step-two}}

We first state a technical lemma that controls the probability that some $X_j$ is equal to zero.
\begin{lemma}
    \label{lem:luce non-zero-condition}
    Suppose $U_i^j\sim \operatorname{Exp}(\lambda_i)$ independently for $i=1,\dots,n+1$ and $j=1,\dots,m$, where each $\lambda_i\in [A,B]$ with $B\geqslant A>0$. Let $N_j=\{i\in \{2,\cdots,n+1\}:U_i^j<U_1^j\}$ and $X_j=|N_j|$, then 
    \begin{align}\label{eq:luce non-zero-condition-1}
        \mathbb{P}\left(X_j=0\right)\leqslant \frac{B}{nA+B}\,,\,\text{ for all }j=1,\cdots,m\,.
    \end{align}
    Furthermore, consider their order statistics $\{X_{(j)}\}_{j\in [m]}$ with $X_{(1)}\geqslant \cdots \geqslant X_{(m)}$, we have 
    \begin{align*}
        \mathbb{P}\left(X_{(m-1)}=0\right)\leqslant \frac{m(m-1)B^2}{(nA+B)^2}\,.
    \end{align*}
\end{lemma}
\begin{proof}[Proof of Lemma \ref{lem:luce non-zero-condition}]
The inequality \eqref{eq:luce non-zero-condition-1} directly follows from 
\begin{align*}
    &\mathbb{P}\left(X_j=0\right)=\mathbb{P}\left(\bigcap_{i=2}^{n+1}\left\{U_i^j\geqslant U_1^j\right\}\right)\\=&\ \int_0^\infty \lambda_1e^{-\lambda_1 u_1^j}\cdot \prod_{i=2}^{n+1} e^{-\lambda_i u_1^j}\mathrm{~d}u_1^j\\
    = &\ \int_0^\infty \lambda_1 e^{-\left(\sum_{i=1}^{n+1}\lambda_i\right)u_1^j}\mathrm{~d}u_1^j=\frac{\lambda_1}{\sum_{i=1}^{n+1}\lambda_i}\leqslant \frac{\lambda_1}{\lambda_1+nA}\leqslant \frac{B}{nA+B}\,.
\end{align*}
Thus, for any $1\leqslant \ell,k\leqslant n$ and $\ell\neq k$, 
\begin{align*}
    \mathbb{P}\left(X_\ell=X_k=0\right)=\mathbb{P}\left(X_\ell=0\right)\mathbb{P}\left(X_k=0\right)\leqslant \frac{B^2}{(nA+B)^2}\,.
\end{align*}
Hence, through union bounds, we get 
 {\footnotesize\begin{align*}
    \mathbb{P}\left(X_{(m-1)}=0\right)\leqslant &\ \mathbb{P}\left(\bigcup_{\substack{1\leqslant \ell,k\leqslant m\\ \ell\neq k}}\left\{X_\ell=X_k=0\right\}\right)\leqslant\sum_{\substack{1\leqslant \ell,k\leqslant n\\ \ell\neq k}}\mathbb{P}\left(X_\ell=X_k=0\right)\leqslant \frac{m(m-1)B^2}{(nA+B)^2}\,.
\end{align*}}
We complete the proof.
\end{proof}

\begin{proof}[Proof of Lemma \ref{lem:step-two}]
First, we consider $\mathbb{P}\left(X_{(1)}\ldots X_{(l)}\leqslant n^{l-1}\left(\log n\right)^2, X_{(m)}\neq 0\right)$. 
    Notice that 
    \[
    X_{(1)}\cdots X_{(l)}\leqslant n^{l-1}\left(\log n\right)^2\Longrightarrow X_1\cdots X_m\leqslant n^{\frac{(l-1)m}{l}}\left(\log n\right)^{\frac{2m}{l}},
    \]
    which follows from the inequality $X_{(1)}\ldots X_{(l)}\geqslant \left(X_1\ldots X_m\right)^{\frac{l}{m}}$. 
    Thus, we have 
    \begin{equation}\label{proof-main-2-step-2-bound-1-1}
        \begin{aligned}
        &\ \mathbb{P}\left(X_{(1)}\cdots X_{(l)}\leqslant n^{l-1}\left(\log n\right)^2, X_{(m)}\neq 0\right)\\
        = &\ \mathbb{P}\left(X_{(1)}\cdots X_{(l)}\leqslant n^{l-1}\left(\log n\right)^2, X_1\cdots X_m>0\right)\\
        \leqslant &\ \mathbb{P}\left(
        X_1\cdots X_m\leqslant n^{\frac{(l-1)m}{l}}\left(\log n\right)^{\frac{2m}{l}}, X_1\cdots X_m>0
        \right)\\
        \leqslant &\ \mathbb{P}\left(
        X_1\cdots X_m\leqslant n^{\frac{(l-1)m}{l}}\left(\log n\right)^{\frac{2m}{l}}\ \Big\vert \ X_1\cdots X_m>0
        \right).
    \end{aligned}
    \end{equation}
   Thus, by Lemma \ref{lem:luce Condorcet winning-bad profile prob}, we have 
    \begin{equation}\label{proof-main-2-step-2-bound-1-2}
    \begin{aligned}
        &\ \mathbb{P}\left(X_1\cdots X_m\leqslant n^{\frac{(l-1)m}{l}}\left(\log n\right)^{\frac{2m}{l}}\ \Big\vert \ X_1\cdots X_m>0\right)\\
        \leqslant &\ \frac{n^{\frac{(l-1)m}{l}}}{n^m}\left(\log n\right)^{\frac{2m}{l}}\left(\log n+1\right)^{m}\left(\frac{2B}{A}\right)^{2m}
        =n^{-\frac{m}{l}}\left(\log n\right)^{\frac{2m}{l}}\left(\log n+1\right)^{m}\left(\frac{2B}{A}\right)^{2m}.
    \end{aligned}
    \end{equation}
    Combing \eqref{proof-main-2-step-2-bound-1-1} with \eqref{proof-main-2-step-2-bound-1-2}, we obtain:
    \begin{equation}\label{proof-main-2-step-2-bound-1}
        \mathbb{P}\left(X_{(1)}\cdots X_{(l)}\leqslant n^{l-1}\left(\log n\right)^2, X_{(m)}\neq 0\right)\leqslant n^{-\frac{m}{l}}\left(\log n\right)^{\frac{2m}{l}}\left(\log n+1\right)^{m}\left(\frac{2B}{A}\right)^{2m}.
    \end{equation}
    Then we focus on $$\mathbb{P}\left(X_{(1)}\cdots X_{(l)}\leqslant n^{l-1}\left(\log n\right)^2, X_{(m)}=0, X_{(m-1)}\neq 0\right).$$
    We have the following decomposition:
    \begin{equation}\label{proof-main-2-step-2-bound-2-decomposition}
        \begin{aligned}
            &\ \mathbb{P}\left(X_{(1)}\cdots X_{(l)}\leqslant n^{l-1}\left(\log n\right)^2, X_{(m)}=0, X_{(m-1)}\neq 0\right)\\
            = &\ \mathbb{P}\left(\bigcup_{j=1}^m\left\{X_{(1)}\cdots X_{(l)}\leqslant n^{l-1}\left(\log n\right)^2,X_j=0,\prod_{i\neq j}X_i>0\right\}\right)\\
            \leqslant &\ \sum_{j=1}^m\mathbb{P}\left(X_{(1)}\cdots X_{(l)}\leqslant n^{l-1}\left(\log n\right)^2,X_j=0,\prod_{i\neq j}X_i>0\right).
        \end{aligned}
    \end{equation}
    Notice that 
    \[\begin{aligned}
    &X_{(1)}\cdots X_{(l)}\leqslant n^{l-1}\left(\log n\right)^2, X_j=0, \prod_{i\neq j}X_i> 0\\\Longrightarrow &\prod_{i\neq j}X_i\leqslant n^{\frac{(l-1)(m-1)}{l}}\left(\log n\right)^{\frac{2(m-1)}{l}}, X_j=0, \prod_{i\neq j}X_i> 0,
    \end{aligned}\]
    which follows from $X_{(1)}\cdots X_{(l)}\geqslant \left(\prod_{i\neq j}X_i\right)^{\frac{l}{m-1}}$ given $X_j=0, \prod_{i\neq j}X_i>0$ and $l\leqslant m-1$. Therefore, using Lemma \ref{lem:luce non-zero-condition}, we have 
    \begin{equation}\label{proof-main-2-step-2-bound-2-1}
        \begin{aligned}
            &\ \mathbb{P}\left(X_{(1)}\cdots X_{(l)}\leqslant n^{l-1}\left(\log n\right)^2,X_j=0,\prod_{i\neq j}X_i>0\right)\\
            \leqslant &\ \mathbb{P}\left(\prod_{i\neq j}X_i\leqslant n^{\frac{(l-1)(m-1)}{l}}\left(\log n\right)^{\frac{2(m-1)}{l}}, X_j=0, \prod_{i\neq j}X_i> 0\right)\\
            = &\ \mathbb{P}\left(\prod_{i\neq j}X_i\leqslant n^{\frac{(l-1)(m-1)}{l}}\left(\log n\right)^{\frac{2(m-1)}{l}}, \prod_{i\neq j}X_i> 0\right)\cdot \mathbb{P}\left(X_j=0\right)\\
            \leqslant &\ \mathbb{P}\left(\prod_{i\neq j}X_i\leqslant n^{\frac{(l-1)(m-1)}{l}}\left(\log n\right)^{\frac{2(m-1)}{l}}\ \bigg\vert \ \prod_{i\neq j}X_i> 0\right)\cdot \frac{B}{nA+B}.
        \end{aligned}
    \end{equation}
    Again by Lemma \ref{lem:luce Condorcet winning-bad profile prob},
    \begin{equation}\label{proof-main-2-step-2-bound-2-2}
    	\begin{aligned}
        &\mathbb{P}\left(\prod_{i\neq j}X_i\leqslant n^{\frac{(l-1)(m-1)}{l}}\left(\log n\right)^{\frac{2(m-1)}{l}}\bigg| \prod_{i\neq j}X_i> 0\right)\\
        \leqslant & n^{-\frac{m-1}{l}}\left(\log n\right)^{\frac{2(m-1)}{l}}\left(\log n+1\right)^{m-1}\left(\frac{2B}{A}\right)^{2m-2}.
        \end{aligned}
    \end{equation}
    Therefore, combining \eqref{proof-main-2-step-2-bound-2-decomposition}, \eqref{proof-main-2-step-2-bound-2-1}, and \eqref{proof-main-2-step-2-bound-2-2}, we have 
    \begin{equation}\label{proof-main-2-step-2-bound-2}
        \begin{aligned}
             &\ \mathbb{P}\left(X_{(1)}\cdots X_{(l)}\leqslant n^{l-1}\left(\log n\right)^2, X_{(m)}=0, X_{(m-1)}\neq 0\right)\\
            \leqslant &\ \sum_{j=1}^m \frac{B}{nA+B}\cdot n^{-\frac{m-1}{l}}\left(\log n\right)^{\frac{2(m-1)}{l}}\left(\log n+1\right)^{m-1}\left(\frac{2B}{A}\right)^{2m-2}\\
            \leqslant &\ m n^{-\frac{m-1}{l}-1}\left(\log n\right)^{\frac{2(m-1)}{l}}\left(\log n+1\right)^{m-1}\left(\frac{2B}{A}\right)^{2m-1}\\
            \leqslant &\ mn^{-\frac{m}{l}}\left(\log n\right)^{\frac{2(m-1)}{l}}\left(\log n+1\right)^{m-1}\left(\frac{2B}{A}\right)^{2m-1},
        \end{aligned}
    \end{equation}
    where the last inequality holds as $l\geqslant 1$. 
    By Lemma \ref{lem:luce non-zero-condition}, we have 
     {\footnotesize\begin{equation}\label{proof-main-2-step-2-bound-3}
        \begin{aligned}
             \mathbb{P}\left(X_{(1)}\cdots X_{(l)}\leqslant n^{l-1}\left(\log n\right)^2, X_{(m-1)}= 0\right)
             \leqslant &\ \mathbb{P}\left(X_{(m-1)}= 0\right)\leqslant \frac{m(m-1)B^2}{(nA+B)^2}.
        \end{aligned}
    \end{equation}}
    Finally, combining \eqref{proof-main-2-step-2-bound-1}, \eqref{proof-main-2-step-2-bound-2}, and \eqref{proof-main-2-step-2-bound-3}, we have 
    \begin{equation*}
        \begin{aligned}
        &\ \mathbb{P}\left(X_{(1)}\cdots X_{(l)}\leqslant n^{l-1}\left(\log n\right)^2\right)\\
        = &\ \mathbb{P}\left(X_{(1)}\cdots X_{(l)}\leqslant n^{l-1}\left(\log n\right)^2, X_{(m)}\neq 0\right) \\
        &\  + \mathbb{P}\left(X_{(1)}\cdots X_{(l)}\leqslant n^{l-1}\left(\log n\right)^2, X_{(m)}= 0, X_{(m-1)}\neq 0\right)\\
        &\ + \mathbb{P}\left(X_{(1)}\cdots X_{(l)}\leqslant n^{l-1}\left(\log n\right)^2, X_{(m-1)}= 0\right)\\
        \leqslant &\ n^{-\frac{m}{l}}\left(\log n\right)^{\frac{2m}{l}}\left(\log n+1\right)^{m}\left(\frac{2B}{A}\right)^{2m}\\
        &\ +  mn^{-\frac{m}{l}}\left(\log n\right)^{\frac{2(m-1)}{l}}\left(\log n+1\right)^{m-1}\left(\frac{2B}{A}\right)^{2m-1} + \frac{m(m-1)B^2}{(nA+B)^2}\\
        \leqslant &\ (m+1)n^{-\frac{m}{l}}\left(\log n\right)^{\frac{2m}{l}}\left(\log n+1\right)^{m} \left(\frac{2B}{A}\right)^{2m}+ \frac{m(m-1)B^2}{A^2n^2},
    \end{aligned}
    \end{equation*}
    which holds for any $1\leqslant l\leqslant m-1$.

\end{proof}

\subsection{Proof of Lemma~\ref{lem:generating function}}

\begin{proof}[Proof of Lemma~\ref{lem:generating function}]

Based on the Binomial Theorem, we have
\[
\begin{aligned}
    (1-x)^n=&\ 1- \sum_{k=1}^n (-1)^{k-1}{n\choose k}x^k,\\
    (1+x)^n=&\ 1 + \sum_{k=1}^n {n\choose k}x^k.
\end{aligned}
\]
Thus, we obtain the following equality: 
\begin{equation*}
        \sum_{k=0}^{\lfloor{\frac{n-1}{2}}\rfloor}{n\choose 2k+1} x^{2k+1} =\frac{(1+x)^n-(1-x)^n}{2} .    
    \end{equation*}
Therefore, by letting $x=\frac{1}{2n}$, we have
\begin{equation*}
    \sum_{k=0}^{\lfloor{\frac{n-1}{2}}\rfloor}{n\choose 2k+1}\left(\frac{1}{2n}\right)^{2k+1}=\frac{(1+\frac{1}{2n})^n-(1-\frac{1}{2n})^n}{2} :=f(n).
\end{equation*}
Notice that for all $n\geqslant 1$, we have $1>f(n)>0$. Furthermore, $\lim_{n\to\infty}f(n)=\frac{e^{\frac{1}{2}}-e^{-\frac{1}{2}}}{2}<1$. This shows that there exists a universal constant $c$ with $1>c>0$ such that for all $n\geqslant 1$, $f(n)\leqslant 1-c$. 
\end{proof}

\section{Simulation Study}
\label{append:simulations}
We conduct simulations to validate our theoretical tight bound on the probability of the existence of a Condorcet winning response. Specifically, we consider settings where $m \in \{3,4,5,6\}$, with $n \in [100,300]$ for $m = 3,4$ and $n \in [200,400]$ for $m = 5,6$. Under the Luce model, we set the parameters as
\[
\lambda_i = \exp\left( \lambda_p \cdot \mathds{1}\{i=1\} \right), \quad i \in \{1,\cdots,n\}\,,
\]
where $\lambda_p \in \{0, 0.005, 0.05, 0.1, 0.2, 0.3, 0.4, 0.5\}$. When $\lambda_p = 0$, this setup exactly recovers the Impartial Culture model. For each configuration, we estimate the probability that $y_1$ is the Condorcet winning response using $1,000,000$ simulations. The results, shown in Figure~\ref{fig:simulation-probability}, align perfectly with our theoretical tight bound $\Theta(n^{-m/\lceil m/2\rceil})$.

\begin{figure}[htbp]
    \centering
    \includegraphics[width=0.96\linewidth]{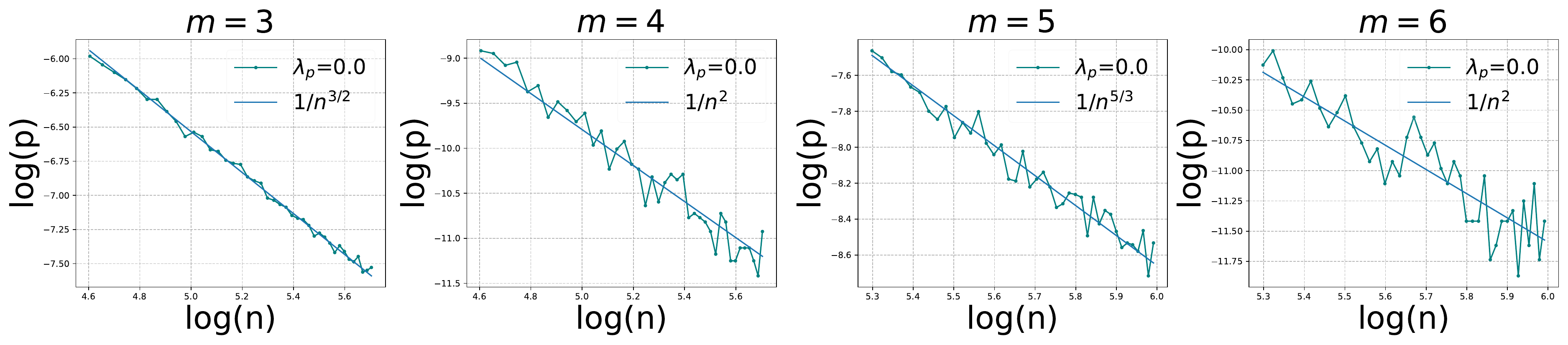}
    \includegraphics[width=0.96\linewidth]{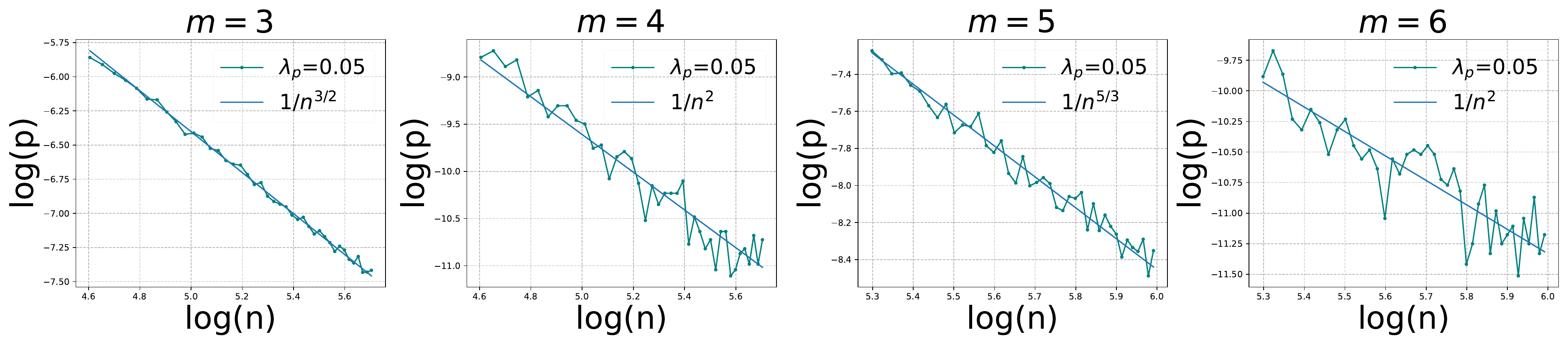}
    \includegraphics[width=0.96\linewidth]{figs/lambda=0.1.pdf}
    \includegraphics[width=0.96\linewidth]{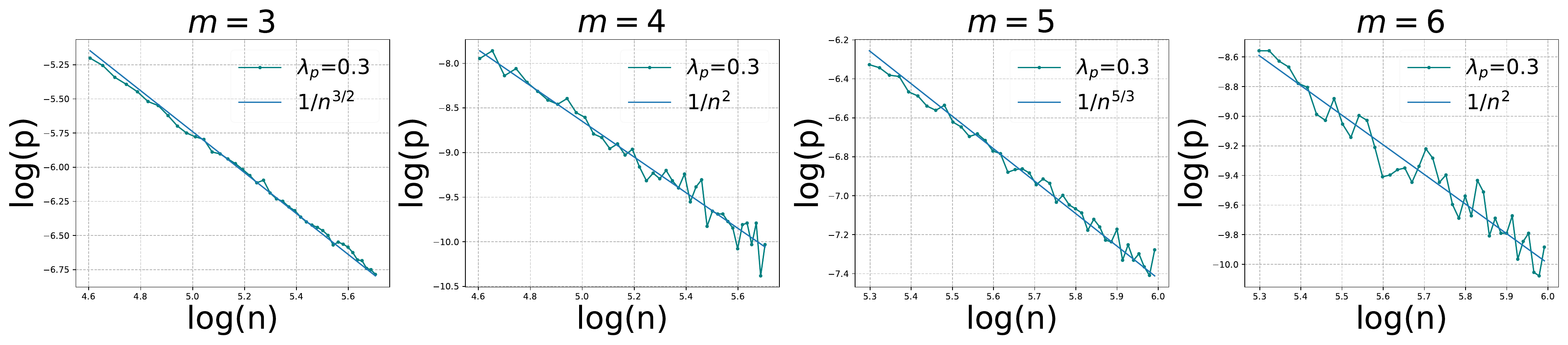}
    \includegraphics[width=0.96\linewidth]{figs/lambda=0.5.pdf}
    \caption{Simulation results for the probability that $y_1$ is the Condorcet winning response under the Luce model with varying $\lambda_p$, $m$, and $n$. The empirical estimates align closely with our theoretical tight bound $\tilde{\Theta}(n^{-\frac{m}{l}})$.}
    \label{fig:simulation-probability}
\end{figure}
\end{document}